\documentclass[11pt]{aart}

\usepackage{titlesec}
\titleformat{\section}[block]{\filcenter\normalfont\bfseries\large}{\thesection.}{.5em}{}\titlespacing*{\section}{0pt}{2\baselineskip}{1\baselineskip}
\titleformat{\subsection}[runin]{\normalfont\bfseries}{\thesubsection.}{.4em}{}[.]\titlespacing{\subsection}{0pt}{2ex plus .1ex minus .2ex}{.8em}
\titleformat{\subsubsection}[runin]{\normalfont\itshape}{\thesubsubsection.}{.3em}{}[.]\titlespacing{\subsubsection}{0pt}{1ex plus .1ex minus .2ex}{.5em}
\titleformat{\paragraph}[runin]{\normalfont\itshape}{\theparagraph.}{.3em}{}[.]\titlespacing{\paragraph}{0pt}{1ex plus .1ex minus .2ex}{.5em}

\usepackage[T1]{fontenc}
\usepackage[utf8]{inputenc}

\usepackage{mlmodern}

\usepackage[dvipsnames,table,xcdraw]{xcolor}
\usepackage[english]{babel}
\usepackage[letterpaper, hmargin=1in, top=1in, bottom=1.2in, footskip=0.6in]{geometry}
\usepackage{enumerate}  
\usepackage{amsmath}
\usepackage{dsfont}
\usepackage{graphicx}
\usepackage{color}
\usepackage{amsthm}
\usepackage{amssymb}
\usepackage{microtype}

\usepackage{bbm}

\usepackage{csquotes}

\usepackage{mathtools}

\definecolor{vdarkred}{rgb}{0.4,0,0.15}
\definecolor{vdarkblue}{rgb}{0,0.15,0.4}
\usepackage[pdftex, colorlinks, linkcolor=vdarkblue,citecolor=vdarkred]{hyperref}

\usepackage[capitalise,nameinlink,noabbrev]{cleveref}

\usepackage{enumitem} 

\DeclareMathOperator*{\esssup}{ess\,sup}

\renewcommand{\le}{\leqslant}
\renewcommand{\leq}{\leqslant}

\renewcommand{\geq}{\geqslant}

\newcommand{\dd}{\mathrm d}
\newcommand{\ee}{\mathrm e}
\newcommand{\ii}{\mathrm i}

\let\a=\alpha          

        \let\k=\kappa     
                          \let\r=\rho
 \let\t=\tau         \let\ph=\varphi

\let\P=\Pi

\renewcommand{\r}{\mathrm} 
\renewcommand{\cal}{\mathcal} 
 
\newcommand{\fra}{\mathfrak} 
\newcommand{\wh}{\widehat}

\renewcommand{\P}{\mathbb{P}}
\newcommand{\E}{\mathbb{E}}
\newcommand{\R}{\mathbb{R}}
\newcommand{\C}{\mathbb{C}}
\newcommand{\N}{\mathbb{N}}

\def\bR{\mathbb{R}}

\def\cF{\mathcal{F}}

\def\eps{\varepsilon}
\def\ph{\varphi}

\def\bN{\mathbb{N}}  
 
\def\bR{\mathbb{R}} 
 
\def\bE{\mathbb{E}}

\renewcommand{\d}[1]{\dd #1}

\newcommand{\WU}{{\mathbf{W}}}
\newcommand{\PU}{{\mathfrak{P}}}

\newcommand\restr[2]{{
  \left.\kern-\nulldelimiterspace 
  #1 
  \vphantom{\big|} 
  \right|_{#2} 
  }}

\DeclareMathOperator{\tr}{tr}
\DeclareMathOperator{\Tr}{Tr}

\DeclareMathOperator{\re}{Re}

\newcommand{\col}{\vcentcolon}
\newcommand{\wick}[1]{{\col\!#1\!\col}}
\newcommand*{\deq}{\mathrel{\vcenter{\baselineskip0.65ex \lineskiplimit0pt \hbox{.}\hbox{.}}}=}
\newcommand*{\eqd}{=\mathrel{\vcenter{\baselineskip0.65ex \lineskiplimit0pt \hbox{.}\hbox{.}}}}

\DeclareMathAlphabet{\mathbbold}{U}{bbold}{m}{n}

\renewcommand{\epsilon}{\varepsilon}

\theoremstyle{plain} 
\newtheorem{theorem}{Theorem}[section]
\newtheorem*{theorem*}{Theorem}
\newtheorem{lemma}[theorem]{Lemma}
\newtheorem*{lemma*}{Lemma}
\newtheorem{cor}[theorem]{Corollary}
\newtheorem{prop}[theorem]{Proposition}

\newtheorem*{conjecture*}{Conjecture}

\theoremstyle{definition} 
\newtheorem{definition}[theorem]{Definition}
\newtheorem*{definition*}{Definition}
\newtheorem{example}[theorem]{Example}
\newtheorem*{example*}{Example}
\newtheorem{remark}[theorem]{Remark}
\newtheorem*{remark*}{Remark}
\newtheorem{assumption}[theorem]{Assumption}
\newtheorem*{assumption*}{Assumption}

\numberwithin{equation}{section}
\numberwithin{equation}{section}
\usepackage{etoolbox}

\usepackage{ragged2e}

\newcommand{\floor}[1] {\lfloor #1 \rfloor}

\newcommand{\pb}[1]{\bigl(#1\bigr)}

\newcommand{\pbb}[1]{\biggl(#1\biggr)}

\newcommand{\qb}[1]{\bigl[#1\bigr]}
\newcommand{\qB}[1]{\Bigl[#1\Bigr]}
\newcommand{\qbb}[1]{\biggl[#1\biggr]}

\newcommand{\hB}[1]{\Bigl\{#1\Bigr\}}

\newcommand{\abs}[1]{\lvert #1 \rvert}

\newcommand{\absB}[1]{\Bigl\lvert #1 \Bigr\rvert}
\newcommand{\absbb}[1]{\biggl\lvert #1 \biggr\rvert}

\newcommand{\norm}[1]{\lVert #1 \rVert}

\newcommand{\normbb}[1]{\biggl\lVert #1 \biggr\rVert}

\newcommand{\scalar}[2]{\langle#1 \mspace{2mu}, #2\rangle}

\newcommand{\assP}{\hyperlink{P}{\normalfont \textbf{(P)}}}
\newcommand{\assD}{\hyperlink{D}{\normalfont \textbf{(D)}}}

\title{The Euclidean \(\phi^4_2\) theory as a limit of an inhomogeneous Bose gas}
\author{Cristina Caraci \and
	Antti Knowles \and  Alessio Ranallo \and  Pedro Torres Giesteira}

\date{\today}

\begin{document}
	
	\maketitle

\begin{abstract}

We prove that the grand canonical Gibbs state of an interacting two-dimensional quantum Bose gas confined by a trapping potential converges to the complex Euclidean field theory with local quartic self-interaction, when the density of the gas becomes large and the range of the interaction becomes small. We obtain convergence of the relative partition function and convergence in $L^1 \cap L^\infty$ of the renormalised reduced density matrices. The field theory is supported on distributions of negative regularity, which requires a renormalisation by divergent mass and energy counterterms. Unlike previous results in the homogeneous setting of the torus without a trapping potential, the counterterms are not given by a finite collection of scalars but by diverging counterterm functions. This leads to significant new mathematical challenges. For our proof, we also derive quantitative bounds on the Green function of Schrödinger operators and of its gradient, which might be of independent interest.
\end{abstract}
\section{Introduction}

A Euclidean field theory of a scalar complex field on a domain $\Lambda \subset \bR^d$ is specified by a formal probability measure on a space of fields\footnote{Rigorously, the space of fields is the space of complex-valued distributions over $\Lambda$.} $\phi\col \Lambda \to \C$ of the form
\begin{equation} \label{EFT}
\rho(\dd \phi) = \frac1c e^{-S(\phi)} \, \r D\phi\,,
\end{equation}
where $\r D \phi = \prod_{x \in \Lambda} \d \phi(x)$ is the formal uniform measure on the space of fields, and $S$ is the action.
One of the simplest field theories with nontrivial interaction is the \emph{Euclidean $\phi^4_d$ theory}, whose action is given by
\begin{equation}\label{eq:action}
S(\phi) \deq \int_{\Lambda} \d x\, \bar \phi(x) \pbb{- \frac{\Delta}{2} - \theta +U(x)} \phi(x) + \frac{\lambda}{2} \int_{\Lambda} \d x\, |\phi(x)|^4\,, 
\end{equation}
where $\theta \in \R$ is a constant, $\lambda > 0$ a coupling constant, $\Delta$ is the Laplacian on $\Lambda$, and $U \col \Lambda \to \R$ is an external potential.

Euclidean field theories originally arose in high-energy physics in $d = 4$ space-time dimensions \cite{schwinger1958euclidean, nakano1959quantum}. Subsequently, Euclidean field theories have proven of great importance in statistical mechanics in $d \leq 3$ dimensions. The works \cite{symanzik1966euclidean, symanzik1969euclidean} recognised the analogy between Euclidean field theories and classical statistical mechanics, which was followed by a purely probabilistic formulation of Euclidean field theories in \cite{nelson1973free, nelson1973probability}. The rigorous study of field theories of the form \eqref{EFT} has been a major topic in mathematical physics since the late sixties; see e.g.\ \cite{glimm2012quantum, Simon74, Hairer2016} for reviews. The basic philosophy in all of these works is to include the quadratic part of \eqref{eq:action} in a free, Gaussian, measure
\begin{equation*}
\P(\dd \phi) \deq \frac{1}{c_0} \ee^{-\int_{\Lambda} \d x\, \bar \phi(x) (- \frac{\Delta}{2} - \theta +U(x)) \phi(x)}\, \r D \phi\,,
\end{equation*}
which can be easily made sense of. The quartic part of \eqref{eq:action},
\begin{equation*}
V(\phi) \deq \frac{\lambda}{2} \int_{\Lambda} \d x\, |\phi(x)|^4\,,
\end{equation*}
is regarded as an interaction that allows one to write \eqref{EFT} in the form
\begin{equation} \label{}
\rho(\dd \phi) = \frac{1}{\zeta} \, \ee^{-V(\phi)} \, \P(\dd \phi)\,, \qquad \zeta \deq \frac{c}{c_0}\,.
\end{equation}

The fundamental difficulty in such an undertaking arises from ultraviolet divergencies: in dimensions greater than one, the field $\phi$ is $\P$-almost surely a distribution of negative regularity, and hence the interaction $V(\phi)$ is ill-defined. As a consequence, to make rigorous sense of $\rho$, the interaction has to be renormalised by subtracting suitably chosen infinite counterterms. 

In this paper we identify the Euclidean field theory of the form \eqref{EFT} with action \eqref{eq:action} in $d = 2$ dimensions as a high density limit of an interacting quantum Bose gas.
We recall that a quantum system of $n$ spinless non-relativistic bosons in $\Lambda$ is described by the Hamiltonian 
\[
\mathbb{H}_n \deq \sum_{i= 1}^n \biggl(-\frac{\Delta_i}{2m} +\mathcal{U}(x_i)\biggr)   +\frac g2 \sum_{i,j=1}^n v(x_i-x_j)
\]
acting on the subspace $L^2(\Lambda)^{\otimes_s n}$ of $L^2(\Lambda^n)$ consisting of symmetric functions with respect to permutations. In the Hamiltonian, $\Delta_i$ is the Laplacian acting in the variable $x_i$, $g$ is a coupling constant, $\mathcal U$ is a confining external potential, and $v$ is a repulsive two-body interaction potential. We consider the system in the grand canonical ensemble at positive temperature, defined by the density matrix
\begin{equation}
    \label{Gibbsstate}
\frac{1}{Z} \bigoplus_{n \in \bN} \ee^{-\beta(\mathbb{H}_n - \theta n)}
\end{equation}
acting on the Fock space $\cF =\bigoplus_{n \in \bN} L^2(\Lambda)^{\otimes_s n}$, where $\beta <\infty$ is the inverse temperature, $\theta$ is the chemical potential, and $Z$ is a normalisation factor. 

For the limiting regime that we are considering, we introduce two parameters $\nu, \eps >0$, where $\nu = \frac \beta m = \sqrt{\beta g}$, and the interaction potential $v$ is taken to be an approximate delta function of range $\eps$. 
We show that as $\nu, \epsilon \to 0$ there exists a suitable renormalisation of the chemical potential $\theta \equiv \mathfrak h^\eps_\nu$ such
that the reduced density matrices of the quantum state \eqref{Gibbsstate} converge to the correlation functions
of the field theory \eqref{EFT}, \eqref{eq:action}.

Previously, such a connection between \eqref{EFT} and \eqref{Gibbsstate} was obtained for $d = 1$ in \cite{lewin2015derivation, lewin2018gibbs, frohlich2017gibbs}, where the short-range singularities of the field do not appear and hence no renormalisation is required. The first result of this kind in dimension larger than one was \cite{FKSS_24}, where it was established in $d = 2$ dimensions. This result was recently extended to $d = 3$ dimensions in \cite{nam2025phi}. Both of these results were established in the \emph{homogeneous}, or translation invariant, setting, where $\Lambda$ is the $d$-dimensional unit torus and $U = \cal U = 0$.

The goal of this paper is to establish the connection between \eqref{EFT} and \eqref{Gibbsstate} in the \emph{inhomogeneous} setting, where the assumption of translation invariance is dropped. Doing so is of some importance for the physical relevance of such a result. Indeed, the assumption of translation invariance does not hold in a typical experimental setup, where the Bose gas is either confined by an external trapping potential or constrained to a bounded domain with Dirichlet boundary conditions.

Forgoing translation invariance leads to significant new mathematical challenges as compared to homogeneous case. The most fundamental one is that the infinite counterterms subtracted from $V(\phi)$ no longer consist of a finite collection of scalar quantities, but instead of diverging functions with nontrivial spatial dependence. It turns out that a key cancellation between such counterterms, which can be trivially seen to be exact in the homogeneous setting, fails in the inhomogeneous setting. In fact, we show that generically such a cancellation \emph{cannot} hold as soon as the system is inhomogeneous; see \eqref{eq:eqAmirali} below. As a consequence, in the inhomogeneous setting one has to handle a new family of diverging \emph{counterterm functions}, on both the level of the field theory \eqref{EFT} and the quantum many-body theory \eqref{Gibbsstate}. We refer to Section \ref{sec:strategy} below for more details on the strategy of the proof.

As an additional ingredient of our proof, we derive precise quantitative bounds on the Green function of general Schrödinger operators on $\R^d$, of its regularised version, and of its gradient. We believe that these results might be of independent interest.

We conclude this section with a summary of some related previous work. In dimensions $d \leq 3$, the \emph{mean-field limit} was investigated in \cite{lewin2018classical, LNR3, FKSS_2020, frohlich2017gibbs, frohlich2019microscopic,sohinger2019microscopic}, where the parameter $\epsilon$ was fixed as $\nu \to 0$. The resulting limiting field theory differs from $\phi^4_d$ in that the interaction term $V(\phi)$ is nonlocal, given by a convolution with a bounded two-body interaction potential $v$. This nonlocal interaction term is considerably less singular than the local one of $\phi^4_d$ theory. The stronger singularity of $V(\phi)$ requires additional renormalisation as compared to the nonlocal potential. The fundamental complications mentioned above, arising from going from the homogeneous to the inhomogeneous setting, only appear in this additional renormalisation. Accordingly, the mean-field limit was treated also in the inhomogeneous setting in \cite{FKSS_2020, LNR3} with minor adaptations of the proof in the homogeneous setting.

\subsection*{Acknowledgements} We are very grateful to Amirali Hannani for insightful discussions and suggestions, as well as a proof of non-solvability of the equation \eqref{eq:eqAmirali} given in Appendix \ref{appendix:counter example} below. We also thank Phan Thành Nam and Nicolas Rougerie for helpful discussions. The authors  acknowledge funding from the European Research Council (ERC) and the Swiss State Secretariat for Education, Research and Innovation (SERI) through the consolidator grant ProbQuant, as well as funding from the Swiss National Science Foundation through the NCCR SwissMAP grant.

\section{Setup and main results}

In this section we give the precise setup and state our main results.

\subsection{The one-particle Hamiltonian}
On the one-particle space \(\mathcal{H}\deq L^2(\mathbb{R}^d, \mathbb{C})\) we define the one-particle Hamiltonian
\begin{equation}\label{def:one body hamiltonian}
    h\deq \kappa -\frac{\Delta}{2}  + U\,,
\end{equation}
where \(\kappa>0\), $\Delta$ is the Laplacian on $\R^d$, and \(U\) is an external potential satisfying the following assumption.

\begin{assumption} \label{assumption:P and D for U}
    Let $d \geq 1$ and \(\theta  > 0\) . The function \(U\col \mathbb{R}^d \to (0, \infty) \) satisfies one of the two following conditions. 
    \begin{itemize}
        \item[\hypertarget{P}{\textbf{(P)}}] $U 
    \in L^\infty_{\operatorname{loc}}(\mathbb R^d)$ and there exist constants \( C,c>0\) such that, for almost all \(x \in {\mathbb{R}^d}\), 
\begin{equation} \label{U_bounds}
c\bigl(1 + \abs{x}^\theta\bigr)\leq U(x) \leq  C\bigl(1 + \abs{x}^\theta\bigr)\,.
\end{equation}

        \item[\hypertarget{D}{\textbf{(D)}}] \(U \in W^{1, \infty}_{\operatorname{loc}}(\mathbb{R}^d)\) and there exist a radial non-decreasing function \(g\col \mathbb{R}^d \to [1, \infty)\) and constants \(C,c>0\), \(0<\gamma< 1\)  such that, for almost all \(x \in \mathbb{R}^d\), 
        \[U(x)\geq cg(x)\geq c\bigl(1 + \abs{x}^\theta\bigr)\,,\]
        \[\abs{\nabla U (x)} \leq C g(\gamma x)^{3/2} \,.\]
    \end{itemize}
\end{assumption}

\begin{example}
The following examples illustrate Assumption \ref{assumption:P and D for U}.
\begin{enumerate}[label=(\roman*)]
\item
The step potential $U(x) \deq 1 + \floor{\abs{x}}^\theta$ satisfies Assumption \ref{assumption:P and D for U} \assP{}, where $\floor{\cdot}$ denotes the integer part.
\item
For any $C > 0$, the rapidly growing potential $U(x) = \exp\pb{(1 + \abs{x}^2)^C}$ satisfies Assumption \ref{assumption:P and D for U} \assD{}.
\end{enumerate}
\end{example}

\begin{remark} \label{rmk:thetha_s2}
Let $d = 2,3$ and \(s\in (d/2,2)\). Throughout this paper we shall frequently require that
\begin{equation} \label{eq:tracehs}
\tr{h^{-s}} < \infty\,,
\end{equation}
which, by the Lieb-Thierring type inequality (see e.g.\ \cite[Theorem~1]{liebthirring}), holds whenever $\int \d x \, \bigl(\kappa + U(x)\bigr)^{d/2-s}<\infty$. For $U(x) \geq c (1 + \abs{x}^\theta)$ as in Assumption \ref{assumption:P and D for U}, this condition holds whenever $\theta>\frac{d}{s-d/2}$.
\end{remark}
\subsection{Classical field theory}
In this subsection we define the field theory and its correlation functions.
The classical free field $\phi$ is by definition the complex-valued Gaussian field with mean zero and covariance $G \deq h^{-1}$. Explicitly, the free field may be constructed as follows.
We use the spectral decomposition\footnote{Note that under the condition $U(x) \geq c \abs{x}^\theta$ from Assumption \ref{assumption:P and D for U}, $h$ has a compact resolvent.} \(h=\sum_{k =0}^\infty \lambda_k u_k u_k^* \), with eigenvalues $\lambda_k > 0$ and normalised eigenfunctions $u_k \in \cal H$. Let $X = (X_k)_{k \in \N}$ be a family of independent standard complex Gaussian random variables\footnote{The random variable \(X\) is standard complex Gaussian random variable if \(\mathbb{E}X = 0\), \(\mathbb{E} X^2 = 0 \) and \(\mathbb{E}\abs{X}^2 = 1 \).}, whose law and associated expectation are denoted by $\P$ and $\E$, respectively. The \emph{classical free field} is then given by
\begin{equation*}
\phi \deq \sum_{k \in \N} \frac{X_k}{\sqrt{\lambda_k}} \, u_k\,.
\end{equation*}

In order to define the interacting theory, we need to regularise the field \(\phi\), which is well known to be almost surely a distribution of negative regularity. To that end, we choose a nonnegative rapidly decaying function \(\vartheta\col[0, \infty)\to [0, \infty)\) satisfying \(\vartheta(0)=1\), and for \(N>0\) define the regularised field
\begin{equation}\label{eq:vartheta}
    \phi_N \deq \sum_{k=0}^{\infty} X_k \sqrt{
\frac{\vartheta( \lambda_k/N)}{\lambda_k}}u_k\,.
\end{equation}

The \emph{Green function} of the free field \(\phi\) is the integral kernel of $G = h^{-1}$,
\begin{equation}\label{eq:defG}
	G(x,y)=\mathbb{E}\bigl[\phi(x)\bar{\phi}(y)\bigr]=\frac{1}{h}(x,y)\,.\end{equation}
Here, and throughout the following, we use the notation $A(x,y)$ for the kernel of an operator $A \col \mathcal H \to \mathcal H$. Similarly, we define \(G_N\) as the Green function of regularised field \(\phi_N\) 
 \begin{equation}
     \label{eq:defGN}
     G_N(x,y) = \mathbb E \bigl[\phi_N(x)\overline\phi_N(y)\bigr] = \frac{\vartheta(h / N)}{h}(x,y)\,.
 \end{equation}
 
Next, we suppose that $d = 2$ and define the regularised interaction 
\begin{equation*}\label{eq:def:regularised interaction V}
    V_{N}\deq\frac{1}{2}\int \d x\, \wick{\abs{\phi_N(x)}^4}\,,
\end{equation*}
where \(\wick{\cdot}\) denotes Wick ordering (see e.g.\ \cite[Appendix A]{FKSS_24}) with respect to the measure \(\mathbb{P}\). Explicitly,  
\begin{equation*}
     \wick{\abs{\phi_N(x)}^4} = \abs{\phi_N(x)}^4 - 4G_N(x,x)\abs{\phi_N(x)}^2+ 2 (G_N(x,x))^2\,.
\end{equation*}
An application of Wick's theorem combined with the decay estimate for $G_N$ from Proposition \ref{prop:GN} shows that, provided the lower bound of \eqref{U_bounds} holds with $\theta > 2$, the regularised interaction $V_N$ is uniformly bounded in $L^2(\P)$. A similar argument shows that \(V_{N}\) converges in \(L^2(\mathbb{P)}\). We denote its limit by $V$, which can be suggestively written as
\begin{equation}
\label{eq:defV}
V = \frac{1}{2}\int \d x \,\wick{\abs{\phi(x)}^4}\,.
\end{equation}
The interacting complex $\phi^4_2$ field theory is given by the probability measure 
\begin{equation}\label{eq:interacting field theory}
    \frac{1}{\zeta}\ee^{-V} \dd{ \mathbb{P}}, \qquad \zeta\deq\mathbb{E}\bigl[\ee^{-V}\bigr]\,.
\end{equation}
\begin{remark}
The existence of the measure \eqref{eq:interacting field theory} is nontrivial. It amounts to the integrability condition $\zeta < \infty$. For the homogeneous $\phi^4_2$ theory on the torus, this was first established in the landmark work of Nelson \cite{nelson1973free, nelson1973probability}. Nelson's argument can be extended to the inhomogeneous case provided that the lower bound of \eqref{U_bounds} holds for large enough $\theta$. In fact, in Section \ref{sec:field} we prove a more general result, which establishes the uniform integrability for a family of nonlocal field theories, under the condition $\theta > 10$. See also Remark \ref{rem:theta_optimal} below.
\end{remark}

The interacting field theory \eqref{eq:interacting field theory} is characterised by its correlation functions. For \(p\in \mathbb{N}\) and \(\mathbf{x, \tilde{x}}\in \mathbb{R}^{dp} \), we define the \emph{\(p\)-point correlation function} as 
\begin{equation}\label{eq:corrfunction}
    (\gamma_p)_{\mathbf{x, \tilde{x}}} \deq \frac{1}{\zeta}\mathbb{E}\Bigl[\bar{\phi}(\tilde{x}_1)\cdots\bar{\phi}(\tilde{x}_p){\phi}({x}_1)\cdots{\phi}({x}_p) \, \ee^{-V} \Bigr]\,,
\end{equation}
which is the $2p$-th moment of the field $\phi$ under the probability measure \eqref{eq:interacting field theory}. This measure is sub-Gaussian, and is hence determined by its moments $(\gamma_p)_{p \in \N^*}$. (Note that any moment containing a different number of $\bar \phi$s and $\phi$s vanishes by invariance of the measure \eqref{eq:interacting field theory} under the gauge transformation $\phi \mapsto \alpha \phi$, where $\abs{\alpha} = 1$.)

As explained in \cite[Section 1.5]{FKSS_2020}, the correlation function $\gamma_p$ is divergent on the diagonal, even for the free field. Hence, for instance, it cannot be used to analyse the distribution of the mass density $\abs{\phi(x)}^2$. As in \cite[Section 1.5]{FKSS_2020}, we remedy this issue by introducing the \emph{Wick-ordered $p$-point correlation function}
\begin{equation}\label{eq:wickcorrfunction}
    (\widehat \gamma_p)_{\mathbf{x, \tilde{x}}} \deq \frac{1}{\zeta}\mathbb{E}\Bigl[\wick{\bar{\phi}(\tilde{x}_1)\cdots\bar{\phi}(\tilde{x}_p){\phi}({x}_1)\cdots{\phi}({x}_p)} \,\ee^{-V} \Bigr]\,,
\end{equation}
which has a regular behaviour on the diagonal. The Wick-ordered correlation function \eqref{eq:wickcorrfunction} can be expressed explicitly in terms of the correlation functions \eqref{eq:corrfunction} and the correlation functions of the free field; see \eqref{gamma_p_link} below.

\subsection{Quantum many-body system}
In this subsection we define the quantum many-body system and its reduced density matrices. For $n \in \N$, we denote by $\cal P_n$ the orthogonal projection onto the symmetric subspace of $\cal H^{\otimes n}$; explicitly, for $\Psi_n \in \cal H^{\otimes n}$,
\begin{equation*}
\cal P_n \Psi_n(x_1, \dots, x_n) \deq \frac{1}{n!} \sum_{\pi \in S_n} \Psi_n(x_{\pi(1)}, \dots, x_{\pi(n)})\,,
\end{equation*}
where $S_n$ is the group of permutations on $\{1, \dots, n\}$. For $n \in \N^*$, we define the $n$-particle space as $\cal H_n \deq \cal P_n \cal H^{\otimes n}$. We define Fock space as the Hilbert space $\cal F \equiv \cal F (\cal H) \deq \bigoplus_{n \in \N} \cal H_n$. We denote by $\Tr_{\cal F}(X)$ the trace of an operator $X$ acting on $\cal F$. For $f \in \cal H$ we define the bosonic annihilation and creation operators $a(f)$ and $a^*(f)$ on $\cal F$ through their action on a dense set of vectors $\Psi = (\Psi_n)_{n \in \N} \in \mathcal{F}$ as
\begin{align*}
\pb{a(f) \Psi}_n(x_1, \dots, x_n) &= \sqrt{n+1} \int \dd x \, \bar f(x) \, \Psi_{n+1} (x,x_1, \dots, x_n)\,,
\\
\pb{a^*(f) \Psi}_n(x_1, \dots, x_n) &= \frac{1}{\sqrt{n}} \sum_{i = 1}^n f(x_i) \Psi_{n - 1}(x_1, \dots, x_{i - 1}, x_{i+1}, \dots, x_n)
\,.
\end{align*}
We regard $a$ and $a^*$ as operator-valued distributions and use the customary notations
\[
a(f) = \int \d x\, \bar{f}(x)a(x), \qquad \qquad a^*(f) = \int \d x\, f(x)a^*(x)\,.
\]

The interaction potential of the Bose gas satisfies the following assumption.
\begin{assumption}\label{assumptions: non local interaction}
    Let \(v : \mathbb{R}^d \to \mathbb{R}\) be an even, smooth compactly supported function of positive type\footnote{That is, the Fourier transform of $v$ is a positive measure.}, whose integral is equal to one. 
\end{assumption}
For \(\eps>0\) we define the rescaled interaction potential
\begin{equation}\label{def: rescaled interaction potential}
    v^\eps(x)\deq\frac{1}{\eps^d}v\biggl(\frac{x}{\eps}\biggr)\, . 
\end{equation}
For  \(\eps, \nu >0\), we define the interacting quantum Hamiltonian  $H \equiv H_\nu^\eps$
\begin{multline}\label{eq:Hwrenormalization}
    H \deq \nu \int \d x \,\d {\tilde{x}}\, a^*(x)\Bigl(\kappa-\frac{\Delta}{2} + \mathcal{U} \Bigr)(x, \tilde{x})a(\tilde{x}) \\+ \frac{\nu^2}{2}  \int \d{x} \,  \d{\tilde{x}} \,a^*(x)a(x)  \, v^\varepsilon(x-\tilde{x})\, a^*(\tilde{x})a(\tilde{x}) 
    - \nu\mathfrak{g}^\eps_\nu \int \d x\, a^*(x)a(x) + \mathfrak{h}^\eps_\nu\,,
\end{multline}
a self-adjoint operator on $\cal F$. Here, $\mathfrak{g}^\eps_\nu$ and $\mathfrak{h}^\eps_\nu$ are real renormalisation parameters chosen in \eqref{eq:choice of chemical, energy shift phi42} below, so that the density matrix \eqref{Gibbsstate} can be expressed as 
\[
\frac{1}{Z}\bigoplus_{n \in \bN} \ee^{-\beta(\mathbb{H}_n - \theta n)}=\frac{\ee^{-H}}{Z}\,, \quad \quad Z \deq \Tr_{\cF}(\ee^{-H})\,,
\]
where $Z$ is the grand canonical partition function. The parameters $\mathfrak{g}^\eps_\nu$ and $\mathfrak{h}^\eps_\nu$ describe a renormalisation of the chemical potential and of the energy, respectively.

For \(\nu>0\), we define the free Hamiltonian\footnote{Note that $H^0$ has external potential $U$ (recall \eqref{def:one body hamiltonian}) while $H$ has external potential $\cal U$.}
\begin{equation}\label{def: free Hamiltonian acting on the fock space}
    H^0 \deq\nu \int \d x \, \d {\tilde{x}}\, h(x,\tilde{x}) \,a^*(x) a(\tilde{x})\,.
\end{equation}
The free grand canonical partition function and the relative partition functions are respectively defined as
\begin{equation}\label{eq:mf quantum relatve partition function}
Z^0 \deq \Tr_{\cF}(\ee^{-H^0}) \,, \qquad    \mathcal{Z} \deq \frac Z{Z^0}\,.
\end{equation}
In order to determine the renormalisation parameters $\mathfrak{g}_\nu^\eps$ and $\mathfrak{h}_\nu^\eps$, we introduce the the rescaled mean particle density in the free quantum state,
\begin{equation}\label{def quantum green function diagonal}
    \varrho_\nu(x)\equiv\varrho_\nu^U({x})\deq \nu \Tr_{\mathcal{F}}\biggl({a^*(x)a(x)\frac{\ee^{-H^0}}{Z^0}}\biggr)\,.
\end{equation}
Then we set 
\begin{align}\label{eq:choice of chemical, energy shift phi42}
    \mathfrak{g}^\eps_\nu \deq \tau^{\eps,0} + \varrho_\nu^{0}\,,\quad
    \mathfrak{h}^\eps_\nu \deq \frac{1}{2}\int \d{x} \,\d{\tilde{x}}\,\varrho_\nu(x)v^\eps(x- \tilde{x})\varrho_\nu(\tilde{x}) + \int \d x \,\tau^\eps(x) \varrho_\nu(x) - E^\eps\,,
\end{align}
where
\begin{equation}
    \label{def:taueps_Eeps}
    \tau^\varepsilon (x) \deq\int \d{\tilde{x}}\, v^\varepsilon (x-\tilde{x})G(x,\tilde{x})\,, \quad \quad E^\varepsilon \deq\frac{1}{2} \int \d x \,\d {\tilde x }\, v^\varepsilon(x-\tilde x ) \bigl(G(x, \tilde{x})\bigr)^2\,,
\end{equation}
and we defined the constants \(\tau^{\eps,0} \deq \restr{\tau^\eps}{U=0} \) and \(\varrho_\nu^0 \deq \restr{\varrho_\nu}{U=0}\) which do not depend on $x$ by homogeneity if $U = 0$.

By the the choice \eqref{def quantum green function diagonal}, \eqref{eq:choice of chemical, energy shift phi42}, and \eqref{def:taueps_Eeps}, one finds that the Hamiltonian \eqref{eq:Hwrenormalization} coincides with the \emph{Wick-ordered interacting Hamiltonian} 
\begin{multline}\label{eq:WickHamiltonian}
    H =  H^0 + \frac{\nu^2}{2} \int \d x\, \d {\tilde{x}}\, \bigl( a^*(x) a(x) - \varrho_\nu(x)/\nu\bigr)\, v(x-\tilde x)\, \bigl(a^*(\tilde x ) a( \tilde x)- \varrho_\nu(\tilde{x})/\nu\bigr)\\- \nu \int \d x\, \tau^\eps(x) (a^*(x) a(x)- \varrho_{\nu}(x)) - E^{\eps}\,,
\end{multline}
provided that \(U \equiv U_\nu^\eps>0\) satisfies the equation
\begin{equation}\label{counterterm problem}
    \mathcal{U}(x) = U(x)- (\tau^\eps (x)- \tau^{\eps,0}) - v^\eps*(\varrho_\nu-\varrho^0_\nu)(x)\,. 
\end{equation}
The equation \eqref{counterterm problem} relates the \emph{bare potential} \(\mathcal{U}\) with the \emph{dressed, or renormalised, potential} \(U\). Through the dependence of $\tau^\epsilon$ and $\varrho_\nu$ on $U$, it expresses the bare potential $\cal U$ in terms of the dressed potential $U$. Solving $U$ in terms of $\cal U$ requires the solution of the nonlinear integral equation \eqref{counterterm problem}, known as the \emph{counterterm problem}; it is discussed in Theorem \ref{thm:Solution to the counterterm problem} below.

Next, we define the \emph{\(p\)-particle reduced density matrix} as
\begin{equation}\label{eq:Gamma_p}(\Gamma_p)_{\mathbf{x, \tilde{x}}}\deq \Tr_{\mathcal{F}}\biggl(a^*(\tilde{x}_1) \cdots a^*(\tilde{x}_p) a(x_1) \cdots a(x_p) \frac{\mathrm{e}^{-H}}{Z}\biggr)\,.
\end{equation}
As for the correlation function \eqref{eq:corrfunction} and its Wick-ordered version \eqref{eq:wickcorrfunction}, we would like to replace \eqref{eq:Gamma_p} with its Wick-ordered version. To that end, we regard the expressions \eqref{eq:corrfunction} and \eqref{eq:wickcorrfunction} as integral kernels of operators acting on $\cal H_p$, and observe that  (see \cite[Lemma A.4]{FKSS_2020})
\begin{equation}\label{gamma_p_link}
    \widehat{\gamma}_p=\sum_{k=0}^p\binom{p}{k}^2(-1)^{p-k} \mathcal{P}_p\bigl(\gamma_k \otimes \gamma_{p-k}^0\bigr) \mathcal{P}_p\,,
\end{equation}
where \(\gamma^0_p\) is the correlation function of the free field.
By analogy with \eqref{gamma_p_link}, we define the \emph{Wick-ordered reduced density matrix} as 
\begin{equation}\label{eq:Wick ordered reduced density matrices}
    \widehat{\Gamma}_p\deq\sum_{k=0}^p\binom{p}{k}^2(-1)^{p-k} \mathcal{P}_p\bigl(\Gamma_k \otimes \Gamma_{p-k}^0\bigr) \mathcal{P}_p\,,
\end{equation}
where $\Gamma_{m}^0$ is the $m$-particle reduced density matrix of the free grand canonical density matrix $\ee^{-H^0}/Z^0$.  

\subsection{Results}
The following is our main result.

\begin{theorem}\label{thm:main_result}
    Let $d =2$ and $\theta >10$. Let  \( U\) satisfy Assumption \ref{assumption:P and D for U}, and $H$ be the Hamiltonian \eqref{eq:WickHamiltonian}.
    Suppose that Assumption \ref{assumptions: non local interaction} holds and \(\eps\equiv \eps(\nu)\) satisfies
    \begin{equation}\label{eq:assumption on d=2 size of eps}
    \eps \geq\exp \Bigl(-\bigl(\log \nu^{-1}\bigr)^{\frac{1-c}{2}}\Bigr)
    \end{equation}
    for some \(c>0\).
 Then as $\epsilon, \nu \to 0$ we have the convergence of the partition function
\begin{equation} \label{eq:conv_part_function}
\cal Z \to \zeta
\end{equation}
and of the Wick-ordered correlation functions
\begin{equation} \label{eq:decaygammanu}
\nu^p \, \wh \Gamma_p \xrightarrow{L^1 \cap L^\infty} \wh \gamma_p
\end{equation}
for all $p \in \N$.
\end{theorem}

The topology of convergence in \eqref{eq:decaygammanu} is strong enough to address most questions of physical interest, regarding both small-scale (ultraviolet) behaviour as well as large-scale (infrared) behaviour. For instance, it allows for an analysis of arbitrary correlation functions of the renormalised particle densities; see Corollary \ref{cor:density} below.

Combining Theorem \ref{thm:main_result} with the behaviour of the Green function from Proposition \ref{prop:GN} below, arguing as in \cite[Theorem 1.3]{FKSS_2020}, we obtain the following result. It gives convergence of the non-renormalised reduced density matrices, with optimal range of $L^r$ exponents.
\begin{cor}
    Under the assumptions of Theorem \ref{thm:main_result}, for every \(p \in \mathbb{N}\) and \(1\leq r<\infty\),
    \begin{equation*}
\nu^p \, \Gamma_p \xrightarrow{L^r} \gamma_p
    \end{equation*}
    \end{cor}

Another consequence of Theorem \ref{thm:main_result} is the convergence in law of the of the renormalised particle densities, similarly to \cite[Theorem 1.4]{FKSS_2020}. To state the result, we define the \emph{quantum particle density} and the \emph{classical mass density}, respectively
\[
\fra N(x) \deq a^*(x)a(x) \,, \qquad \fra n(x) \deq |\phi(x)|^2\,.
\]
Their renormalised versions are defined by\footnote{For $d = 2$, the mass density $\fra n(x)$ is ill-defined, but its renormalized version $\wick{\fra n(x)}$ is well-defined; the latter is constructed as the $L^2(\P)$-limit of its regularisation, like \eqref{eq:defV}.}
\[
\wick{ \fra N(x) } = \fra N(x) - \Tr_{\mathcal{F}}\biggl( \fra N(x) \frac{\ee^{-H^0}}{Z^0}\biggr)\,, \qquad
\wick{\fra n(x)} = \fra n(x) - \mathbb{E}[ \fra n(x)]\,.
\]

\begin{cor} \label{cor:density}
     Under the assumptions of Theorem \ref{thm:main_result}, for every $p \in \mathbb N^*$ and $1\leq r <\infty $, as $\epsilon, \nu \to 0$,
    \[
    \Tr_{\mathcal F}\biggl(\nu^p \, \wick{\mathfrak{N}(x_1)}\cdots \wick{\mathfrak{N}(x_p)} \, \frac{\ee^{-H}}{Z}\biggr) \xrightarrow{L^r} \frac{1}{\zeta} \E \qB{\wick{|\phi(x_1)|^2}\cdots \wick{|\phi(x_p)|^2} \, \ee^{-V(\phi)}}\,.
    \]
 The convergence is also locally uniform for distinct arguments \(x_1,\dots, x_p\).
\end{cor}

\begin{remark} \label{rem:theta_optimal}
Our results hold under the assumption $\theta > 10$. This restriction is technical and can be improved, for instance under stronger assumptions on the smoothness of $U$; see Remark \ref{rem:theta_condition} below.
\end{remark}

\begin{remark}
The lower bound on the interaction length from \eqref{eq:assumption on d=2 size of eps} is the same as in the homogeneous result from \cite{FKSS_24}. In the homogeneous setting, it was recently improved in \cite{nam2025phi, jougla2025phi} to $\epsilon \geq \nu^c$ for a sufficiently small constant $c > 0$, at the expense of a weaker topology of convergence in \eqref{eq:decaygammanu} (weak operator convergence against test functions with finite Fourier support.)
\end{remark}

Finally, we discuss the \emph{counterterm problem} \eqref{counterterm problem}, which relates the bare potential $\cal U$ with the dressed potential $U$, and hence the Hamiltonian \eqref{eq:Hwrenormalization} with its Wick-ordered formulation \eqref{eq:WickHamiltonian}. We make the following assumption on the bare potential $\mathcal U$.

\begin{assumption}\label{assumption:mathcalU_cp}
Fix $\theta > 2$. There exist a radially increasing function \(g\col\bR^d \to [1, \infty)\) and constants \(C,c,b >0\), \( \gamma \in (0,1)\) satisfying 
     \begin{equation}\label{eq:boundOnMaxg}
           c \abs{x}^\theta\leq g(x)\leq  C\ee^{c\abs{x}^b}\,
        \end{equation} and
        \begin{equation}\label{cp:boundU}
           cg(x)\leq \mathcal{U}(x) \leq C g^{3/2}( \gamma x ) \,   
        \end{equation}
    for almost every $x \in \bR^d$. 
\end{assumption}

\begin{remark}
If \(\mathcal{U}\) satisfies Assumption \ref{assumption:P and D for U} \assD{} then \eqref{eq:boundOnMaxg} implies \eqref{cp:boundU}. Indeed, in that case
    \[\mathcal{U}(x) \leq C + \bigg|{\int_0^1 \d{\lambda} \,\nabla \mathcal{U}(\lambda x)\cdot  x}\bigg|\leq C + C g^{3/2}(\gamma x) \cdot \abs{x} \leq Cg^{3/2}(\tilde \gamma x)\quad  \text{a.e.}  \]
    for any \(\tilde \gamma \in (\gamma,1)\) and some constant \(C\). 
\end{remark}

Our main result on the counterterm problem, Theorem \ref{thm:Solution to the counterterm problem} below, is the solvability of \eqref{counterterm problem} in $U$ for a given $\cal U$. It is a generalisation of the analogous result in \cite[Theorem 5.2 and Remark 5.3]{frohlich2017gibbs} and \cite[Proposition 1.15]{FKSS_2020}. It expands on these works in four main directions. First, it goes beyond the mean-field regime, where $\epsilon$ is fixed, to short-range interactions, where $\epsilon \to 0$. Second, we have to deal with the additional term $\tau^\eps (x)- \tau^{\eps,0}$ in \eqref{counterterm problem}, which arises from the stronger renormalisation for local interactions. Third, we control the gradient of the dressed potential, uniformly in $\nu$ and $\epsilon$. Fourth, we relax the condition of subexponential growth for $\cal U$ in the \cite{frohlich2017gibbs, FKSS_2020} to the weaker condition \eqref{eq:boundOnMaxg}. The former three are important for our application in Theorem \ref{thm:main_result_counterterm} below.

For the solvability of \eqref{counterterm problem} in $U$, we require $\kappa$ to be large enough, which is necessary as the following remark explains.

\begin{remark}
    By the Feynman-Kac formula (Lemma \ref{Feynman-Kac formula} below),  the functions \(\tau^{\eps,0}-\tau^\eps\) and \(\varrho^0_\nu-\varrho_\nu\) appearing on the right-hand side of \eqref{counterterm problem} are positive whenever \(v \geq 0\), and they are decreasing with \(\kappa>0\). For a given $\cal U > 0$, since $U$ has to be also positive, we therefore require the above positive terms to be small enough, which imposes a lower bound on $\kappa$. (Note that, since these two terms are bounded while $U$ and $\cal U$ grow at infinity, this restriction arises from the behaviour of $U$ and $\cal U$ near the origin.)
\end{remark}

\begin{theorem}\label{thm:Solution to the counterterm problem}
Let $d=2$. Let \(\mathcal{U}\) and \(g\) satisfy Assumption \ref{assumption:mathcalU_cp} and $v$ satisfy  Assumption \ref{assumptions: non local interaction}. 
Then, for \(\kappa>0\) large enough, the following holds.
\begin{enumerate}[label=(\roman*)]
    \item The counterterm problem \eqref{counterterm problem} has a unique solution \(U_{\nu}^ \eps\) for every \(\eps, \nu>0\). 
    \item Suppose that \(\epsilon \equiv \eps(\nu)>0\) satisfies  \begin{equation}\label{eq:rate of convergence of eps counterterm problem}
        \eps(\nu)>\nu^{\beta}
    \end{equation}
    with \(0\leq\beta< \frac{2-s}{6}\) and \(s\) satisfying \eqref{eq:tracehs}. There exists a unique positive function \(U\) satisfying the limiting counterterm problem \begin{equation}\label{eq:limiting counterterm equation}
       \mathcal{U}(x)=U(x)-2\biggl(\frac{1}{\kappa-\Delta / 2+U}-\frac{1}{\kappa-\Delta / 2}\biggr)(x,x)\,.
    \end{equation}
It satisfies
    \[ \lim_{\eps,\nu \to 0} \normbb{\frac{U_{\nu }^{\eps}-U}{\mathcal{U}}}_{L^\infty} = 0 \,.  \]
    \item There exists a constant \(C>0\) such that
    \[\mathcal{U}/C < U_{\nu}^\eps< C \mathcal{U}\,.\]
    \item If \(\mathcal{U}\) satisfies \assD{} in Assumption \ref{assumption:P and D for U}, there exists a constant \(C>0\) such that
    \[\abs {\nabla U_{\nu}^\eps(x)} \vee \abs {\nabla U(x)} \leq C  g(\gamma x)^{3/2}\,,\]
    and \(U^{\eps}_\nu, U\) satisfy \assD{} for all \(\nu, \eps>0\).
\end{enumerate}
\end{theorem}
We remark that a similar statement holds in dimension \(d=3\), with a similar proof; we omit the details. Combining Theorem \ref{thm:Solution to the counterterm problem} with Theorem \ref{thm:main_result} easily yields the following result.

\begin{theorem}  \label{thm:main_result_counterterm}
    Let $d =2$ and $\theta >10$. Let the bare potential \(\mathcal{U}\) satisfy Assumption \ref{assumption:P and D for U} and \ref{assumption:mathcalU_cp}. Let $v$ satisfy Assumption \ref{assumptions: non local interaction} and \(\eps>0\) satisfy \eqref{eq:assumption on d=2 size of eps}. Let \(\kappa>0\) be large enough. Let $H$ be the Hamiltonian \eqref{eq:Hwrenormalization} with renormalisation parameters $\mathfrak{g}^\eps_\nu$ and $\mathfrak{h}^\eps_\nu$ defined in \eqref{eq:choice of chemical, energy shift phi42}.
Then there exists a dressed potential $U$ which is the unique solution of \eqref{eq:limiting counterterm equation}, such that \eqref{eq:conv_part_function} and \eqref{eq:decaygammanu} hold.
\end{theorem}

In this result, the bare potential $\cal U$ is the external potential of the quantum theory, while the dressed potential $U$ is the external potential of the limiting field theory.

\section{Strategy of the proof} \label{sec:strategy}

The rest of the paper is devoted to the proofs of Theorems \ref{thm:main_result} and \ref{thm:Solution to the counterterm problem}. In this section we describe the general strategy of the proof. Throughout the paper we use $C,c$ to denote generic positive constants, which may change from one line to another, and may depend on fixed parameters. We write $a \lesssim b$ to mean $a \leq C b$. If $C$ depends on a parameter $\alpha$, we write $a \lesssim_\alpha b$. To emphasise the dependence of a quantity $X$ on the potential $U$ or on the interaction $V$, we include these parameters as a superscript and write $X^U$ or $X^V$, respectively. For a positive-definite operator $A$, we denote by $\mu_{A}$ the centred Gaussian measure with covariance $A$.

Our general strategy is the two-step approach introduced in \cite{FKSS_24}. As in that paper, we introduce two nonlocal interactions approximating $V$ from \eqref{eq:defV}:
\begin{align}\label{eq:Veps}
	V^\eps &\deq \frac12 \int \, \d x\,\d {\tilde x}\, v^\eps(x-\tilde x) \, \wick{|\phi(x)|^2 |\phi(\tilde x)|^2}\,,\\ 
	\label{eq:Weps} W^\eps &\deq \frac12    \int \d x\,\d{\tilde x} \, \wick{|\phi(x)|^2} \,v^\eps(x-\tilde x) \,\wick{|\phi(\tilde x)|^2} - \int \d x \, \tau^\eps(x) \, \wick{|\phi(x)|^2} - E^\eps\,,
\end{align}
which are rigorously constructed by regularisation, as in  \eqref{eq:defV}. 

The first step consists in estimating the difference between the relative partition function $\mathcal{Z}$ and the classical partition function $\zeta^{W^\varepsilon}$, as well as between the Wick-ordered reduced density matrices $\nu^p \widehat{\Gamma}_p$ and the Wick-ordered correlation functions $\widehat{\gamma}_p^{W^\varepsilon}$.
The second step establishes the convergence of $\zeta^{W^\varepsilon}$ and $\widehat{\gamma}_p^{W^\varepsilon}$ to $\zeta^{V}$ and $\widehat{\gamma}_p^{V}$, respectively, as $\varepsilon \to 0$, via the intermediate potential $V^\varepsilon$.

These two steps are the subject of the two following propositions, which are the counterparts of \cite[Proposition 3.1 and 3.2]{FKSS_24}, respectively.

\begin{prop}\label{prop:convergence quantum many body system and Weps}
    Let $d=2,3$, and let \(\eps \equiv\eps (\nu)\geq 0\) satisfy
    \begin{equation}\label{eq:epsbound}
    \eps(\nu) \geq 
    \begin{cases}
        \exp \Bigl(-\bigl(\log \nu^{-1}\bigr)^{\frac{1-c}{2}}\Bigr) \quad &{\rm if }\, d=2 \\
         (\log \nu^{-1})^{-\frac{1-c}{2}} \quad &{\rm if }\, d=3
    \end{cases}
    \end{equation}
        for some $c \in (0,1)$.
    Then, under the assumptions of Theorem \ref{thm:main_result}, 
    \begin{equation*}
        \lim_{\eps,\nu \to 0  } \abs{\mathcal{Z}-\zeta^{W^\eps}}=  0\,.
    \end{equation*}
     Moreover, for any \(p \in \mathbb{N}\),  we have
\begin{equation}\label{eq:convergence of relative partition function with decay}
\lim_{\eps,\nu \to 0 }  \big\|\nu^p\widehat{\Gamma}_p-\widehat{\gamma}_p^{W^\eps}\big\|_{L^1\cap L^\infty}=0 \,.
\end{equation}

\end{prop}

\begin{prop} \label{prop:convergence classical field theories}
    Let $d=2$. Under the assumptions of Theorem \ref{thm:main_result}, we have 
    \begin{equation}\label{eq:Wepspartitionfunc}\lim_{\eps \to 0} \zeta^{W^\eps} \to \zeta\,.\end{equation}
     Moreover, for any \(p \in \mathbb{N}\),  
\begin{equation}\label{eq:diff gamma-gammaW with decay}
\lim_{\eps\to 0 }  \big\|\widehat{\gamma}_p-\widehat{\gamma}_p^{W^\eps}\big\|_{L^1\cap L^\infty}=0\,.
\end{equation}
\end{prop}

Going to an unbounded domain and giving up translation invariance results in major differences in the proofs as compared to \cite{FKSS_24}, some resulting from fundamentally new difficulties which simply do not appear in the homogeneous setting.  In the rest of this section, we sketch some of these differences.

The first step, corresponding to Proposition \ref{prop:convergence quantum many body system and Weps}, relies on a functional integral representation for the quantum and classical problems, derived in \cite{FKSS_2020} and also used in \cite{FKSS_24}. To explain this representation, we focus on the partition function $\zeta^{W^\eps}$ which can be formally written as
\begin{equation*}
\zeta^{W^\eps} = \frac{1}{\int \r D \phi \, \ee^{- \scalar{\phi}{h \phi}}} \int \r D \phi \, \ee^{- \scalar{\phi}{h \phi}  - \frac{1}{2} \int \dd x \, \dd \tilde x \, :\abs{\phi(x)}^2: \, v^\epsilon(x - \tilde x) \, :\abs{\phi(\tilde x)}^2: + \int \d x \, \tau^\eps(x) \, :|\phi(x)|^2: + E^\epsilon}\,.
\end{equation*}
We rewrite the interaction term $- \frac{1}{2} \int \dd x \, \dd \tilde x \, \wick{\abs{\phi(x)}^2} \, v^\epsilon(x - \tilde x) \, \wick{\abs{\phi(\tilde x)}^2}$ using a Hubbard-Stratonovich transformation
\begin{equation} \label{HS}
\int \mu_{v^\epsilon}(\d{\xi)}\,\ee^{\ii \langle f, \xi\rangle} = \ee^{-\frac{1}{2} \langle f, v^\epsilon f\rangle}\,,
\end{equation}
where $\mu_{v^\epsilon}$ is a real centred Gaussian measure with covariance $\int \mu_{v^\epsilon}(\dd \xi) \, \xi(x) \, \xi(y) = v^\epsilon(x - y)$. Choosing $f(x) = \wick{\abs{\phi(x)}^2}$ and interchanging the order of integration, we obtain
\begin{equation*}
\zeta^{W^\eps} = \frac{1}{\int \r D \phi \, \ee^{- \scalar{\phi}{h \phi}}} \int \mu_{v^\epsilon}(\dd \xi) \int \r D \phi \, \ee^{- \scalar{\phi}{h \phi} + \ii \int \dd x \, \xi(x)\, :\abs{\phi(x)}^2: + \int \d x \, \tau^\eps(x) \, :|\phi(x)|^2: + E^\epsilon}\,.
\end{equation*}
The advantage of this representation is that the integration over $\phi$ has become Gaussian, at the expense of introducing the auxiliary Gaussian field $\xi$. However, the integral over $\phi$ on the right-hand side is badly divergent owing to the rapid quadratic growth of the positive term $\int \d x \, \tau^\eps(x) \, :|\phi(x)|^2:$, which overwhelms the negative term $- \scalar{\phi}{h \phi}$ for small $\epsilon$ (it is easy to see that $\lim_{\epsilon \to 0} \tau^\epsilon(x) = + \infty$ for all $x \in \R^d$). Thus, the interchange of the order of integration over $\xi$ and $\phi$ is not allowed owing to strong oscillations in the integral over $\phi$.

To remedy this issue, we perform the Hubbard-Stratonovich transformation \eqref{HS} with the shifted phase $f(x) = \wick{\abs{\phi(x)}^2} - \alpha(x)$ for some suitably chosen real function $\alpha$, and obtain
\begin{align*}
\zeta^{W^\eps} &= \frac{\ee^{E^\epsilon + \frac{1}{2} \int \dd x \, \dd \tilde x \, \alpha(x) v^\epsilon(x - \tilde x) \alpha(\tilde x)}}{\int \r D \phi \, \ee^{- \scalar{\phi}{h \phi}}} \int \mu_{v^\epsilon}(\dd \xi)
\\
&\qquad \times
\int \r D \phi \,
\exp\pbb{- \scalar{\phi}{h \phi} + \ii \int \dd x \, \xi(x) \, (\wick{\abs{\phi(x)}^2} - \alpha(x)) - \int \dd x \, q^\epsilon(x) \, \wick{\abs{\phi(x)}^2}}
\\
&=\ee^{E^\epsilon + \frac{1}{2} \int \dd x \, \dd \tilde x \, \alpha(x) v^\epsilon(x - \tilde x) \alpha(\tilde x)}
\int \mu_{v^\epsilon}(\dd \xi)
\, \E \qB{\ee^{- \scalar{\phi}{h \phi} + \ii \int \dd x \, \xi(x) \, (:\abs{\phi(x)}^2: - \alpha(x)) - \int \dd x \, q^\epsilon(x) \, :\abs{\phi(x)}^2:}}
\end{align*}
where we abbreviated
\begin{equation*}
q^\epsilon(x) \deq \int \dd \tilde x \, v^\epsilon (x - \tilde x) \alpha(\tilde x) - \tau^\epsilon(x)\,.
\end{equation*}
The term involving $q^\epsilon$ is potentially dangerous, since, as explained above, it can become very large and hence overwhelm the quadratic decay provided by the term $- \scalar{\phi}{h \phi}$. Its vanishing is equivalent to $v^\eps * \alpha^\eps = \tau^\eps$, i.e.
\begin{equation} \label{eq:eqAmirali}
\int \dd \tilde x \, v^\epsilon(x - \tilde x) \, \alpha(\tilde x) = \int \dd \tilde x \, v^\epsilon(x - \tilde x) \, G(x,\tilde x)\,.
\end{equation}
In the \emph{homogeneous case}, where $G(x,\tilde x) = G(x-\tilde x)$, solving \eqref{eq:eqAmirali} is trivial since $\alpha$ can simply be chosen as a constant.

On the other hand, whenever $G$ is not translation invariant, \eqref{eq:eqAmirali} generically has no solution for any $\epsilon > 0$; the basic heuristic is that its right-hand side is much rougher than its left-hand side (which is smooth by smoothness of $v^\epsilon$). As an illustration of this phenomenon, in Appendix \ref{appendix:counter example} we show the non-existence of solutions to \eqref{eq:eqAmirali} for any discontinuous locally bounded potentials \(U\) and any positive interaction \(v\) satisfying Assumption \ref{assumptions: non local interaction}.

Our basic strategy in the inhomogeneous setting, where $q^\epsilon$ cannot be chosen to vanish, is to choose the function $\alpha$ in a suitable way so that the integral over $\phi$ remains well-defined, and to absorb the function $q^\epsilon$ into the external potential $U$ in the Hamiltonian $h$. This latter step results in a modified, $\epsilon$-dependent, external potential, which results in a modified $\epsilon$-dependent covariance for the classical field theory. As a consequence, we have to re-Wick-order all diverging quantities after the change of covariance.

Let us explain a bit more carefully how this strategy is implemented. In fact, owing to the growth of the potential $U$ inside $h$ in its argument $x$, it is sufficient to construct $\alpha$ such that 
\[\int \dd \tilde x \, v^\epsilon (x - \tilde x) \, \alpha^\epsilon(\tilde x) - \tau^\epsilon(x) \geq 0\]
holds only in a ball of sufficiently large radius (depending on $\epsilon$). Since $\tau^\epsilon$ decays in $x$ for any fixed $\epsilon$ by the decay of $G$, this can be done by choosing $\alpha \equiv \alpha^\epsilon$ supported in a large enough ball.

Denoting $U^\epsilon \deq U + q^\epsilon$, we therefore obtain, for every $\epsilon > 0$, a modified one-particle Hamiltonian
\[h^{U^\eps}\deq \kappa- \frac{\Delta}{2}+U^\eps \] satisfying the conditions in Remark \ref{rmk:thetha_s2}. This results in an $\epsilon$-dependent tilting of the Gaussian measure $\P$, according to the elementary formula
\begin{equation*}
           \ee^{-\langle z,B z\rangle}\mu_{A^{-1}}(\d z) =\mu_{(A+B)^{-1}}(\d z)\int \ee^{-\langle y,B y\rangle}\mu_{A^{-1}}(\d y) 
\end{equation*}
for a finite-dimensional Gaussian measure \(\mu_{A^{-1}}\) with covariance \(A\) (change of covariance and correction by multiplicative factor). This tilting is rigorously performed by truncation of $\P$ to restrict it to finite-dimensional measures.

Having tilted the free Gaussian measure $\P$ with covariance $h^{-1}$, we have to adapt the Wick-ordering of the mass \({\abs{\phi(x)}}^2\) with respect to the new measure with covariance \((h^{U^\eps})^{-1}\). This yields 
\begin{multline}\label{eq:formal rewritting of rel part funct classical}
    \zeta^{W^\eps}=\cal C \ee^{E^\epsilon + \frac{1}{2} \int \dd x \, \dd \tilde x \, \alpha^\epsilon(x) v^\epsilon(x - \tilde x) \alpha^\epsilon(\tilde x)} 
    \\ 
   \times
    \E^{U^\epsilon}  \qB{\ee^{ - \frac12  \int \d x\,\d{\tilde x} \, \bigl(|\phi(x)|^2-G^{U^\eps}(x,x)- \alpha^\eps(x)- {T}^\eps(x)\bigr) v^\eps(x-\tilde x) \bigl(|\phi(\tilde x)|^2-G^{U^\eps}(\tilde x,\tilde x)- \alpha^\eps(\tilde x)- {T}^\eps(\tilde x)\bigr)}}
\end{multline}
with
\begin{equation}\label{eq:changeOfMeasureConstant}
    \cal C \deq  \E \qB{\ee^{-\int \d{x}\,\bigl(\alpha^\eps(x)- \tau^\eps(x)\bigr){\bigl(\abs{\phi(x)}^2}-G^U(x,x)\bigr)}}
\end{equation}
and 
\begin{equation}\label{eq:diffGreenFunctions}
    T^\eps(x)=G^U(x,x)-G^{U^\eps}(x,x)\,.
\end{equation}

As in  \cite{FKSS_2020}, the derivation above is achieved rigorously by introducing a finite-dimensional regularisation of the field \(\phi\), and adding a spatial cutoff to the interaction.  To that end, we choose an even, compactly supported mollifier \(\varphi\) that is positive definite and satisfies \(\varphi(0)=1\). Moreover, we choose the regularisation function $\vartheta$ in \eqref{eq:vartheta} as
 \(\vartheta(x)=\mathbbold{1}_{\abs{x}\leq 1}\).  
 The regularised, finite dimensional version of the interaction \(W^\eps\) is defined as 
\begin{equation}\label{eq:def:regularised version of the field Weps}
     W^\eps_{\eta,N} \deq \int \d x \,\d {\tilde x}\, \wick{\abs{\phi_N({x})}^2} \,\varphi(\eta x)v^\eps (x-\tilde x)\varphi(\eta\tilde{x})\, \wick{\abs{\phi_N(\tilde{x})}^2} -\int \d{x}\, \tau^\eps(x) \wick{\abs{\phi_N({x})}^2}- E^{\eps}
\end{equation}
for any \(N>0\) and \(\eta \geq 0\).

These steps require a careful treatment of the convergence of the finite dimensional counterparts of \eqref{eq:changeOfMeasureConstant} and \eqref{eq:diffGreenFunctions} to the infinite dimensional case. More specifically, we perform the Hubbard-Stratonovich transformation described above in the regularised, finite-dimensional version of \(W^\eps\), and thus obtain a functional integral representation where the external potential
\begin{equation}
    \label{def:Uepseta}U^\eps_\eta \deq U - \alpha^\eps_\eta + \tau^\eps 
\end{equation}
depends on \(\eps, \eta>0\), and we set
\begin{equation} \label{def:rhoepseta}
        \alpha^\eps_\eta (x)\deq\int \d {\tilde{x}}\, \alpha^{\eps}(\tilde{x})\, \varphi(\eta x) v^\eps(x-\tilde x) \varphi(\eta \tilde{x})\,. 
    \end{equation}
This concludes the overview of the functional integral representation for $\zeta^{W^\eps}$ in the inhomogeneous setting.
    
In order to establish the connection to the quantum many-body system and its associated functional integral representations, we have to derive and control a representation analogous to the one described above in the quantum setting. This requires a suitable adaptation of the space-time functional integral from \cite{FKSS_2020}, in a form that we can relate to \eqref{eq:formal rewritting of rel part funct classical}. We show that the quantum terms arising from this algebraic transformation converge to their classical counterparts in \eqref{eq:formal rewritting of rel part funct classical}. For this purpose, we need quantitative estimates on the difference of these quantities as \(\nu, \eps\to0\). These estimates are generalisations of the results of \cite{FKSS_24} and provide stronger bounds compared to previous works.

The second step, in analogy to \cite{FKSS_24}, consists in showing integrability of $\ee^{-V^\eps}$, uniformly in $\eps$, to deduce integrability of $\ee^{-W^\eps}$ using hypercontractivity. Moreover, we use Gaussian integration by parts, to
derive a representation of the correlation functions in terms of expectations of derivatives of the
interaction potential. 
The main challenge in this step consists in strong bounds on the decay and regularity of the (truncated) Green function of the operator $h = \kappa - \Delta/2 +U$.

\subsubsection*{Outline of the paper}
We conclude this section with an outline of the paper. In Section \ref{sec:field} we show Proposition \ref{prop:convergence classical field theories}, while in Section \ref{sec:funcintrep} we develope the functional integral representation described above. In Section \ref{sec:quantumtoclassic} we prove Proposition \ref{prop:convergence quantum many body system and Weps}. In Section \ref{app:Greenfunction} we show properties of the Green function which are crucial to show Proposition \ref{prop:convergence classical field theories}. Finally, in Section \ref{Section:Counterterm problem} analyse the counterterm problem and prove Theorem \ref{thm:Solution to the counterterm problem}.

\section{Proof of Proposition \ref{prop:convergence classical field theories}} \label{sec:field}

This section is dedicated to the proof of Proposition \ref{prop:convergence classical field theories}; it is the analogue of \cite[Section 4]{FKSS_24} in the inhomogeneous setting. We derive first $L^2$-estimates for the differences $V^\varepsilon -V$ and $V^\varepsilon -W^\varepsilon$. Then we show integrability of $\ee^{-V^\varepsilon}$, which also implies the integrability of $\ee^{-W^\varepsilon}$.
Thus we show the convergence of the partition functions \eqref{eq:Wepspartitionfunc} and of the correlation functions \eqref{eq:diff gamma-gammaW with decay}.

Throughout this section we assume that $d = 2$. We abbreviate $L^2 \equiv L^2(\P)$. Moreover, \(g\) is a radially increasing lower bound (up to a multiplicative constant) of the potential \(U\). In the case where \(U\) satisfies Assumption \ref{assumption:P and D for U} \assP{}, \(g(x)= 1+\abs{x}^\theta\). For fixed \(0<\gamma<1 \), we also define 
\begin{equation}\label{eq:gtilde}
\tilde g(x)\deq g(\gamma x)\,.
\end{equation}

\subsection{Partition function}

We start by deriving $L^2$-estimates for the differences  $V^\varepsilon -W^\varepsilon$ and $V^\varepsilon -V$.

\begin{lemma}
	\label{lm:Veps-Weps}
	Consider $V^\eps$ and $W^\eps$ as in  \eqref{eq:Veps} and \eqref{eq:Weps}. Let $\theta>2$. Then 
	\[ \|V^\eps -W^\eps\|_{L^2} \to 0 \quad \quad \text{as } \eps \to 0\,. \]
	
\end{lemma}
\begin{proof} The proof is analogous to the one of \cite{FKSS_24}. The main novelty involves controlling the large-$x$ regime, which is done using the Green function estimates from Section \ref{app:Greenfunction}.
	
	A straightforward computation shows that 
	\[\begin{split}
		W^\eps -V^\eps & =\int \, \d{x}\,\d{\tilde x}\, v^\eps(x-\tilde x) G(x,\tilde x) \bigl( \wick{\overline{\phi}(x) \phi(\tilde x)} - \wick{\abs{\phi(x)}^2}\bigr)\,\\
		& =\, (W^\eps - V^\eps)_R + (W^\eps - V^\eps)_{R^c}\,,
	\end{split}
	\]
	where we split the difference $W^\eps -V^\eps$ in two regions: inside and outside of a ball of radius $R =R(\eps)$. We analyse the two terms separately, using different arguments. To estimate the term inside the ball, we can employ a-priori bounds on the gradient of Green's function, whereas outside the ball, we rely on its exponential decay.
    Let us consider first the region $|x| \leq R$. Then using Wick ordering formula and since $G(x;y) =  \int \d{\mu} \, \overline \phi (y) \phi(x)$ we have
	\[\begin{split}
		\|(V^\eps -W^\eps)_R\|^2_{L^2} & =\int_{|x|,|y|\leq R} \d{x}\,\d{\tilde x}\, \d{y}\,\d{\tilde y}\, v^\eps(x-\tilde x) v^\eps(y-\tilde y) G(x,\tilde x) G(y, \tilde y) \\
		&\hspace{3cm}\times\bigl(G(\tilde y, x) - G(y,x)\bigr) \bigl(G(\tilde x, y)- G(x,y)\bigr)\,.\\
	\end{split}\]
	Hence, we use the mean value theorem to rewrite the difference $G(\tilde y,x) - G(y,x)$ and $G(\tilde x, y)- G(x,y)$ as 
	\begin{equation}\label{eq:diffG}
	\big|G(x,y) - G(x,z)\big| 
    \lesssim \int_0^1 \d{\lambda}\, |y-z| |\nabla_y G(x, y + (z-y)\lambda)|\,.
	\end{equation}
	Together with the bound \eqref{eq:gradG}, in Appendix \ref{app:Greenfunction}, we have
	\begin{equation}\label{eq:VepsWeps}\begin{split}
		&\|(V^\eps -W^\eps)_R\|^2_{L^2}\\
         &\lesssim \int_{|x|,|y|\leq R}\, \d{x}\,\d{\tilde x}\, \d{y}\,\d{\tilde y}\, v^\eps(x-\tilde x) v^\eps(y-\tilde y) G(x,\tilde x)G(y, \tilde y)\\ 
		&\hspace{4cm}\times\int_0^1 \d \lambda\,   \frac{|\tilde y-y|}{|\tilde y + (y - \tilde y)\lambda -x|}\int_0^1 \d s\, \frac{|\tilde x -x|}{|\tilde x +(x-\tilde x)s-y|} \\
		&\lesssim \int_{|x|,|y|\leq R}\, \d{x}\,\d{y}\, \d{h}\,\d{k}\, v(h) v(k) G(x,\eps h +x)G(y, \eps k+y) \\ 
		&\hspace{5cm}\times  \Big|\log\frac{|x-y +\eps k|}{|x-y|}\Big| \Big|\log\frac{|y-x +\eps h|}{|y-x|}\Big|
	\end{split}\end{equation}
	where in the last line, after integrating,  we used the change of variable $h = (x-\tilde x)/\eps$ and $k = (y-\tilde y)/\eps$. To bound the term in the last line on the right-hand side in   \eqref{eq:VepsWeps} we use on one side $G(x,y)|x-y| \leq C_\alpha |x-y|^\alpha \lesssim_\alpha |x-y|^\alpha $ for $\alpha \in (0,1)$, which is a consequence of Proposition \ref{prop:GN} and on the other the bound (as in \cite{FKSS_24}) 
	\begin{equation}\label{eq:boundlog}
	\Big|\log\frac{|x+y|}{|x|}\Big| \lesssim \begin{cases}
		|\log|x+y|| + | \log|x|| \quad &\text{if } |x|\leq 2|y|\\
		\frac{|y|}{|x|}\quad\quad &\text{if } |x|>2|y\,|.
	\end{cases}
	\end{equation}
    Thus for $0<\alpha' <\alpha <1$ we have 
	\begin{equation}
		\begin{split}
			&\int_{|x|,|y|\leq R}\! \d{x}\,\d{y}\, \d{h}\,\d{k}\, v(h) v(k) G(x,\eps h +x)G(y, \eps k+y)\Big|\log\frac{|x-y +\eps k|}{|x-y|}\Big| \Big|\log\frac{|y-x +\eps h|}{|y-x|}\Big|\\ 
            & \lesssim \int_{|x|,|y|\leq R}\! \d{x}\,\d{y}\, \d{h}\,\d{k}\, v(h) v(k) G(x,\eps h +x)G(y, \eps k+y)\eps^\alpha |h|^\alpha \eps^\alpha |k|^\alpha \mathbbm{1}_{|x-y| \leq 2 \eps \min(|k|,|h|)}  \\
            & \quad + \int_{|x|,|y|\leq R}\! \d{x}\,\d{y}\, \d{h}\,\d{k}\, v(h) v(k) G(x,\eps h +x)G(y, \eps k+y)2 \frac{\eps^\alpha |h|^\alpha \eps|k|}{|x-y|}\mathbbm{1}_{2\eps|h|\geq|x-y|>2\eps |k|} \\
            & \quad + \int_{|x|,|y|\leq R}\! \d{x}\,\d{y}\, \d{h}\,\d{k}\, v(h) v(k) G(x,\eps h +x)G(y, \eps k+y)\frac{\eps|k|\eps |h|}{|y-x|^2}\mathbbm{1}_{|x-y|>2\eps \max(|h|,|k|)}\\
            &\lesssim R^2\eps^{2\alpha'}\,,
		\end{split}
	\end{equation}
     where in the last step we used 
    \[\begin{split}
    &\int_{|x|,|y|\leq R}\! \d{x}\,\d{y}\, G(x,\eps h +x)G(y, \eps k+y)\frac{\eps|k|\eps |h|}{|y-x|^2}\mathbbm{1}_{|x-y|>2\eps \max(|h|,|k|)}\\  &\lesssim\int_{|x|,|y|\leq R}\! \d{x}\,\d{y}\, \frac{(\eps|k|\eps |h|)^{\tilde\alpha}}{|y-x|^{2\alpha}}\mathbbm{1}_{|x-y|>2\eps \max(|h|,|k|)}\lesssim (\eps|k|\eps |h|)^{\alpha'} R^2\,.
    \end{split}
    \]
	Next, we consider the region outside the ball, i.e.\ $|x|>R$. Moreover, we consider the two terms appearing in $(W^\eps -V^\eps)_{R^c}$ proportional to $\wick{\overline{\phi}(x) \phi(\tilde x)}$ and  $\wick{|\phi(x)|^2}$ separately. 
	
	On one side,  performing a similar change of variables as above and using the bound \eqref{eq:boundG} we have
	\begin{equation}\label{eq:WepsVepsRc}
	\begin{split}
		& \Big\|\int_{|x|>R}  \d{x}\,\d{\tilde x}\, v^\eps(x-\tilde x) G(x,\tilde x)  \wick{|\phi(x)|^2} \Big\|_{L^2}\\
		&\quad \lesssim \int_{|x|,|y|>R}\, \d{x} \,\d{y}\,\d{h}\,\d{k}\, v(h) v(k) G(x, x+\eps h) G(y, y+\eps k)  G(y,x)^2\\
		&\quad \lesssim \int_{\substack{|x|,|y|>R}}\d{x} \,\d{y}\,\d{h}\,\d{k}\, v(h) v(k) (\log(\eps|h|)^{-1}\vee 1)(\log(\eps|k|)^{-1} \vee 1)  \\&\hspace{6cm}\times\bigl((\log|x-y|^{-1})^2 \vee 1\bigr) \ee^{-c(|x-y|)\sqrt{\tilde g(x) +\tilde g(y)}}\\  
		&\quad \lesssim \int_{\substack{|x|,|y|>R\\
				|x-y|\leq 1}}\, \d{x}\, \d{y}\, \Bigl[\Bigl(\log\frac{1}{\eps}\Bigr)^{2}\vee 1\Bigr] \bigl[(\log|x-y|^{-1})^2 \vee 1\bigr] \ee^{-c|x-y|\sqrt{\tilde g(x) +\tilde g(y)}} \\
		&\quad \quad + \int_{\substack{|x|,|y|>R\\
				|x-y|> 1}}\, \d{x} \,\d{y}\, \Bigl[\Bigl(\log\frac{1}{\eps}\Bigr)^{2} \vee 1\Bigr] \bigl[(\log|x-y|^{-1})^2 \vee 1\bigr] \ee^{-c|x-y|\sqrt{\tilde g(x) +\tilde g(y)}}\,.
	\end{split}
	\end{equation}
	Now, for the region where $|x-y|>1$ the integral is controlled by 
	\[
    \begin{split}
    &\int_{\substack{|x|,|y|>R\\
				|x-y|> 1}}\, \d{x}\, \d{y}\, \Bigl[\Bigl(\log\frac{1}{\eps}\Bigr)^{2} \vee 1\Bigr] \bigl[(\log|x-y|^{-1})^2 \vee 1\bigr] \ee^{-c|x-y|\sqrt{\tilde g(x) +\tilde g(y)}} \\
	&\quad \lesssim \int_{\substack{|x|,|y|>R\\
			|x-y|> 1}}\d{x}\,\d{y}\, \Bigl[\Bigl(\log\frac{1}{\eps}\Bigr)^{2} \vee 1\Bigr]  \ee^{-c(\sqrt{\tilde g(x)}+\sqrt{\tilde g(y)})} \lesssim \Bigl[\Bigl(\log\frac{1}{\eps}\Bigr)^{2} \vee 1\Bigr] R^{-10}\,.
	\end{split}\]
	As for $|x-y| \leq 1$, we have 
	\[
	\begin{split}
		&\int_{\substack{|x|,|y|>R\\
				|x-y|\leq 1}}\, \d{x}\, \d{y}\, \Bigl[\Bigl(\log\frac{1}{\eps}\Bigr)^{2} \vee 1\Bigr] \bigl[(\log|x-y|^{-1})^2 \vee 1\bigr] \ee^{-c|x-y|\sqrt{\tilde g(x) }} \\
		& \quad \lesssim \int_{\substack{|x|>R\\
		}}\, \d{x}\, \frac{1}{\tilde g(x)}\int_{|z|\leq \sqrt{\tilde g(x)}} \d{z} \,\ee^{-c|z| } \Bigl[\Bigl(\log\Bigl(\frac{\sqrt{\tilde g(x)}}{|z|}\Bigr) \Bigr)^2 \vee 1\Bigr]  \Bigl[\Bigl(\log\frac{1}{\eps}\Bigr)^{2} \vee 1\Bigr] \\
		& \quad \lesssim  \int_{\substack{|x|>R\\
		}} \d{x}\,\frac{1}{\tilde g(x)} (\log (\tilde g(x))^2 \Bigl[\Bigl(\log\frac{1}{\eps}\Bigr)^{2} \vee 1\Bigr] \lesssim R^{-c}\Bigl[\Bigl(\log\frac{1}{\eps}\Bigr)^{2} \vee 1\Bigr] 
	\end{split}\] 
	where the last step is justified by Assumption \ref{assumption:P and D for U} on $U$ and thus on $g$.
	
	On the other side, if we consider the term proportional to $ \wick{\overline{\phi}(x) \phi(\tilde x)}$ and using the fact that $v$ has compact support,  we have 
	\[
	\begin{split}
		&\Big\|\int_{|x|>R} \d{x}\,\d{\tilde x}\, v^\eps(x-\tilde x) G(x,\tilde x)  \, \wick{\overline{\phi}(x) \phi(\tilde x)}\Big\|_{L^2} \\
		& \quad \lesssim\int_{\substack{|x|,|y| >R-1\\ |\tilde x|, |\tilde y| >R-1}} \d{x}\, \d{y}\,\d{\tilde x}\,\d{\tilde y}\, v^\eps(x-\tilde x) v^\eps(y-\tilde y) G(x,\tilde x) G(y, \tilde y) G(x,y)G(\tilde x,\tilde y)\\
		&\quad \lesssim\int_{\substack{|x|,|y| >R-1\\ |\tilde x|, |\tilde y| >R-1}} \d{x} \, \d{y}\, \d{\tilde x}\,\d{\tilde y}\, v^\eps(x-\tilde x) v^\eps(y-\tilde y) G(x,\tilde x) G(y, \tilde y)\bigl( G(x,y)^2 + G(\tilde x,\tilde y)^2\bigr) \\
	\end{split}\]
	which can be bounded analogously as in \eqref{eq:WepsVepsRc}. Putting together the bounds for $ (W^\eps - V^\eps)_R,  (W^\eps - V^\eps)_{R^c}$, and optimizing over $R= R(\eps)$, we can conclude for some $\tilde c, c>0$
	\[
	\|V^\eps -W^\eps\|^2_{L^2}  \lesssim R^{\tilde c} \eps^c + R^{-c} \Bigl(\log\frac{1}{\eps}\Bigr)^{2} \lesssim \eps^{c/2} \to 0 \quad \text{as } \quad \eps \to 0\,.
    \qedhere
    \]
	
\end{proof}

Now we compare the difference $V^\eps -V$ in $L^2$-norm. The proof is on the same line as the one above. We have the following lemma.

\begin{lemma}\label{lm:Veps-V}
	Let $\theta >2$. Let $V$ and $V^\eps$ defined as in \eqref{eq:defV}, \eqref{eq:Veps}. Then we have
	\[
	\|V^\eps- V\|_{L^2} \to 0  \quad \quad \text{as } \eps\to 0\,.
	\]
\end{lemma}
\begin{proof}
	Proceeding similarly as in the proof of Lemma \ref{lm:Veps-Weps} we consider $V^\eps -V$ inside and outside a ball  of radius $R$. Clearly we have,
	\begin{equation}\label{eq:Veps-V}\begin{split}
	V^\eps -V & =\frac12\int \,\d{x} \, \d{\tilde x} \,\bigl(v^\eps(x-\tilde x) -\delta(x-\tilde x)\bigr) \wick{|\phi(x)|^2|\phi(\tilde x)|^2}\\
     &\eqd  (V^\eps -V)_{R} + (V^\eps -V)_{R^c}\,.
	\end{split}\end{equation}
	The term $(V^\eps -V)_R$ can be estimated as $(W^\eps -V^\eps)_R$ (compare also with \cite[Lemma 4.2]{FKSS_24}), and in particular using the Wick formula (cf.\ \cite[Appendix A]{FKSS_24}) we have, 
    \[\begin{split}
		\|(V-V^\eps)_R\|_{L^2}^2 & = \frac12 \int_{|x|,|y| <R} \d{x}\,\d{\tilde x}\, \d{y}\, \d{\tilde y}\, \bigl(v^\eps(x-\tilde x) -\delta(x-\tilde x)\bigr)\bigl(v^\eps(y-\tilde y) -\delta(y-\tilde y)\bigr)\\
		&\hspace{3cm}\times G(y,x)G(\tilde y, \tilde x)\bigl[G(x,y)G(\tilde x,\tilde y) + G(x,\tilde y) G(\tilde x,y)\bigr]\,.\\
	\end{split}
	\]
    We analyse only the first term on the right-hand side; the second one can be bounded similarly. Using the change of variables $\tilde x - x= h$ and $\tilde y -y =k$, and that $\int v =1$ we have 
    \[
    \begin{split}
        & \frac12 \int_{|x|,|y| <R} \d{x}\,\d{y}\,G(y,x)^2\int \d{h} \,\d{k}\,v^\eps(h)v^\eps(k)\bigl[ G(y+k, x+h)^2 - G(y,x+h)^2 \\
        & \hspace{8cm}- G(y+k, x)^2 + G(y,x)^2\bigr]\,.
    \end{split}
    \]
    We can bound the sum of the first two terms on the right-hand side using the bound \eqref{eq:diffG} and \eqref{eq:boundlog} 
    \[
    \begin{split}
        &\frac12 \int_{|x|,|y| <R} \d{x}\,\d{y}\,G(y,x)^2\int \d{h}\, \d{k}\,v^\eps(h)v^\eps(k)\bigl( G(y+k, x+h) + G(y,x+h)\bigr)\\
        &\hspace{7cm}\times\big| G(y+k, x+h) - G(y,x+h)\big|\\
        & \lesssim\int_{|x|,|y| <R} \d{x}\,\d{y}\,G(y,x)^2\int \d{h} \,\d{k}\,v^\eps(h)v^\eps(k)\bigl( G(y+k, x+h) + G(y,x+h)\bigr)\\
        &\hspace{3cm} \times\Bigl[ \bigl(\big|\log|y+k-x-h|\big| + \big|\log|y-x-h|\big|\bigr)\mathbbm{1}_{|y-x-h| \leq 2|k|}\\
        &\hspace{9cm}+ \frac{|k|}{|y-x-h|}\mathbbm{1}_{|y-x-h| \geq 2|k|} \Bigr]\\
        &\lesssim \eps R^{C}\,.
    \end{split}
    \]

    As for the second term on the right-hand side of \eqref{eq:Veps-V}
	\[\begin{split}
		\|(V-V^\eps)_{R^c}\|_{L^2}^2 \lesssim &  \int_{\substack{|x|,|y|>R-1\\ |\tilde x|,|\tilde y| >R-1}} \d{x}\,\d{\tilde x}\, \d{y} \,\d{\tilde y}\, \bigl(v^\eps(x-\tilde x) -\delta(x-\tilde x)\bigr)\bigl(v^\eps(y-\tilde y) -\delta(y-\tilde y)\bigr)\\
		&\hspace{4cm}\times\bigl[ G(y,x)^4 + G(\tilde y, \tilde x)^4 +  G(x,\tilde y)^4 + G(\tilde x,y)^4\bigr]\\ 
        \lesssim &  \int_{\substack{|x|,|y|>R-1\\ |\tilde x|,|\tilde y| >R-1}} \d{x} \,\d{\tilde x} \,\d{y}\, \d{\tilde y}\,v^\eps(x-\tilde x)v^\eps(y-\tilde y) G(y,x)^4\\
        \lesssim & \int_{\substack{|x|,|y|>R-1}} \d{x}\, \d{y}\, G(y,x)^4 \lesssim R^{-c}\,,
	\end{split}
	\]
where in the last line we handled the term as the equivalent bound in \eqref{eq:WepsVepsRc} in Lemma \ref{lm:Veps-Weps}. Hence, we can conclude analogously as Lemma \ref{lm:Veps-V}. 
\end{proof}

We can now show the integrability of $\ee^{-V^\eps}$, uniformly in $\eps$. Analogously as in \cite[Sec 4.2]{FKSS_24}, the following result holds, which is an adaptation of Nelson's argument \cite{nelson1973free}. In the following, we use the notation $(\cdot)^\pm$ to denote $(\cdot)^{\pm\delta}$ for any constant $\delta>0$. 
\begin{prop}\label{prop:nelson}
	Let $\theta >10$. Then there is a constant $c>0$ such that for all $\eps >0$, $t\geq 1$ and $\gamma <\frac14 - \frac1{2\theta}$ we have 
	\[
	\mathbb{P}(\ee^{-V^\eps}>t) \lesssim \exp\hB{-c (\log t)^{(-\frac{1}{2\gamma}+)(\frac12-\frac1\theta)}}.\]
	The same bound holds for $V^\eps$ replaced by $V$.
\end{prop}
The structure of the proof of Proposition \ref{prop:nelson} follows the proof in \cite[Prop. 4.3]{FKSS_24}. However, there are a few key points that need to be changed in our setting. First, we need to redefine the truncated Green function $G_N$ and the truncated field $\phi_N$ with respect to the definition used in \cite{FKSS_24}. Indeed, since the system is non-translational invariant, we cannot use properties of the Fourier coefficients as on the torus. To this end, we consider the definition of $\phi_N$ as in  \eqref{eq:vartheta}, where now we choose  the truncation to be $\vartheta(x)\deq \ee^{-x} $.
Hence, the truncated field \eqref{eq:vartheta} and Green function \eqref{eq:defGN} are respectively given by
\begin{equation}\label{eq:phiGdelta}
	\phi_N(x) = \sum_{k=0}^{\infty} \frac{X_k}{\sqrt{\lambda_k}} 
	\ee^{-\frac{ \lambda_k}{2N}}u_k\;, \quad\quad G_N = \sum_{k=0}^{\infty} 
	\frac{\ee^{-\frac{\lambda_k}{N}}}{\lambda_k}u_k u_k^*\,.   \end{equation}
Then, we define the truncated version of $V^\eps$, as in \eqref{eq:Veps}, converging to $V^\eps$ as $N\to \infty$
\begin{equation}\label{eq:defVepsN}
	V^\eps_N \deq \frac12 \int \, \d x\,\d \tilde x\, v^\eps(x-\tilde x) \wick{|\phi_N(x)|^2 |\phi_N(\tilde x)|^2}\,.
\end{equation}
 Let $(Y_k)_{k\in \bN}$ be a family of i.i.d standard complex Gaussian variables, independent of the family $(X_k)_{k \in \bN}$. For $0< N\leq M \leq \infty$ we define the field 
\[
\psi_{N,M} \deq\sum_{k=0}^{\infty} \frac{Y_k}{\sqrt{\lambda_k}} 
\Bigl(\ee^{-\frac{\lambda_k}M} -\ee^{-\frac{\lambda_k}N}\Bigr)^{1/2}u_k,\] 
by construction $\phi_N$ and $\psi_{N,M}$ are independent. For $M<\infty$, they are almost surely continuous on $\mathbb{R}^2$\footnote{As \(U \in L^\infty_{\text{loc}}(\mathbb{R}^d) \subset L^1_{\text{loc}}(\mathbb{R}^d)\), we can conclude from \cite[10.1 Remark (3)]{lieb2001analysis} that the eigenfunctions of \eqref{def:one body hamiltonian} are \(\alpha\)-Hölder continuous for any \(\alpha<1\) by elliptic regularity of \(h-\lambda \) for any \(\lambda\in \mathbb{R}\). } . From \eqref{eq:phiGdelta} we deduce
\[
\begin{split}
	\mathbb E [\psi_{N,M}(x)\overline\psi_{N,M}(y)] & =  G_M(x,y) -G_N(x,y)\,.
\end{split}\]
Since $\phi_N$ and $\psi_{N,M}$ are independent, for $N \leq M$ we can decompose the fields into low and high frequencies\footnote{Here $\overset{\dd}{=}$ stands for equality in law.} 
\[
\phi_N + \psi_{N,M} \overset{\dd}{=} \phi_M 
\]
so that for $M =\infty$ 
\[
\phi_N + \psi_{N,\infty} \overset{\dd}{=} \phi\,.
\]
By the definition of $\phi_M$ we have, for any $N\leq M$,
\[
V_M^\eps \overset{\dd}{=} \frac12 \int \d{x}\d{\tilde x}\, v^\eps(x-\tilde x) \wick{|\phi_N(x) + \psi_{N,M}(x)|^2 |\phi_N(\tilde x) + \psi_{N,M}(\tilde x)|^2}\,.
\]
The following lemma is the first ingredient to show Proposition \ref{prop:nelson}.
\begin{lemma}\label{lm:boundVN}
	Let \(s \in (1,2)\) satisfy \eqref{eq:tracehs}.  There exists a constant $C>0$ depending on $v$ such that almost surely 
	\[
	V_N^\eps \geq - CN^\gamma\log N\] 
	for all $\eps>0$, and $\gamma=s-1$. 
\end{lemma}
\begin{proof}
	Let us denote $S= 1+ \int_{\mathbb R^2} \d x\,|v(x)|$. Using Wick ordered form of \eqref{eq:defVepsN} we find the lower bound
	\begin{equation}\begin{split} \label{eq:boundVNeps}
		V_N^\eps & = \frac12 \int \d x\,\d\tilde x\, v^\eps (x-\tilde x) \Bigl[|\phi_N(x)|^2 |\phi_N(\tilde x)|^2  -G_N(x,x)|\phi_N(\tilde x)|^2 - G_N(\tilde x,\tilde x)|\phi_N(x)|^2 \\
		&\hspace{3cm } - 2\text{Re} G_N(x,\tilde x) \phi_N(x)\overline\phi_N(\tilde x) +G_N(x,x)G_N(\tilde x,\tilde x) + G_N(x,\tilde x)^2\Bigr] \\
		& = \frac12 \int \d x \,\d\tilde x\, v^\eps (x-\tilde x) \Bigl[ \bigl( |\phi_N(x)|^2 - SG_N(x,x)\bigr)\bigl(|\phi_N(\tilde x)|^2 - SG_N(\tilde x,\tilde x)\bigr)\\
		&\hspace{3.5cm} + 2(S-1) G_N(x,x)|\phi_N(\tilde x)|^2 -2\text{Re} G_N(x,\tilde x) \phi_N(x)\overline\phi_N(\tilde x) \\
		& \hspace{6cm} -(S^2-1) G_N(x,x)G_N(\tilde x, \tilde x) + G_N(x,\tilde x)^2\Bigr]\\
		&\geq  \frac12 \int \d x \,\d\tilde x\, v^\eps (x-\tilde x) \Bigl[ 2(S-1) G_N(x,x)|\phi_N(\tilde x)|^2 -(S^2-1) G_N(x,x)G_N(\tilde x, \tilde x)\\ 
		& \hspace{6cm}-2 \sqrt{G_N(x,x)G_N(\tilde x, \tilde x)}|\phi_N(x)||\phi_N(\tilde x) |  \Bigr]\\
		& \geq (S-2)\int \d x \, \d{\tilde x}\, v^\eps (x-\tilde x)G_N(x,x)|\phi_N(\tilde x)|^2 \\
		&\hspace{4cm} - \frac{(S^2-1)}{2}\int \d x\,\d {\tilde x}\, v^\eps (x-\tilde x)G_N(x,x)G_N(\tilde x, \tilde x)\,.
	\end{split}\end{equation}
	To conclude, on one hand we notice that the first term on the right-hand side is non-negative, by definition of $S$. On the other hand, using that $\ee^{-x} \lesssim x^{-\gamma}$ for any $x>1$, we have  that 
    \[
    \tr(G_N) = \int\d{x}\,\int_{1/N}^\infty\d{t}\, \ee^{-th}(x,x)  = \tr(h^{-1}\ee^{-h/N})\lesssim N^\gamma \tr(h^{-1-\gamma})\,.
    \]
   Recalling the condition in Remark \ref{rmk:thetha_s2}, in order to have $\tr(h^{-1-\gamma}) <\infty $  we require  $\gamma =s-1 $.   Hence, using that the integral of $v$ is one, as well as  \eqref{eq:boundG}, 
    we can bound the second term on the right-hand side of \eqref{eq:boundVNeps} as 
	\[
	\begin{split}
		&\frac{(S^2-1)}{2}\int \d x \,\d\tilde x\, v^\eps (x-\tilde x)G_N(x,x)G_N(\tilde x, \tilde x) \\
        &\hspace{1cm}\leq  \frac{(S^2-1)}{2}\log N\int \d x \, G_N(x,x) \lesssim \frac{(S^2-1)}{2}N^{\gamma}\log N\,,  \\
	\end{split}
	\]
    which concludes the proof.       
\end{proof}
The last key lemma to show Proposition \ref{prop:nelson} is an estimate for the $L^2$-norm of the difference $V_N^\eps - V_M^\eps$, which is the equivalent of \cite[Lemma 4.5]{FKSS_24}. Also in this case, it is important to use sharper estimates on the decay of the Green function.
\begin{lemma}\label{lm:VMVN}
	 For any $0< N \leq M \leq \infty$, and $\theta > 2$ we have
	\begin{equation}\label{eq:normVM-VN}
	\|V_M^\eps -V_N^\eps\|_{L^2} \lesssim N^{-\frac12+\frac1\theta +}\,.
	\end{equation}
\end{lemma}
\begin{remark} \label{rem:theta_condition}
The bound in \eqref{eq:normVM-VN} is essentially optimal. If we compare with the analogous estimate for the homogeneous case from \cite[Lemma 4.5]{FKSS_24}, the same $L^2$-norm was controlled by the stronger bound $N^{-1+}$ by regularising the Green function of $\kappa-\Delta/2$ with a smooth mollifier. The reason for our weaker bound is twofold. First, we use the exponential regularisation function $\vartheta(x)=\ee^{-x}$ of the Green function, which is very convenient throughout our analysis, but leads to an optimal bound $N^{-1/2+}$ even in the homogeneous case on the torus. Second, if we regularise $G$ by convolution with a compactly supported smooth function, as in \cite{{FKSS_24}}, we can apply the second-order expansion from \cite[Lemma B.3]{FKSS_24}, which yields the stronger bound $N^{-1+}$. However, this expansion requires $G$ to be locally $C^2$, which is in general not true under Assumption \ref{assumption:P and D for U}. Thus, unless we make an additional smoothness assumption on the potential $U$, we believe that Lemma \ref{lm:VMVN} cannot be improved.

Assuming a smooth enough potential $U$ and regularizing $G$ by convolution with a smooth compactly supported function, we believe that our condition $\theta > 10$ can be improved to $\theta > 4$.
\end{remark}

\begin{proof}
	Following \cite[Section  4.2]{FKSS_24}, we have that the difference $V_M^\eps -V_N^\eps$, for $N \leq M$ can be rewritten as 
	\begin{equation}\label{eq:VN-VM}
		V_M^\eps -V_N^\eps = \sum_{a,\tilde a, b,\tilde b \in \{0,1\}} \mathbbm{1}_{a+\tilde a+b +\tilde b >0} V_{N,M}^\eps (a,\tilde a, b, \tilde b)
	\end{equation}
	where 
	\[\begin{split}
		V_{N,M}^\eps(a,\tilde a, b, \tilde b) &\deq\frac12 \int \d x\d\tilde x\, v^\eps(x-\tilde x)\vcentcolon\phi_N(x)^{1-a}\phi_N(\tilde x)^{1-\tilde a} \overline{\phi}_N(x)^{1-b} \overline\phi_N(\tilde x)^{1-\tilde b} \\
		&\hspace{5cm} \times \psi_{N,M}(x)^a\psi_{N,M}(\tilde x)^{\tilde a}\overline\psi_{N,M}(x)^b\overline\psi_{N,M}(\tilde x)^{\tilde b}\vcentcolon\,.\\
	\end{split}	\]
	Hence, by \eqref{eq:VN-VM} and Minkowski's inequality, it suffices to estimate
	\(
	 \mathbb E[|V_{N,M}^\eps(a,\tilde a, b, \tilde b)|^2] 
	\)
	for any fixed $a,\tilde a, b,\tilde b \in \{0,1\}$ satisfying $a+\tilde a+b+\tilde b >0$. Using Proposition \ref{prop:GN}, we have 
    \[\begin{split} |G_M(x,y) - G_N (x,y)| \lesssim q_N(x,y)\deq \Bigl[\log(N|x-y|^2)^{-1}\Bigr]_+  \wedge \frac{\ee^{-N|x-y|^2}}{N|x-y|^2}\,. 
    \end{split}\]
    Then using Wick's theorem and Young's inequality we obtain
\begin{equation}\label{eq:meanVNM}
\begin{split}
   \mathbb E\bigl[|V_{N,M}^\eps(a,\tilde a, b, \tilde b)|^2\bigr]  & = \int \d{x}\,\d{\tilde x}\,\d{y}\,\d{\tilde y} \,v^\eps(x-\tilde x)v^\eps(y-\tilde y)G_N(x,y)G_N(\tilde x,\tilde y)\\
   &\hspace{1.5cm} \times \Bigl[\bigl(G_M(\tilde y,x)-G_N(\tilde y,x)\bigr)\bigl(G_M(y,\tilde x)-G_N( y,\tilde x)\bigr) \\
   &\hspace{2cm}  + \bigl(G_M( y,x)-G_N( y,x)\bigr)\bigl(G_M(\tilde y,\tilde x)-G_N( \tilde y,\tilde x)\bigr)\Bigr] \\
   &\lesssim \int \d{x}\,\d{\tilde x}\,\d{y}\,\d{\tilde y}\, v^\eps(x-\tilde x)v^\eps(y-\tilde y)G_N(x,y)G_N(\tilde x,\tilde y)\\
   &\hspace{2.5cm} \times \Bigl[q_N(\tilde y,x)q_N(y,\tilde x) + q_N(y,x)q_N(\tilde y,\tilde x)\Bigr]\\ 
   & \lesssim \int \d{x}\,\d{\tilde x}\,\d{y}\,\d{\tilde y} \,|v^\eps(x-\tilde x)||v^\eps(y-\tilde y)|G_N(x,y)^2q_N(y,x)^2 \\ 
   &\quad + \int \d{x}\,\d{\tilde x}\,\d{y}\,\d{\tilde y} \,|v^\eps(x-\tilde x)||v^\eps(y-\tilde y)|G_N(x,y)^2q_N(\tilde y,x)^2 \\
   & \lesssim \int \d{x}\,\d{y}\, G_N(x,y)^2q_N(y,x)^2 \\
   &\quad+ \int \d{x}\,\d{y}\,\d{\tilde y}\,|v^\eps(y-\tilde y)| G_N(x,y)^2q_N(\tilde y,x)^2\,.
\end{split}
\end{equation}
Using the bound on $q_N$ and Proposition \ref{prop:GN}, the first term on the right-hand side of \eqref{eq:meanVNM} can be bounded as 
\[
\begin{split}
    &\int \d{x}\,\d{y} \, G_N(x,y)q_N(y,x) \\ 
    &\; \lesssim \int \d{x}\,\d{y}\, \Bigl[\Bigl(\Bigl[\log(|x-y|\sqrt{\tilde g(x)\vee \tilde g(y)})^{-1}\Bigr]_+ \wedge \log N\Bigr)\vee 1\Bigr] \ee^{-|x-y|\sqrt{\tilde g(x)\vee \tilde g(y)}}  \\
    & \hspace{8cm} \times \Bigl[\Bigl[\log(N|x-y|^2)^{-1}\Bigr]_+  \wedge \frac{\ee^{-N|x-y|^2}}{N|x-y|^2}\Bigr]\\
    & \; \lesssim \int_{|x-y|\leq \frac{1}{ \sqrt N}} \d{x}\,\d{y}\, \Bigl[\Bigl(\Bigl[-\log(|x-y|\sqrt{\tilde g(x)})\Bigr]_+ \wedge \log N\Bigr)\vee 1\Bigr] \ee^{-|x-y|\sqrt{\tilde g(x)}}\\
    &\hspace{9cm} \times   \Bigl[-\log(N|x-y|^2)\Bigr]_+ \\
    & \; + \int_{|x-y|> \frac{1}{
     \sqrt N}} \d{x}\,\d{y}\, \Bigl[\Bigl(\Bigl[-\log(|x-y|\sqrt{\tilde g(x)})\Bigr]_+ \wedge \log N\Bigr)\vee 1\Bigr]\ee^{-|x-y|\sqrt{\tilde g(x)}}  \frac{\ee^{-N|x-y|^2}}{N|x-y|^2} \\
     &\quad = \rm{(A)} +\rm{(B)}\,.
\end{split}
\]
For the term $\rm (A)$ we perform a change of variable setting $|x-y|\sqrt{\tilde g(x)} =|u|$, and we distinguish the cases when $\tilde g(x)/N <1$ and $\tilde g(x)/N\geq 1$. When $\tilde g(x)/N <1$ we have
\[\begin{split}
\rm{(A)} & = \int_{|u|\leq \sqrt{\frac{\tilde g(x)}{ N}}} \d{x}\,\d{y}\,\frac{1}{\tilde g(x)} \Bigl[\Bigl(\Bigl[-\log(|u|)\Bigr]_+ \wedge \log N\Bigr)\vee 1\Bigr] \ee^{-|u|}   \Bigl[-\log(\sqrt{Ng(x)^{-1}}|u|)\Bigr]_+ \\
&\lesssim \frac{\log N}N \int_{|x|\leq N^{\frac{1}{\theta}}} \d{x} \lesssim N^{-1+ \frac{2}{\theta} +}\,.
	\end{split}\]
On the other hand if $\tilde g(x)\geq N$, we can use the exponential decay and the fact that $\tilde g(x) \gtrsim  1 +|x|^\theta$ to have 
\[
{\rm (A)} \lesssim \int_{|x|\geq N^{\frac1\theta}} \d{x}\, \frac{1}{\tilde g(x)}  (\log N)^2 \lesssim  (\log N)^2\int_{ N^{\frac1\theta}}^\infty \d{r}\, r^{1-\theta} \lesssim N^{\frac{2}{\theta}-1 + }\,.
\]
As for the term $\rm (B)$, using $\tilde g(x) \geq 1$ and $|x-y| >1/\sqrt N$ we have
\[
\begin{split}
    {\rm (B) }&\lesssim \log N \int_{|x-y|> \frac{1}{
     \sqrt N}} \d{x}\d{y}\, \ee^{-|x-y|\sqrt{\tilde g(x)}} \frac{\ee^{-N|x-y|^2}}{N|x-y|^2} \\
    &\lesssim \log N \int \d{x}\, \int_{\frac{1}{ \sqrt N}}^\infty \d{r}\, r \ee^{-r\sqrt{\tilde g(x)}}\frac{\ee^{-r^2N}}{Nr^2}\\
    & \lesssim \log N \int_{\frac{1}{ \sqrt N}}^\infty \d{r}\, r \frac{\ee^{-r^2N}}{Nr^{2 +\frac{4}{\theta}}} \lesssim N^{-1+\frac{2}{\theta}+}\,,
\end{split}
\]
where in the third step we used again that $\tilde g(x) \gtrsim  1+|x|^\theta$.
The second term in the last line on the right-hand side of \eqref{eq:meanVNM} can be bounded in a similar way. 
\end{proof}

\begin{proof}[Proof of Proposition \ref{prop:nelson}]
A detailed proof of Proposition \ref{prop:nelson} can be found in \cite[Prop. 4.3]{FKSS_24}. For completeness we briefly explain the main changes. For any $N \geq 1$, and by Lemma \ref{lm:boundVN} we have
\[
\mathbb{P}\bigl(\ee^{-V^\eps}>t \bigr) = \mathbb{P}\bigl(V^\eps - V_N^\eps < -\log t -V^\eps_N\bigr) \leq \mathbb{P}\bigl(V^\eps - V_N^\eps < -(\log t - CN^\gamma \log N)\bigr)\,.
\]
   Choosing $N \geq 1$ so that 
   \[
   \log t - CN^\gamma \log N =1
   \]
   for $t$ large enough. Using the fact that $V^\eps_M$ is in the 4th polynomial chaos combined with hypercontractivity (see e.g.\ \cite[Appendix A]{FKSS_24}) and Lemma \ref{lm:VMVN}, we have for any $0< N \leq M <\infty$ and $p \in 2\bN$
   \[
   \|V_M^\eps -V_N^\eps\|_{L^p} \lesssim \frac{p^2}{N^{\frac12 -\frac1\theta-}}
   \]
   which holds also for $V^\eps_M$ replaced by $V^\eps$ by Lemma \ref{lm:Veps-V}. Hence, by Chebyshev's inequality, for any $p \in 2\bN$ 
   \[\begin{split}
   \mathbb{P}\bigl(\ee^{-V^\eps}>t \bigr) \leq \mathbb{E}\bigl[|V^\eps -V^\eps_N|^p] \leq \biggl(\frac{Cp^2}{N^{\frac12-\frac1\theta-}} \biggr)^p &\leq \Bigl(Cp^2(\log t)^{(-\frac{1}{\gamma}+)(\frac12-\frac1\theta)}\Bigr)^p \\
   & \leq \ee^{-c (\log t)^{(-\frac{1}{2\gamma}+)(\frac12-\frac1\theta)}}
   \end{split}\]
   where in the last step we optimised over $p \in 2\mathbb{N}$. For $\gamma < \frac14 - \frac1{2\theta}$ and $t$ large enough the statement of Proposition \ref{prop:nelson} holds.  From \eqref{eq:tracehs} this restriction on $\gamma$ implies an upper bound on $s< \frac{5}{4} -\frac{1}{2\theta} $. Together with the lower bound $s> 1 +\frac2\theta$ in Remark \ref{rmk:thetha_s2} we obtain the condition $\theta >10$.   
\end{proof}
Finally, the integrability of $\ee^{-V}$ follows from Proposition \ref{prop:nelson} and Lemma \ref{lm:Veps-V}. To conclude this subsection, we mention that the proof of the first statement in Proposition \ref{prop:convergence classical field theories}, namely the convergence of the partition functions $\zeta^{W^\eps} \to \zeta$, follows from the following result, whose proof can be found in \cite[Lemma 4.6]{FKSS_24}.
\begin{lemma}
    For any $1 \leq p <\infty$ we have $\|\ee^{-V} - \ee^{-W^\eps}\|_{L^p}\to 0$ as $\eps \to 0$. 
\end{lemma}

We conclude this subsection, by noting that the results above also imply the integrability of the field $W^\eps_\eta$ \eqref{eq:def:regularised version of the field Weps}. Indeed proceeding similarly as for $W^\eps$, we need to introduce a field $V^\eps_\eta$, analogously as $V^\eps$, 
\[
V^\eps_\eta \deq \frac12 \int \, \d x \,\d {\tilde x} \,\varphi(\eta x)v^\eps (x-\tilde x)\varphi(\eta\tilde{x}) \, \wick{|\phi(x)|^2 |\phi(\tilde x)|^2}
\]
and prove the estimates for the $L^2(\mathbb{P})$ norm of $V^\eps_\eta - W^\eps_\eta$. For conciseness, we omit the details here.

  \subsection{Correlation functions}
In this subsection, we prove the second claim of Proposition \ref{prop:convergence classical field theories}.  The arguments are analogous to those of \cite[Section 4.4]{FKSS_24}; the main difference is the control of the large scale behaviour of the correlation functions, which we obtain through a representation derived from repeated Gaussian integration by parts. In particular, in our results we can show the local uniform convergence of the Wick-ordered correlation functions with respect to a decaying weight function, which is a stronger result than the desired convergence in $L^1 \cap L^\infty$.

To define the weight function, let \(y_1, \dots , y_\ell \in \mathbb{R}^d\) for \(\ell\geq 1\) and consider the graph of vertices \(\{y_i\}_{i=1}^\ell\) that are connected when \(\abs{y_i-y_j}\leq 1\). Let \(\sim\) be the equivalence relation between two vertices \(y_1, y_2 \) such that \(y_1 \sim y_2\) when \(y_1\) and \(y_2\) belong to the same connected component. For  \(\mathbf{y}\in \mathbb{R}^{d\ell}\), $\theta >0$, we define the weight function  
\[\Upsilon_{\theta}(\mathbf{y})\deq\prod_{ x \in \{y_i\}_i/\sim}\frac{1}{1 + \abs{x}^\theta}\,, \] and for $\mathbf{x}, \tilde{\mathbf{x}} \in \bR^{dp}$ the symmetrised weight  
\[ \Upsilon_{\theta}(\mathbf{x}, \tilde{\mathbf{x}})\deq\sum_{\pi_1, \pi_2, \pi_3 \in S_p} \prod_{i=1}^p \Upsilon_{\theta}(x_i, x_{\pi_1(i)}, \tilde{x}_{\pi_2(i)}, \tilde{x}_{\pi_3(i)})\,.\] 
It can be checked that, for any \(\theta\) satisfying the condition in Remark \ref{rmk:thetha_s2}, we have \(\Upsilon_{\theta}\in L^1( \mathbb{R}^{2dp})\cap L^\infty( \mathbb{R}^{2dp})\).

The claim of  \eqref{eq:diff gamma-gammaW with decay} in Proposition \ref{prop:convergence classical field theories} is now a consequence of the following result.
    \begin{prop} \label{prop:corrfunctionwUpsilon}
    Let $d=2$. Under the assumptions of Theorem \ref{thm:main_result}, for any \(p \in \mathbb{N}\), $\mathbf{x}, \tilde{ \mathbf{x}} \in \mathbb R^d$  and \(c>0\), we have 
\begin{equation}\label{eq:decaywUpsilon}
\lim_{\substack{\eps\to 0} } \sup _{\mathbf{x}, \tilde{\mathbf{x}}} \frac{|(\widehat{\gamma}_p)_{\mathbf{x}, \tilde{\mathbf{x}}}-(\widehat{\gamma}_p^{W^\eps})_{\mathbf{x}, \tilde{\mathbf{x}}}|}{\Upsilon_{\theta-c}(\mathbf{x}, \tilde{\mathbf{x}})}=0\,.
\end{equation}
\end{prop}

We now prove Proposition \ref{prop:corrfunctionwUpsilon}. To this end, we start by considering the probability space consisting of elements \(X=(X_k)_{k \in \mathbb{N}}\) with \(X_k \in \mathbb{C}.\) We define the first order differential operators 
\begin{equation}
    L_{N, x}\deq\sum_{k \in \mathbb{N}} \frac{1}{\sqrt{\lambda_k}} \sqrt{\vartheta(\delta \lambda_k )} u_k(x) \partial_{\bar{X}_k}, \quad \bar{L}_{N, x}\deq\sum_{k \in \mathbb{N}} \frac{1}{\sqrt{\lambda_k}} \sqrt{\vartheta(\delta \lambda_k )} \bar{u}_k(x) \partial_{X_k} .
\end{equation}
over the space of random variables \(f(X)\), where \(f\) is smooth in the sense that all its partial derivatives exist. These operators satisfy $L_{N, x}\phi_{N}(y)=0$ and $L_{N,x}\bar{\phi}_{N}(y)=G_{N}(x,y)$
and similarly for \(\bar{L}_{\delta, x}\). In the sense of distributions, we also have  $L_{N, x}\phi(y)=0$ and $L_{N,x}\bar{\phi}(y)=G_{2N}(x,y)$,
where we use the regularisation function \(\vartheta(x) \deq \ee^{-x}\).

Using these differential operators, we obtain the following representation for the Wick ordered correlation functions.

\begin{prop}\label{prop:DifferentialRepgammap} We have
    \begin{equation*}
        (\widehat{\gamma}_{N, p})_{\mathbf{x}, \tilde{\mathbf{x}}}
=\frac{1}{\zeta} \mathbb{E}\Bigl[\bar{L}_{N, \tilde{x}_1} \cdots \bar{L}_{N, \tilde{x}_p} L_{N, x_1} \cdots L_{N, x_p} \mathrm{e}^{-V}\Bigr]\,.
    \end{equation*}
    The same statement holds for \(V\) replaced with \(W^\eps\). 
\end{prop}
Using this representation, we can show the following statement. 
\begin{prop}\label{prop:convgammaVgammaWeps}Uniformly in \(\mathbf{x},\mathbf{\tilde{x}} \in \mathbb{R}^{2p} \),
    \begin{equation}
        \lim_{\substack{\eps\to 0} } (\widehat{\gamma}_p^{W^\eps})_{\mathbf{x}, \tilde{\mathbf{x}}}= (\widehat{\gamma}_p)_{\mathbf{x}, \tilde{\mathbf{x}}}\,.
    \end{equation}
\end{prop}
\begin{proof}
    This statement follows from \cite[Lemma 4.10 and proof of (3.4)]{FKSS_24}, which we can adapt to our setting using techniques from Lemmas \ref{lm:Veps-Weps} and \ref{lm:Veps-V}. 
\end{proof}
The proof of \eqref{eq:diff gamma-gammaW with decay} follows from the following statement. 
\begin{lemma}\label{lm:boundWickgammax}
    Uniformly in \(\eps>0\),
    \[\abs{(\widehat{\gamma}^{W^\eps}_p)_{\mathbf{x}, \tilde{\mathbf{x}}}} \lesssim_v \Upsilon_{\theta}(\mathbf{x}, \tilde{\mathbf{x}})\, .
          \]
\end{lemma}
\begin{proof}[Proof of \eqref{eq:diff gamma-gammaW with decay}] We can reproduce here the proof of \cite[Theorem 1.12]{FKSS_2020}.
    From Proposition \ref{prop:convgammaVgammaWeps} and Lemma \ref{lm:boundWickgammax}, we have \(\abs{(\widehat{\gamma}_p)_{\mathbf{x}, \tilde{\mathbf{x}}}} \leq C  \Upsilon_{\theta}(\mathbf{x}, \tilde{\mathbf{x}}) \). For any \(c,\delta>0\), we can choose \(R\equiv R(\delta)\) such that 
    \[\sup _{\abs{\mathbf{x}}+ \abs{\tilde{\mathbf{x}}}> R} \frac{\big|(\widehat{\gamma}_p)_{\mathbf{x}, \tilde{\mathbf{x}}}-(\widehat{\gamma}_p^{W^\eps})_{\mathbf{x}, \tilde{\mathbf{x}}}\big|}{\Upsilon_{\theta-c}(\mathbf{x}, \tilde{\mathbf{x}})}\leq \sup _{\abs{\mathbf{x}}+ \abs{\tilde{\mathbf{x}}}> R} \frac{C\Upsilon_{\theta}(\mathbf{x}, \tilde{\mathbf{x}})}{\Upsilon_{\theta-c}(\mathbf{x}, \tilde{\mathbf{x}})}\leq C R^{-c\theta}\leq \delta\,.\]
    By Proposition \ref{prop:convgammaVgammaWeps}, we then have 
    \[\lim_{\eps\to 0} \sup _{\abs{\mathbf{x}}+ \abs{\tilde{\mathbf{x}}}\leq R} \frac{\big|(\widehat{\gamma}_p)_{\mathbf{x}, \tilde{\mathbf{x}}}-(\widehat{\gamma}_p^{W^\eps})_{\mathbf{x}, \tilde{\mathbf{x}}}\big|}{\Upsilon_{\theta-c}(\mathbf{x}, \tilde{\mathbf{x}})} =0\,. \qedhere\]
    \end{proof}
\begin{proof}[Proof of Proposition \ref{lm:boundWickgammax}] Let \(g(x) =1+\abs{x}^\theta \) chosen as in Assumption \ref{assumption:P and D for U}. 
    Assume that \(\eps>0 \) small.  Let \(\ell\geq 0\) and write 
    \begin{equation*}
\mathcal{L}_{N, \mathbf{z}}\deq L_{N, z_1}^{\#} \cdots L_{N, z_{\ell}}^{\#}\,,
\end{equation*}
for \(\mathbf{z}=(z_1, \ldots, z_{\ell}) \in \mathbb{R}^{2\ell}\), where \(L^{\#}\) stands for either \(L\) or \(\bar{L}\). 

    Applying the Leibniz rule to the \(2p\) derivatives \(L_{N, x_i}\) and \(\bar{L}_{N, \tilde{x}_i}\), we estimate \(\abs{(\widehat{\gamma}^{W^\eps}_p)_{\mathbf{x}, \tilde{\mathbf{x}}}} \) by a sum of finite number of terms of the form
    \[\sup_{\mathbf{z}_1, \dots, \mathbf{z}_k}\mathbb{E}\Big|\prod_{i=i}^k (\mathcal{L}_{N, \mathbf{z}_i}{W^\eps}) \ee^{-{W^\eps}}\Big| \]
    where \(\lceil p/2\rceil\leq k\leq 2p\) and, for all \(i=1, \cdots, k\), \(\mathbf{z}_i \in \mathbb{R}^{2\ell_i}\) for some \(\ell_i\leq 4 \). Notice that we have \(\sum_{i}\ell_i\leq 2p\). 
    For \(1\leq \ell\leq4\), using Hölder's inequality and hypercontractivity, we only need to show \begin{equation}\label{eq:boundLny}
        \|\mathcal{L}_{N, \mathbf{y}}{W^\eps}\|_{L^2}^2\lesssim \Upsilon_{\theta}(\mathbf{y})\,.
    \end{equation}

    Let us look at first order derivatives, that is, when \(\ell=1\). 
    We start by considering 
    \begin{multline}\label{eq:LWeps}
        L_{N,y}W^\eps = \frac{1}{2}\int \d u\, \d {\tilde{u}}\, v^\eps(u-{\tilde{u}}) G_{2N}(u, y) \phi(u)(\phi({\tilde{u}}) \phi({\tilde{u}})-G({\tilde{u}}, {\tilde{u}}))\\+\frac{1}{2}\int \d u \,\d {\tilde{u}}\,v^\eps(u-{\tilde{u}})(\phi(u) \phi(u)-G(u, u)) G_{2N}(y, {\tilde{u}}) \phi({\tilde{u}})\\-\int \d u\, \d {\tilde{u}}\,v^\eps(u-{\tilde{u}}) G(u, {\tilde{u}}) G_{2N}(u, y) \phi(u)
    \end{multline}
    and we can proceed in the same way for \(\bar{L}_{N,y}W^\eps\).  We can then estimate 
    \begin{align*}    \mathbb{E}\bigl[({L_{N,y}W^\eps})^2\bigr]& \lesssim \frac{1}{\tilde g(y)^{2}} 
    \end{align*}
    by using Proposition \ref{prop:GN} and considering separately the cases where the integration variables \(\hat u\) are close (\(\abs{\hat u-y}\leq1\)) or far (\(\abs{\hat u -y }>1\)) from \(y\). 
    A direct computation shows that \( \mathbb{E}\bigl[({L_{N,y}W^\eps})^2\bigr]\) is a sum of terms that can be grouped as 
    \begin{equation}\label{eq:nodiagonalterm}
        \int \d {u_1} \, \d{\tilde{u}_1} \, \d {u_2} \, \d{\tilde{u}_2}\,v^\eps(u_1- \tilde{u}_1)v^\eps(u_2- \tilde{u}_2)G_{2N}(a_1,y)G_{2N}(a_2,y)G(a_3,a_4)G(a_5,a_6)G(a_7,a_8)
    \end{equation}
    where \(\{a_{i}\}_{i=1}^8= \{u_1, \tilde{u}_1,u_2, \tilde{u}_2\}\) such that \(\{a_{2i-1}, a_{2i}\}\notin\{\{u_1, \tilde{u}_1\},\{u_2, \tilde{u}_2\}\} \) and \(a_{2i-1}\neq a_{2i}\) for all \(i=1, \dots, 4\)
    or   \begin{multline}\label{eq:diffinfterms}
        \int \d {u_1}\, \d{\tilde{u}_1}\,\d {u_2}\, \d{\tilde{u}_2}\, v^\eps(u_1- \tilde{u}_1)v^\eps(u_2- \tilde{u}_2)G_{2N}(a_1,y)G_{2N}(a_2,y)\\\times G(u_1,\tilde{u}_1)G(u_2,\tilde{u}_2)(G(a_3,a_2)- G(a_3, a_4))
    \end{multline}
    where \(a_i\in \{u_1, \tilde{u}_1,u_2, \tilde{u}_2\}\) such that \(\{a_1, a_{2}\}\notin\{\{u_1, \tilde{u}_1\},\{u_2, \tilde{u}_2\}\} \), \(a_3\notin \{a_2, a_4\}\), and \(\{a_2, a_4\}\in \{\{u_1, \tilde{u}_1\},\{u_2, \tilde{u}_2\}\}\). 
    
    We can deal with \eqref{eq:nodiagonalterm} using Young's inequality for products and estimate 
    \begin{multline}
        {G_{2N}(a_2,y)}G(a_3,a_4)G(a_5,a_6)G(a_7,a_8) \\\leq\frac{1}{4}\bigl( G(a_2,y)^4 + G(a_3,a_4)^4 +G(a_5,a_6)^4 +G(a_7,a_8)^4 \bigr)\,
    \end{multline}
    and split domains based on whether \(a_1\) or \(a_2\) belong to \(B_1(y)\). Assume without loss of generality that \(a_1 = u_1\) and \(a_2=u_2\) and note that \( \{a_{2i-1},a_{2i}\}\cap\{u_2, \tilde{u}_2\} \neq \emptyset \).  Since \(v^\eps\) has support in a ball of radius \(\eps\) around the origin, then \(\tilde g({u}_i)\sim \tilde g(\tilde{u}_i)\). 
    In the case where \(a_1,a_2\in B_1(y)\), we have that 
    \begin{equation*}
        \int_{B_1(y)} \d{ u_2}\int \d{\tilde{u}_2}\,  v^\eps(u_2-\tilde{u}_2)G(s,t)^n \leq \int_{B_1(y)} \d{ u_2}\int \d{\tilde{u}_2}\,  v^\eps(u_2-\tilde{u}_2)\log^n(\abs{s-t}\sqrt{\tilde g(y)})^{-1} \lesssim\frac{1}{\tilde g(y)}
    \end{equation*}
    for any \(s\in \{y,u_1, \tilde{u}_1\}\), \(t\in\{u_2, \tilde{u}_2\}\) and \(n\in \mathbb{N}_{>0}\). This implies 
    \begin{multline*}
        \int_{_{\substack{u_1 \in B_1(y)\\
    u_2 \in B_1(y)}}} \d{u_1} \,\d{\tilde{u}_1}\,\d{u_2}\, \d{\tilde{u}_2}\, v^\eps (u_2 -{\tilde{u}_2} )v^\eps (u_1-{\tilde{u}_1} )  G(y,u_1)G(s, t)^n \\ \lesssim \frac{1}{\tilde g(y)}\int_{u_1 \in B_1(y)}\d{u_1}\,G(y,u_1)\lesssim \frac{1}{\tilde g(y)^2}
    \end{multline*}
    for any \(n>0\), and using the fact that \(\widehat{v^\eps}(0)=1\). The remaining cases can be treated analogously, using the fact that
    \begin{equation*}
        \int_{B_1(y)^c} \d{ u_2}\int \d{\tilde{u}_2}\,  v^\eps(u_2-\tilde{u}_2)G(s,t)^n \lesssim \int_{B_1(y)^c} \d{ u_2}\int \d{\tilde{u}_2}\,  v^\eps(u_2-\tilde{u}_2)\ee^{-\sqrt{\tilde g(y)}\abs{s-t}} \lesssim\frac{1}{\tilde g(y)}\,. 
    \end{equation*}
    The term in \eqref{eq:diffinfterms} requires a more detailed analysis, as it contains an integration over the diagonal of \(G\) and \(\int \d {y} \, v^\eps(y)G(x,x+y)\sim \log(\eps^{-1})\). Let us assume that \(a_2=u_1, a_4 = \tilde{u}_1\). Performing a change of variables \(z=(a_4-a_2)/\eps\) 
    \begin{equation}\label{eq:secondtermgammaWeps}
        \begin{aligned}
      \eqref{eq:diffinfterms} &\leq \int_0^1 \d{\lambda}\int \d {u_1} \,\d{\tilde{u}_1}\,\d {u_2}\, \d{\tilde{u}_2}\, v(z)v^\eps(u_2- \tilde{u}_2)G(a_1,y)G(u_1,y)\\ & \qquad \qquad\qquad \qquad\qquad \qquad\times G(u_1,\tilde{u}_1)G(u_2,\tilde{u}_2)\frac{\eps\abs{z}\ee^{-\sqrt{\tilde g(a_3)}\abs{a_3-(u_1+ (1-\lambda)\eps\abs{z})}}}{\abs{a_3-(u_1+ (1-\lambda)\eps{z})}}\\
             &\leq \int_0^1 \d{\lambda}\int \d {u_1} \,\d{\tilde{u}_1}\,\d {u_2}\, \d{\tilde{u}_2}\, v(z)v^\eps(u_2- \tilde{u}_2)G(a_1,y)G(u_1,y)\\ & \qquad \qquad\qquad \qquad\qquad \qquad\times G(u_2,\tilde{u}_2)\frac{\eps^{1-c}\abs{z}^{1-c}\ee^{-\sqrt{\tilde g(a_3)}\abs{a_3-(u_1+ (1-\lambda)\eps\abs{z})}}}{\abs{a_3-(u_1+ (1-\lambda)\eps{z})}}
        \end{aligned}
    \end{equation}
    for any \(c>0\). By splitting cases and using the fact that \(G(y,u_1)G(y,a_1)\lesssim \frac{1}{2}(G(y,u_1)^2 + G(y,a_1)^2)\), we obtain 
    \[\eqref{eq:secondtermgammaWeps} \lesssim \eps^{1-c}\frac{1}{\tilde g(y)^2}\]
    for any \(c>0\). 
    
    We can follow the arguments above for \(\ell=2\) and we obtain
    \begin{align*}
        &\|\bar{L}_{N, y_2}L_{N, y_1}{W^\eps}\|_{L^2} \lesssim_v  \Upsilon_{\theta}(y_1,y_2)
    \end{align*}
    by considering the cases \(|y_1-y_2|\leq 2\) and \(|y_1-y_2|>2\) separately. 

    In order to show
    \[\|{L}_{N, y_2}L_{N, y_1}{W^\eps}\|_{L^2} \vee \|\bar{L}_{N, y_2}\bar{L}_{N, y_1}{W^\eps}\|_{L^2}\lesssim   \Upsilon_\theta(y_1,y_2)\]
    and the corresponding results when \(\ell>2\), it is only necessary to bound terms of the form \eqref{eq:diffinfterms}, and we obtain \eqref{eq:boundLny} by proceeding as in \eqref{eq:nodiagonalterm}. 
\end{proof}
\section{Functional integral representations}
\label{sec:funcintrep}
In this section we derive the functional integral representation underlying the proof of Proposition \ref{prop:convergence quantum many body system and Weps} (see also the overview in Section \ref{sec:strategy}). In the following two propositions, Propositions \ref{Prop:regularisation of euclidean field theory} and \ref{Prop:regularisation of phi42 quantum rel part funct}, we give the functional integral representation of the (quantum)
relative partition function $\cal Z$ from  \eqref{eq:mf quantum relatve partition function} and the (classical) relative partition function $\zeta^{W^\eps}$ corresponding
to the interaction $W^\eps$.

Before  starting, we recall the auxiliary fields used in \cite{FKSS_2020} to derive the functional integral representation.
Let us fix a function $\varphi \in \mathcal{C}^\infty_c(\bR)$ which is even, non-negative, of positive type, and which satisfies $\varphi(0)=1$.  For given $\eta >0$, we define  a positive definite function on the interval \( [0,\nu)\) as 
\begin{equation}\label{eq:defdelta}
    \delta_{\eta,\nu}(x) \deq  \frac{1}{\eta} \sum_{y \in \mathbb{Z}} \check \varphi \Bigl(\frac{x-\nu y}{\eta}\Bigr)\,,
\end{equation}
such that \(\int \d x \,\delta_{\eta, \nu}(x)=1\), where $\check \varphi$ denotes the inverse Fourier transform of $\varphi$. The function $\delta_{\eta,\nu}$ can be interpreted as an approximate delta function on $[0,\nu)$.
For $\eta >0$, we introduce a centred real Gaussian field \(\sigma\col [0,\nu)\times \mathbb{R}^d \to \mathbb{R}\) with covariance kernel \[C_{\eta}(\tau, \tilde{\tau};x, \tilde{x})\deq \delta_{\nu, \eta}(\tau-\tilde{\tau})\varphi(\eta x)v^\eps(x-\tilde{x})\varphi(\eta \tilde{x})\]
whose law we denote with \(\mu_{C_\eta}\). We denote its time average as
\[\xi(x)\deq \frac{1}{\nu}\int_0^\nu\d{\tau}\,\sigma(\tau, x)\]
whose law is \(\mu_{v^\eps_\eta}\), a centred Gaussian with covariance kernel
\[v^\eps_\eta(x,\tilde{x}) = \varphi(\eta x)v^\eps(x-\tilde{x})\varphi(\eta \tilde{x})\,.\]
Note that the fact that \(v\) is smooth and the use of the spatial cutoff  through \(\varphi\) ensure that the field \(\xi\) is smooth and bounded almost surely for every \(\eps, \eta>0\). We need the following assumption on a family of functions $(\alpha^\eps)_{\eps >0}$ as introduced in Section \ref{sec:strategy}.

\begin{definition}
\label{assumption:existence of alpha}
    Let \(U\) satisfy Assumption \ref{assumption:P and D for U} and  \(v\) satisfy Assumption \ref{assumptions: non local interaction}. For each $\epsilon > 0$ we choose a nonnegative function \(\alpha^\eps \in L^\infty(\mathbb{R}^d)  \cap  L^2(\mathbb{R}^d) \) such that, uniformly for all \(\eta\geq0\) small enough,
    \begin{equation}\label{eq:lower bound on Uepseta}
        U^{\eps}_\eta(x)=(U  + \alpha^\eps_\eta- \tau^\eps)(x) \geq \frac{1}{2}U(x)\,,
    \end{equation}
    \begin{equation}\label{unif convergence of alpha_eta}
        \lim_{\eta \to 0} \|\alpha ^\eps_{\eta} -\alpha^\eps_0 \|_{L^\infty}= 0\,,
    \end{equation}
    and \begin{equation}\label{eq:bound on Linf norm of alpha}
        \|\alpha^{\eps}_\eta\|_{L^{\infty}} \vee \|\alpha^{\eps}\|_{L^{\infty}}\lesssim \|\tau^\eps\|_{L^\infty}\,.
    \end{equation}
Here we recall the definition of $\alpha_\eta^\eps$ from \eqref{def:rhoepseta}. 
\end{definition}
\begin{remark}
    A function $\alpha^\eps$ as in Definition \ref{assumption:existence of alpha} can be easily constructed. For instance, we can choose \(\alpha^\eps(x) \deq(\|\tau^\eps_\eta \|_{L^\infty}+1)  \mathds{1}_{\bar{B}_{r_\eps}}\), where \(B_{r}\) is a ball of radius \(r\) centred at 0 and \(r_\eps>0\) is chosen appropriately. 
\end{remark}
From Definition \ref{assumption:existence of alpha}, more precisely from \eqref{eq:lower bound on Uepseta}, it follows that $(-\Delta/2 +\kappa + U^\eps_\eta)^{-1} $ is Hilbert-Schmidt. Indeed, by the Feynman-Kac formula (see Lemma \ref{Feynman-Kac formula} below) we have
     \begin{equation}\label{eq: heat kernel of U tilde upper bound}
         \ee^{-t{h^U}}(x,y)\vee \ee^{-t{h^{U^\eps_\eta}}}(x,y)<\ee^{-t{h^{ U/2}}}(x,y)
     \end{equation} for all \(x,y \in \mathbb{R}^d\), \(t>0\) and uniformly in \(\eps>0, \eta\geq 0\) small enough. Since, for \(h,r>0\)
     \[\frac{1}{h^{r}}=\frac{1}{\Gamma(r)} \int_0^\infty \d{t}\, {\ee^{-ht}} \, t^{r-1}\,,\]
     \eqref{eq: heat kernel of U tilde upper bound} implies that
     \begin{equation}
     \label{eq:uniform lower bound potential in eta and eps}
     \tr{\frac{1}{(h^U)^s}}\vee \tr{\frac{1}{(h^{U^\eps_\eta})^s}}\leq \tr{\frac{1}{(h^{U/2})^s}} \end{equation}
     which is finite whenever \(\tr{\bigl(h^{U}\bigr)^{-s}}<\infty\).

We have the following functional integral representation of the relative partition function $\zeta^{W^\eps}$.
\begin{prop}\label{Prop:regularisation of euclidean field theory}
Suppose that Assumption \ref{assumption:P and D for U} and Assumption \ref{assumptions: non local interaction} hold. Let \(\alpha^\eps: \mathbb{R}^d \to \mathbb{R} \) be the function as in Definition \ref{assumption:existence of alpha}.
    \begin{enumerate}[label = (\roman*)]
        \item   
               The partition function \(\zeta^{W^\eps}=\mathbb{E}\bigl[\ee^{-W^{\eps}}\bigr] \) is given by

            \begin{equation}\label{classical partition function in phi42 case}
                                \zeta^{W^\eps}= \lim_{\eta \to 0}\ee^{E^\eps + \frac{1}{2} \int \d{x}\,\alpha^\eps(x) \alpha^\eps_\eta(x) + S^\eps_\eta }\int \mu_{v^{\eps}_\eta}(\mathrm{d} \xi) \,\mathrm{e}^{-i\langle T^\eps_\eta+ \alpha^\eps,\xi\rangle} \mathrm{e}^{f_2(\xi)}\,,
            \end{equation}
              where we defined
              \begin{equation}\label{def:f_2}
                    f_2(\xi)\deq\int_0^{\infty} \mathrm{d} t \, \tr\biggl(\frac{1}{t+h^{U^{\eps}_{\eta}}- \mathrm{i} \xi}-\frac{1}{t+h^{U^{\eps}_{\eta}}}-\frac{1}{t+h^{U^{\eps}_{\eta}}} \mathrm{i} \xi \frac{1}{t+h^{U^{\eps}_{\eta}}}\biggr)\,,
              \end{equation}
              \begin{equation}
              \label{eq:S_eta}
                S_\eta^\eps = \log \mathbb{E}\biggl[{\exp\Bigl(-\int \d{x}\, (\alpha^\eps_\eta(x)- \tau^\eps(x))\,\wick{\abs{\phi(x)}^2}\Bigr)}\biggr]\,,
              \end{equation}
              and 
              \begin{equation}\label{def:E^Wick_eta}
                T^{\eps}_\eta(x)\deq \biggl(\frac{1}{\kappa- \Delta/2 + U} - \frac{1}{\kappa- \Delta/2 + U^\eps_\eta}\biggr)(x,x)\,.
              \end{equation}
              Here \(\xi\) is interpreted as a multiplication operator.
        \item The function \(f_2(\xi)\) is finite and has non-positive real part.
    \end{enumerate}
\end{prop}
We can obtain a similar expression for a regularised version of the quantum relative partition function \(\mathcal{Z}\) from  \eqref{eq:mf quantum relatve partition function}. In order to do it, we introduce the following definitions.

\begin{definition}
    Let $b$ an operator on \( L^2 ([0,\nu] \times \mathbb{R}^d)\) acting on functions $f \in L^2 ([0,\nu) \times \mathbb{R}^d)$ of time $\tau \in [0,\nu)$ and space $x \in \mathbb R^d$. We use the notation $b^{\tau, \tilde \tau}_{x,\tilde x}$ for the operator kernel of $b$. Thus 
    \[
    (bf)(\tau,x) = \int \d\tau\, \int \d x\, b^{\tau, \tilde \tau}_{x,\tilde x}f(\tau,x)\,.
    \]
\end{definition}
The following operator  plays an important role in the rest of the analysis.
\begin{definition}\label{def:operatorK}
    Let \(I \in \mathbb{R}\) be a closed bounded interval and \(V:I\times {\mathbb{R}^d} \to [0, \infty)\) a locally bounded function, and suppose that $\re V \leq -c$ for some positive constant $c$. On $L^2 ([0,\nu) \times \mathbb{R}^d)$ define the operator
    \[
    K(V) \deq \partial_\tau -\frac{\Delta}{2} - V(\tau)\,,
    \]
    with periodic boundary conditions in $[0,\nu)$.
\end{definition}
Moreover, for a bounded operator $b$ on $L^2 ([0,\nu) \times \mathbb{R}^d)$ with continuous operator kernel $b^{\tau, \tilde \tau}_{x,\tilde x}$, we define
\begin{equation*}
    \Tr (b) \deq \int_0^\nu \d \tau \int_{\mathbb R^d} \d x\, b^{\tau, \tilde \tau}_{x,\tilde x}\,;
\end{equation*}
see also \cite[Definition 3.10]{FKSS_2020_arxiv} for more details.

\begin{prop}\label{Prop:regularisation of phi42 quantum rel part funct}
Let $H$ be defined as in \eqref{eq:WickHamiltonian}, and $U$ satisfy Assumption \ref{assumption:P and D for U}, and $v^\eps$ Assumption \ref{assumptions: non local interaction}. Moreover, let $K$ be defined as in Definition \ref{def:operatorK}. Then the following holds for fixed $\eps, \nu>0$.
\begin{enumerate}[label = (\roman*)]
    \item  We have
    \begin{equation}\label{def:quantum_relative_partition_function}
        \mathcal{Z} = \lim_{\eta \to 0}\ee^{E^\eps + \frac{1}{2} \int \d{x}\,\alpha^\eps(x) \alpha^\eps_\eta(x) +{ S^\eps_{\nu,\eta}} }\int \mu_{C_\eta} (\d{\sigma})\, \ee^{-\ii \langle\sigma,\alpha^\eps+{T^\eps_{\nu,\eta}}\rangle_\nu} \ee^{F_2(\sigma) }\,,
    \end{equation}
where   we defined
\begin{multline}\label{eq:def_of_F_2}
        F_2(\sigma) \deq \int_{0}^\infty \d{t}\, \Tr\biggl(\frac{1}{t+K(-\kappa - U^\eps_\eta + \ii\sigma)}-\frac{1}{t+K(-\kappa - U^\eps_\eta  )} \\ - \frac{1}{t+K(-\kappa - U^\eps_\eta )}\ii\sigma \frac{1}{t+K(-\kappa - U^\eps_\eta )}\biggr)\,,
    \end{multline}
    \begin{equation}\label{eq:def_S_nueta}
        S^\eps_{\nu, \eta}\deq \int_{0}^\infty \d t \Tr\biggl(\frac{1}{t+K(-\kappa - U^\eps_\eta)}- \frac{1}{t+K(-\kappa - U) }\biggr) + \big\langle\varrho_\nu^U,\alpha^\eps_\eta - \tau^\eps\big\rangle\,,
    \end{equation}
    and \begin{equation}\label{eq:def_Ewick_nueta}
       T^\eps_{\nu,\eta}(x)\deq \varrho_\nu^U(x) -\varrho_\nu^{U^\eps_\eta}(x)\,.
    \end{equation}
    Here  \(\sigma\) is interpreted as a multiplication operator.
    \item The function \(F_2\) is finite and has non-positive real part.
    \end{enumerate}
\end{prop}

Notice that, from the representation in \cite[Lemma 5.4]{FKSS_2020}, it follows that \eqref{eq:def_of_F_2} is finite. From the discussion in Section \ref{sub:RiemannsecEstimates}, we can show that \eqref{eq:def_S_nueta} and \eqref{eq:def_Ewick_nueta} are also bounded quantities using a similar Brownian bridge representation.

From Propositions \ref{Prop:regularisation of euclidean field theory} and \ref{Prop:regularisation of phi42 quantum rel part funct}, we can establish functional integral representations for the Wick-ordered \(p\)-point correlation functions similar to \cite[Lemmas 5.18 and 5.19]{FKSS_24}. This shows \eqref{eq:convergence of relative partition function with decay}. For conciseness, we will not provide further details.

We split the proofs of Propositions \ref{Prop:regularisation of euclidean field theory} and \ref{Prop:regularisation of phi42 quantum rel part funct} in two subsections.

\subsection{Proof of Proposition \ref{Prop:regularisation of euclidean field theory}}
The aim of this subsection is to prove Proposition \ref{Prop:regularisation of euclidean field theory}, by deriving a functional integral representation for the classical field theory $W^\eps$. We first show that the constant \eqref{eq:S_eta} is uniformly bounded in \(\eta\) and that the function \eqref{def:E^Wick_eta} is well-defined, integrable and bounded. We use then these results to arrive at the proof of Proposition \ref{Prop:regularisation of euclidean field theory}, stating whenever we can exploit directly results from \cite[Proposition 4.1]{FKSS_2020}.

Throughout the subsection we use the regularised version of the field $W^\eps$, defined in \eqref{eq:def:regularised version of the field Weps}. The following lemma, whose proof is analogous to \cite[proof of Lemma 1.5]{frohlich2017gibbs}, ensures that the field $W^\eps_\eta$ is well defined. 
\begin{lemma}\label{lemma:W_N is a Cauchy sequence}
    Suppose that  Assumption \ref{assumption:P and D for U}, Assumption \ref{assumptions: non local interaction}, and Remark \ref{rmk:thetha_s2} hold. Let $W^\eps_{\eta,N}$ be defined as in \eqref{eq:def:regularised version of the field Weps}. Then, uniformly in \(\eta \geq 0\) and for every \(\eps>0\), the sequence \(\bigl(W^\eps_{\eta,N}\bigr)_{N >0} \)
      is a Cauchy sequence in \(\bigcap_{m \geqslant 1} L^m(\mathbb{P})\). We denote its limit by \(W^\eps_\eta\). Moreover, for all fixed \(\eps>0\), \[\lim_{\eta\to 0 }W^\eps_\eta = W^\eps_0\equiv W^\eps\]
      in \(\bigcap_{m \geqslant 1} L^m(\mathbb{P})\) . 
\end{lemma}

 In the following, we will use the notation
\[P_N \deq  \sum_{k=0}^\infty \mathbbold{1}_{\lambda_k \leq N} u_k u_k^*\,\]
for the projection onto a finite subspace of \(\mathcal{H}\) spanned by the eigenvectors of \(h\)
and define
    \[(b)_N\deq P_N b P_N \,\]
for any operator \(b \) acting on \(\mathcal{H}\). 

\begin{lemma}
    \label{lemma:integration of rho and tauepsilon} 
    Let \(\alpha^\eps_\eta\) be defined as in \eqref{def:rhoepseta}, where \(\alpha^\eps\) satisfies Definition \ref{assumption:existence of alpha}.
    Then, for a fixed \(\eps > 0 \), and for all \(\eta \geq 0\) small enough, we have
    \begin{equation}\label{eq:limNalpha}
    \lim_{N \to \infty}\mathbb{E}\Bigl[{\ee^{-\int \d{x}\,(\alpha^\eps_\eta(x)- \tau^\eps_\eta(x))\col\abs{\phi_N(x)}^2\col}}\Bigr] = \mathbb{E}\Bigl[{\ee^{-\int \d{x}\, (\alpha^\eps_\eta(x)- \tau^\eps_\eta(x))\col\abs{\phi(x)}^2\col}}\Bigr]< \infty\,.
    \end{equation}
\end{lemma}

\begin{proof} 

    The convergence of the limit in  \eqref{eq:limNalpha}
    follows from  Lemma \ref{lemma:bounds on taueps}, Definition \ref{assumption:existence of alpha} applied to \(f \deq \alpha^\eps_\eta - \tau^\epsilon\), and hypercontractivity (see \cite[Appendix A.2]{FKSS_24} for further details). The argument, as in \cite[lemma 4.6]{FKSS_24}, can be summarised in the following steps.
    \begin{enumerate}[labelindent=0pt,labelwidth=\widthof{aa},label=\arabic*.,itemindent=1em,leftmargin=!]
        \item
        \(I_N \deq\int \d{x}\, (\alpha^\eps_\eta(x)- \tau^\eps_\eta(x))\wick{\abs{\phi_N(x)}^2}\) and \(I\deq\int \d{x}\,(\alpha^\eps_\eta(x)- \tau^\eps_\eta(x))\wick{\abs{\phi(x)}^2}\) belong to the second polynomial chaos, using the fact that both quantities are orthogonal to the generators \(\wick{\phi(f_1)\cdots\phi(f_n)}\), \(f_i \in \mathcal{H}^{-1}\deq D(h ^{-1/2})\), of the \(n\)-th polynomial chaos whenever \(n \neq 2\). 
        This follows from the fact that 
        \[ \wick{\phi(f_1)\cdots \phi(f_n) } = \int \d{x_1} \cdots \, \d{x_n}\,f(x_1)\cdots f(x_n) \wick{\phi(x_1)\cdots \phi(x_n)}\]
        and that, when \(n \neq m\)
        \[\mathbb{E}\bigl[ \wick{\phi(x_1)\cdots \phi(x_m)} \, \wick{ \phi(x_1)\cdots \phi(x_n)}\bigr] = 0\,.\]
        \item We can thus show that 
\begin{equation}\label{boudedness_equiv}
            \|\ee^{-I_N}\|_{L^p} < \infty \iff \|\ee^{-I}\|_{L^p} < \infty\,.
        \end{equation}
        By step 1, we can estimate
        \begin{equation}\label{hypercontract_2}
            \big\|\ee^{I_N-I} -1\big\|_{L^p} \leq \sum_{k \geq 1} \frac{1}{k!}\big\|(I_N-I)^k\big\|_{L^p} =\sum_{k \geq 1} \frac{1}{k!}\big\|I_N-I\big\|_{L^{p k}}^k\,. 
        \end{equation}
        Now we use the hypercontractivity of the $k$ moments and that \(I_N- I\) is in the second polynomial chaos, as well as the Stirling approximation in the last inequality
        \begin{equation}\label{hypercontract_3}
            \eqref{hypercontract_2}  \leq \sum_{k \geq 1} \frac{1}{k!}(p k)^k\big\|I_N-I\big\|_{L^2}^k \leq \sum_{k \geq 1}\bigl(C p\|I_N-I\|_{L^2}\bigr)^k.
        \end{equation}
        Hence,  \eqref{hypercontract_3} converges for \(N\) large enough from an argument similar to the proof of Lemma \ref{lemma:W_N is a Cauchy sequence}. Indeed, the result  in Lemma \ref{lemma:W_N is a Cauchy sequence} can also be extended to show that \(\int \d x\, f(x)\, \wick{|\phi_N(x)|^2}\) is a Cauchy sequence in \(L^m\) for all \(m \in \mathbb{N}^*\) as long as \(\|f\|_{L^\infty}<\infty\). This proves \eqref{boudedness_equiv}. 
    \end{enumerate} 
    We can now conclude the proof by showing
    \[\mathbb{E}\biggl[{\exp\biggl({-\int \d{x}\, (\alpha^\eps_\eta(x)- \tau^\eps_\eta(x))\, \wick{\abs{\phi_N(x)}^2}}}\biggr)\biggr] < \infty\]
    uniformly in \(N\) for a given choice of regularisation function \(\vartheta\). 
    
    Let us consider the regularisation function \(\vartheta(x)\deq\mathbbold{1}_{x\leq 1}\). For any \(N>0\), we define
    \begin{equation}\label{eq:finite dimensional reg of phi}
        \phi_N = \sum_{k=0}^{\infty} X_k \sqrt{
\frac{\vartheta(\lambda_k/N)}{\lambda_k}}u_k=\sum_{k=0}^{\infty} X_k 
\frac{\mathbbold{1}_{\lambda_k\leq N}}{\sqrt{\lambda_k}}u_k = P_N \phi.
    \end{equation}
Then \(P_N\phi\) is a finite dimensional centred Gaussian vector with covariance \(P_N h^{-1}\).  As in the proof of \cite[Proposition 4.1]{FKSS_2020}, we have
    \begin{multline}\label{eq:trace_pho_tau}
        \mathbb{E}\Bigl[\ee^{-\int \d{x} \,(\alpha^\eps_\eta(x)- \tau^\eps_\eta(x))\col\abs{\phi_N(x)}^2\col}\Bigr] \\ = \exp \biggl\{\int_0^{\infty} \d t\, \tr\biggl({P_N\biggl(\frac{1}{t+h} \bigl(\alpha^\eps_\eta- \tau^\eps_\eta\bigr)_N \frac{1}{t+h} \bigl(\alpha^\eps_\eta- \tau^\eps_\eta\bigr)_N  \frac{1}{t+h+ \bigl(\alpha^\eps_\eta- \tau^\eps_\eta\bigr)_N } \biggr)}\biggr)\biggr\}\,.
    \end{multline}
    To conclude, we want to show that the right-hand side of~\eqref{eq:trace_pho_tau} is bounded.
    By Assumption \ref{assumption:P and D for U} and \eqref{eq:uniform lower bound potential in eta and eps}, using \(\|(\tau^\eps_\eta - \alpha^\eps_\eta)_N\|_{\mathfrak{S}^{\infty}} \leq \|\tau^\eps_\eta - \alpha^\eps_\eta \|_{L^{\infty}}\) and Hölder's inequality, we get
    \begin{align*}
        \abs{\log {\eqref{eq:trace_pho_tau}}} \leq  \int_0^{\infty} \mathrm{d} t\,\|\alpha^\eps_\eta- \tau^\eps_\eta\|_{L^{\infty}}^2\Big\|\frac{1}{t+{h^{U/2}}}\Big\|_{\mathfrak{S}^3}^3 & =\int_0^{\infty} \mathrm{d} t\,\|\alpha^\eps_\eta- \tau^\eps_\eta\|_{L^{\infty}}^2 \tr\Bigl\{\Bigl(\frac{1}{t+h^{U/2}}\Bigr)^3\Bigr\} \\ & =\frac{1}{2}\|\alpha^\eps_\eta- \tau^\eps_\eta\|_{L^{\infty}}^2 \tr\Bigl\{ \bigl({h^{U/2}}\bigr)^{-2}\Bigr\}<\infty\,. \qedhere
    \end{align*}
\end{proof}
Let us the define the operator $\mathcal{T}^{\eps}_\eta$ acting on \(\mathcal{H}\) as
    \begin{equation}\label{eq:def_of_Ewicketa_as_operator}
           \mathcal{T}^{\eps}_\eta \deq \frac{1}{h^U} -\frac{1}{h^{U^{\eps}_\eta}}\,.
        \end{equation}

\begin{lemma}\label{lemma:convergence_of_EWick}
    Let \(U\) satisfy Assumption \ref{assumption:P and D for U} and \(\alpha^\eps\) as in  Definition \ref{assumption:existence of alpha}. Consider $\mathcal{T}^{\eps}_\eta$ as in \eqref{eq:def_of_Ewicketa_as_operator}. Then,
    \begin{enumerate}[label=(\roman*)]
        \item  the operator \(\mathcal{T}^{\eps}_\eta\)
        is trace-class and has a continuous kernel for any \(\eta \geq 0\) small enough. The same statement holds for its restriction to the span of a finite number of eigenvectors of \(h\) \begin{equation}\label{eq:def_of_EwicketaN_as_operator}
            \mathcal{T}^{\eps}_{\eta,N} \deq P_N \mathcal{T}^{\eps}_{\eta} P_N \,.
        \end{equation}

        \item \label{lemma:convergence_of_Ewick_conv_of_EN} The diagonal of the operator \(\mathcal{T}^{\eps}_{\eta, N}\), denoted by \(T_{\eta,N}^{\eps}(x) \equiv \mathcal{T}_{\eta,N}(x,x) \), converges in \(L^1(\mathbb{R}^d)\) to \(T^{\eps}_\eta(x)\). Moreover, \(T^{\eps}_\eta\in L^{\infty}(\mathbb{R}^d).\)
        \item  Uniformly in \(x \in \mathbb{R}^d\), 
        \[\lim_{\eta \to 0} T^\eps_\eta(x) = T^\eps_0(x)\eqd  T^{\eps}(x)\,. \]
    \end{enumerate}
\end{lemma}

\begin{proof}
    { (i)} Consider the operator \(\mathcal{T}^\eps_\eta\) as in \eqref{eq:def_of_EwicketaN_as_operator} whose kernel we denote by \(\mathcal{T}^\eps_\eta(x,y) \) for \(x,y\in \mathbb{R}^d\). By a resolvent expansion and Hölder's inequality on Schatten spaces, we find that \(\mathcal{T}^\eps_\eta \) is trace-class, since 
    \[\|\mathcal{T}^\eps_{\eta}\|_{\mathfrak{S}_1}\leq \Big\|\frac{1}{h^U}\Big\|_{\mathfrak{S}_2} \|\alpha^\eps_{\eta}- \tau^\eps\|_{L^\infty} \Big\|\frac{1}{h^{U^{\eps}_\eta}}\Big\|_{\mathfrak{S}_2}< \infty\]
    uniformly, for \(\eta \geq0\) small enough. 
    
    Let us now show the continuity of \(T^\eps_\eta(x,y)\) at all points \((x,y) \in \mathbb{R}^d \times \mathbb{R}^d\). Since the kernel of \(1/h^U\), \(1/h^{U^\eps_\eta}\) is continuous outside the set of measure zero \(\{(x,x)\col x \in \mathbb{R}^d\}\) and  
    \begin{equation*}
        {\mathcal{T}^\eps_\eta(x,y) = \int\d w \,\frac{1}{h^U}(x,w)\,({\alpha^\eps_\eta(w)-\tau^\eps(w)})\,\frac{1}{h^{U^\eps_\eta}}(w,y)}\,, 
    \end{equation*}   we can conclude by dominated convergence that \((x,y) \mapsto T^\eps_\eta(x,y)\) is jointly continuous.

    Let us now consider the operator \(\mathcal{T}^\eps_{\eta, N}\deq P_N \mathcal{T}^\eps P_N \) defined in \eqref{eq:def_of_EwicketaN_as_operator}. Arguing as above and using the fact that 
    \[\Big\|(1-P_N)\frac{1}{h^U}\Big\|_{\mathfrak{S}_2}^2=\sum_{k\geq0}\frac{\mathbbold{1}_{\lambda_k>N}}{\lambda_k^2} \xrightarrow[N\to \infty]{}0,\]
    we also have that \(\mathcal{T}^\eps_{\eta, N}\) is a trace-class operator for all \(N>0\), and its kernel 
    \[\mathcal{T}^\eps_{\eta, N}(x,y)= \sum_{n,m \in \mathbb{N}} \mathbbold{1}_{\lambda_n\vee\lambda_m\leq N}\langle u_l,\mathcal{T}^\eps_\eta u_m\rangle u_m(x)\bar{u}_n(y)\]
    is continuous since it is a finite sum of continuous functions.

    (ii) Let \(\{\mu_k\}_{k \geq 0}\) denote the singular values of \(T^\eps_\eta\), \(\{\psi_k\}_{k\geq 0}\) the corresponding singular vectors and  \(\{\phi_k\}_{k\geq 0}\) a family of orthonormal vectors defined as \(\phi_k\deq \mu_k^{-1} T^\eps_\eta\psi_k\) for all \(k\geq 0\). Then, following the argument in \cite[Section 1]{traceableintegralkernels} together with the continuity of \(T^\eps_\eta(x,y)\), we have 
    \begin{equation}\label{eq:Ewick_kernel_expansion_singvect}
        (\mathcal{T}^\eps_\eta f) (x) = \sum_{k\geq 0} \mu_k \langle f,\psi_k\rangle \phi_k(x) 
    \end{equation}
    and the series expansion 
    \[\sum_{k\geq 0} \mu_k\phi_k(x)\bar{\psi}_k(x)\]
    converges for almost every point in the diagonal of \(\mathbb{R}^d \times \mathbb{R}^d\) to \(T^\eps_\eta(x)\). The series expansion \(\eqref{eq:Ewick_kernel_expansion_singvect}\) and the argument above imply that 
    \[T^\eps_{\eta,N} (x)= \sum_{k\geq 0}\mu_k P_N\phi_k(x) P_N\overline{\psi}_k(x) \]
    for almost every \(x\in \mathbb{R}^d\). 
    Using Cauchy-Schwarz we find  
    \begin{equation}
    \label{eq:diff_of_kernelsEwick}
    \begin{split}
        \int \d x \,\abs{T^\eps_\eta(x)- T^\eps_{\eta,N}(x)} & \leq \sum_{k\geq 0}\mu_k \biggl( \int\d x\,  \abs{(1-P_N)\phi_k(x)\bar{\psi}_k(x)} \\
        &\hspace{4cm}+ \int\d x \, \abs{P_N\phi_k(x)(1-P_N)\bar{\psi}_k(x)} \biggr) \\
        &\leq \sum_{k= 0}^M\mu_k \bigl(\|(1-P_N)\phi_k\|_{L^2(\mathbb{R}^d)} + \|(1-P_N)\psi_k\|_{L^2(\mathbb{R}^d)}\bigr) +  \sum_{k> M } \mu_{k}\,.
    \end{split}\end{equation}
    the right-hand side of \eqref{eq:diff_of_kernelsEwick}  converges to 0 as $M, N \to \infty$, since for any \(\epsilon>0\), we can choose \(M>0\) big enough such the second term of \eqref{eq:diff_of_kernelsEwick} is smaller than \(\epsilon\), as \(\mathcal{T}^\eps_\eta\) is Hermitian and \(\|\mathcal{T}^\eps_\eta\|_{\mathfrak{G}^2}<\infty\), and set \(N>0\) large enough such that \(\abs{(1-P_N)\phi_k} + \abs{(1-P_N)\psi_k} \leq \epsilon\) for all \(k=0, \dots,M\).

    (iii) By a resolvent expansion and the bounds on $ (\kappa-\Delta)^{-1}$, which follow from Proposition \ref{prop:GN}, we obtain 
    \begin{align*}
        \abs{T^\eps_\eta(x)-T^\eps_0(x)} & = \int \d y\, \frac{1}{h^{U^\eps_0}}(x,y)\abs{\alpha^\eps_\eta(x)-\alpha^\eps_0(x)}\frac{1}{h^{U^\eps_\eta}}(y,x)\leq \kappa^{d/2-2}\|\alpha^\eps_\eta-\alpha^\eps_0\|_{L^\infty}.
    \end{align*}
   The uniform convergence for \(\eta \to 0\) follows from the limit \eqref{unif convergence of alpha_eta}. 
\end{proof}

\begin{proof}[Proof of Proposition \ref{Prop:regularisation of euclidean field theory}]

    We regularise the interaction \(W^\eps\) replacing it with $W^\eps_\eta$ as in \eqref{eq:def:regularised version of the field Weps}, so that 
    \[\zeta^{W^{\eps}}= \lim_{\eta\to 0}\lim_{N \to \infty} \zeta^{W^{\eps}_{\eta,N}}\]
    where we use the regularisation function \(\vartheta(x)\deq \mathbbold{1}_{\abs{x}\leq 1}\) as in \eqref{eq:finite dimensional reg of phi}.  
    Indeed, by Lemma \ref{lemma:W_N is a Cauchy sequence}, we have 
\begin{equation*}
\lim_{\eta\to 0}\lim_{N \to \infty}\big\|W^\eps - W^\eps_{\eta, N} \big\|_{L^q} \leq \lim_{\eta\to 0}\big\|W^\eps - W^\eps_{\eta} \big\|_{L^q}+ \lim_{N \to \infty}\big\|W^\eps_{\eta}- W^\eps_{\eta, N} \big\|_{L^q} = 0
\end{equation*}
    for any \(1\leq q< \infty\). 
    We  then have for any \(1 \leq p < \infty \) and \(r= pq/(q-p)\) 
    \begin{align*}
        \big\|\ee^{-W^\eps}- \ee^{-W^\eps_{\eta, N} }\big\|_{L^p} &\leq \int_{0}^{1}\d t\, \big\|(W^\eps - W^\eps_{\eta, N})\,\ee ^{-tW^\eps -(1-t) W^\eps_{\eta, N}}\big\|_{L^p}\\
        &\leq \big\|W^\eps - W^\eps_{\eta, N}\big\|_{L^q} \int_{0}^{1}\d t\, \big\| t\ee ^{-W^\eps}+(1-t)\ee ^{- W^\eps_{\eta, N}}\big\|_{L^r}  \\
        &\leq  C  \big\|W^\eps - W^\eps_{\eta, N}\big\|_{L^q}\,,
    \end{align*}
    where we used the integrability of $W^\eps_\eta$ (cf.\ Section \ref{sec:field}) and 
    \[\zeta^{W^\eps_{\eta, N}} < \infty\,,\]
    which can be seen from the identity \eqref{zeta^W^eps_eta_N} below
    and that the fact that \(\re{f_{2,N}}<0\) for any \(N>0\). 
    Let us now show 
    \[\zeta^{W^\eps_\eta}=\ee^{E^\eps + \frac{1}{2} \int \d{x}\,\alpha^\eps(x) \alpha^\eps_\eta(x) + S_\eta }\int \mu_{v^{\eps}_\eta}(\mathrm{d} \xi)\, \ee^{-\ii\langle T^\eps_\eta+ \alpha^\eps,\xi\rangle} \mathrm{e}^{f_2(\xi)}\,.\]
    Introducing a zero in the exponent and rearranging the terms, we find
    \begin{multline} \label{w_n:introductionofzero}
            \zeta^{W^\eps_{\eta, N}}
         =\ee^{E^\eps+ \frac{1}{2}\int \d{x} \,\alpha^\eps(x)\alpha^\eps_\eta(x) }
         \\
\times         \E\qbb{
        \ee^{ - \int\d{x}\, \bigl( \alpha^\eps_\eta(x)-\tau^\eps(x)\bigr):\abs{\phi_N(x)}^2: -\frac{1}{2}\int \d{x}\, \d{\tilde{x}}\,\bigl(\abs{\phi_N(x)}^2 - \alpha^\eps(x) \bigr)\varphi(\eta x) v^\eps(x-\tilde{x})\varphi(\eta \tilde{x})\bigl(\abs{\phi_N(\tilde x)}^2 - \alpha^\eps(\tilde{x}) \bigr) }}\,.
    \end{multline}
    Since \[\E\qb{\langle\phi_N(x),\phi_N(y)\rangle} = (h^{-1})_N (x,y) = \sum_{k=0}^\infty\frac{\mathbbold{1}_{\lambda_k\leq N}}{\lambda_k} u_k(x)u_k^*(y)\,,\] \(\phi_N\) is a finite-dimensional complex Gaussian with covariance \((h^{-1})_N\). We then find by a  direct computation
    \begin{equation}
        \label{eq:change_of_gaussian_law}
        \ee^{-\int \d{x}\, (\alpha^\eps_\eta(x)- \tau^\eps(x))\col\abs{\phi_N(x)}^2\col} \, \P(\d\phi)  = \ee^{S^\eps_{\eta, N}} \, \P^{U^\epsilon_\eta}(\d\phi)\,,
    \end{equation}
    where
    \begin{equation*}
         S^\eps_{\eta, N}\deq \log \mathbb{E}\qB{\ee^{-\int \d{x} \,(\alpha^\eps_\eta(x)- \tau^\eps(x))\, :\abs{\phi_N(x)}^2:}}
    \end{equation*}
    is uniformly bounded in \(N>0\) and converges by Lemma \ref{lemma:integration of rho and tauepsilon}. Similarly, for a centred Gaussian field \(\phi\) covariance of covariance \(\bigl(h^{U^\eps_\eta}\bigr)^{-1}\), we have
    \[\mathbb{E}^{U^\eps_\eta} \bigl[\phi_N(x) \overline{\phi_N}(y)\bigr] = \Bigl(\frac{1}{h^{U^\eps_\eta}} \Bigr)_N (x,y)\,.\] 
Hence,
\begin{multline}\eqref{w_n:introductionofzero} = \ee^{S^\eps_{\eta, N}+E^\varepsilon + \frac{1}{2}\int \d{x} \,\alpha^\eps(x)\alpha^\eps_\eta(x)}
            \\ \times \E^{U^{\eps}_{\eta}}
            \qB{\ee^{-\frac{1}{2}\int \d{x}\,\d{\tilde{x}}\,\bigl(:\abs{\phi_N(x)}^2:_{h^U} - \alpha^\eps(x) \bigr)\varphi(\eta x) v^\eps(x,\tilde{x})\varphi(\eta \tilde{x})\bigl(\col\abs{\phi_N(\tilde x)}^2\col_{h^U} - \alpha^\eps(\tilde{x}) \bigr)}}\label{eq:w_n:change of law}
        \end{multline}
        where the Wick ordering
        \[\wick{|\phi_N(x)|^2}_{h^U}\deq |\phi_N(x)|^2- \bigl(h^{U} \bigr)^{-1}_N (x,x)\]
        is now performed with respect to a covariance different from that of the law of \(\phi_N\). In order to arrive at the expression of $f_2$ in \eqref{def:f_2}, we need to Wick order with respect to the law of \(\phi\) in \eqref{eq:w_n:change of law}. Thus, we replace \(\wick{\abs{\phi_N(x)^2}}_{h^U}\) by \(\wick{\abs{\phi_N(x)^2}}_{h^{U^{\eps}_{\eta}}}\) by subtracting 
    \begin{equation*}
        T^{\eps}_{\eta,N} (x)= \biggl(\frac{P_N}{h^U} - \frac{P_N}{h^U + (\alpha^\eps_\eta- \tau^\eps )_N}\biggr)(x,x)\,.
    \end{equation*}
    Thus, we rewrite 
\begin{multline}\label{final_step_of_f_2}
    ~\eqref{eq:w_n:change of law} = \ee^{S^\eps_{\eta, N}+E^\varepsilon + \frac{1}{2}\int \d{x}\,\alpha^\eps(x)\alpha^\eps_\eta(x)}
    \\
\times \E^{U^{\eps}_{\eta}} \qbb{
        \ee^{-\frac{1}{2}\int \d{x}\,\d{\tilde{x}}\,\bigl(\col|\phi_N(x)|^2\col - \alpha^\eps(x) - T^\eps_{\eta,N} (x) \bigr)\varphi(\eta x) v^\eps(x-\tilde{x})\varphi(\eta \tilde{x})\bigl(\col|\phi_N(\tilde x)|^2\col - \alpha^\eps(\tilde{x}) -  T^\eps_{\eta,N} (\tilde{x}) \bigr)}}\,,
    \end{multline}
    where the Wick ordering is done with respect to the measure $\P^{U^{\eps}_{\eta}}$.
    The rest of the proof follows the one of \cite[Proposition~4.1]{FKSS_2020}. We present here the main ideas.

Applying a Hubbard-Stratonovich transformation \eqref{HS} with \(f(x) = \wick{\abs{\phi_N(x)}^2} - \alpha^\eps(x) - T^\eps_{\eta,N} (x)\) and using the notation \(\varrho_N(x)= \big(\frac{P_N}{h + (\alpha^\eps_\eta- \tau^\eps )_N}\big)(x,x)\), we find
\begin{equation*}
                \eqref{final_step_of_f_2}  =  \ee^{S^\eps_{\eta, N}+E^\varepsilon + \frac{1}{2}\int \d{x}\, \alpha(x)\alpha^\eps_\eta(x)- \ii \langle \xi,\alpha^\eps + T^\eps_{\eta,N}\rangle}
            \int \mu_{v^{\eps}_\eta}(\d {\xi})\,\ee^{ \ii \langle\xi,\varrho_N \rangle } \E^{U^{\eps}_{\eta}} \qB{\ee^{\ii\langle \phi_N, \xi \phi_N\rangle}}\,.
\end{equation*}
    Moreover, we find the expression
    \begin{align*}
            \begin{split}
                {f_{2,N} (\xi)} &\deq
                \log\biggl\{\ee^{ \ii \langle\xi,\varrho_N \rangle } \E^{U^{\eps}_{\eta}} \qB{\ee^{\ii\langle\phi_N,\xi \phi_N\rangle }}\biggr\}\\ 
            &=
                - \ii\langle\xi,\varrho_N \rangle + \log{\E^{U^{\eps}_{\eta}} \qB{\ee^{\ii\langle\phi_N,\, \xi \phi_N\rangle }}} \\
                &= 
                - \ii\langle\xi,\varrho_N \rangle  +\int_0^{\infty} \mathrm{d} t \, \tr\biggl[P_N\biggl(\frac{1}{t+h^{U} + (\alpha^\eps_\eta- \tau^\eps-\ii \xi )_N}{}-\frac{1}{t+h^{U} + (\alpha^\eps_\eta- \tau^\eps)_N}\biggr) \biggr]\\
              & = \int_0^{\infty} \mathrm{d} t \, \tr\biggl[P_N\biggl(\frac{1}{t+h^{U} + (\alpha^\eps_\eta- \tau^\eps-\ii \xi )_N}{}-\frac{1}{t+h^{U} + (\alpha^\eps_\eta- \tau^\eps)_N}\\ & \qquad -\frac{1}{t+h^{U} + (\alpha^\eps_\eta- \tau^\eps)_N} \ii \xi_N \frac{1}{t+h^{U} + (\alpha^\eps_\eta- \tau^\eps)_N}\biggr)\biggr]\,.
            \end{split}
        \end{align*}
        
    Hence, we can rewrite the truncated partition function $\zeta^{W^\eps_{\eta, N}}$ as 
  \begin{equation}\label{zeta^W^eps_eta_N}
        \zeta^{W^\eps_{\eta, N}} = \ee^{E^\eps + \frac{1}{2} \int \d{x}\,\alpha^\eps(x) \alpha^\eps_\eta(x) + S^\eps_{\eta, N} }\int \mu_{v^{\eps}_\eta}(\mathrm{d} \xi)\, \ee^{-\ii\langle T^\eps_{\eta,N}+ \alpha^\eps,\xi\rangle} \ee^{f_{2,N}(\xi)}\,.
    \end{equation}
    
    \noindent Finally, we notice that \(\mathrm{Re} \bigl\{\bigl(- \ii \langle \xi, \alpha^\eps + T^\eps_{\eta,N}\rangle + f_{2,N}(\xi) \bigr) \bigr\},\re{ (- \ii \langle \xi, \alpha^\eps  +T^\eps_{\eta}\rangle + f_2(\xi) ) } < 0\) where \(f_2(\xi)\deq \lim_{N \to \infty} f_{2,N}(\xi) \) exists almost surely since \(\xi\) is almost surely bounded (see \cite[Proof of Proposition 4.1.]{FKSS_2020} for further details). 

    Lastly, we want to show the convergence of $\zeta^{W^\eps_{\eta, N}}$ to $\zeta^{W^\eps_{\eta}}$ as $N\to \infty$. First using Cauchy-Schwarz, recalling \eqref{def: rescaled interaction potential} and the fact that \(\|\varphi\|_{L^\infty}<\infty\) since \(\varphi\) has compact support, we show
    \begin{align}
        \biggl(\int \mu_{v^\eps_\eta}(\d {\xi}) | \langle\xi,T^\eps_\eta-T^\eps_{ \eta, N}\rangle| \biggr)^2&\leq \int \mu_{v^\eps_\eta}(\d {\xi}) \big| \langle\xi,T^\eps_\eta-T^\eps_{ \eta, N}\rangle\big|^2\nonumber\\
        &= \int \d x\, \d y\, (T^\eps_\eta-T^\eps_{ \eta, N})(x) v^\eps_\eta(x-y) (T^\eps_\eta-T^\eps_{ \eta, N})(y)\nonumber\\
        &\lesssim\|T^\eps_\eta-T^\eps_{ \eta, N}\|^2_{L^2 \nonumber}\\ & \lesssim\|T^\eps_\eta-T^\eps_{ \eta, N}\|_{L^\infty} \|T^\eps_\eta-T^\eps_{ \eta, N}\|_{L^1} \xrightarrow[N\to \infty]{} 0\,,
    \label{eq:epsEwickEwick_eta_nu}\end{align}
    where in the last line we used Young's convolution inequality and Lemma \ref{lemma:convergence_of_EWick} (ii). Hence, using dominated convergence, we get
    \begin{align*}
        &\lim_{N \to \infty} \Big|{\zeta^{W^\eps_{\eta,N}}} - \ee^{E^\eps + \frac{1}{2} \int \d{x}\,\alpha^\eps(x) \alpha^\eps_\eta(x) + S^\eps_{\eta} }\int \mu_{v^{\eps}_\eta}(\mathrm{d} \xi)\, \ee^{-\ii\langle T^\eps_{\eta}+ \alpha^\eps,\xi\rangle} \ee^{f_{2}(\xi)} \Big|\\
        &\,\leq  \lim_{N \to \infty} \Big|\ee^{S^\eps_{\eta, N}} -\ee^{S^\eps_{\eta}} \Big| \ee^{E^\eps + \frac{1}{2} \int \d{x}\,\alpha^\eps(x) \alpha^\eps_\eta(x) }\int \mu_{v^{\eps}_\eta}(\mathrm{d} \xi)\, \ee^{-\ii\langle T^\eps_{\eta}+ \alpha^\eps,\xi\rangle} \ee^{f_{2}(\xi)} \\
        &\;\;  +  \lim_{N \to \infty} \ee^{E^\eps + \frac{1}{2} \int \d{x}\,\alpha^\eps(x) \alpha^\eps_\eta(x) + S^\eps_{\eta} }\int \mu_{v^{\eps}_\eta}(\mathrm{d} \xi)\, \big|  \bigl(- \ii \langle \xi ,\, T^\eps_{\eta,N}\rangle + f_{2,N}(\xi) +\ii \langle \xi,\,T^\eps_{\eta}\rangle - f_2(\xi) \bigr)\big|\\
        & \, =  0\,. \qedhere
    \end{align*} 
\end{proof}

\subsection{Brownian paths} \label{sec:brownian}
We recall now some basic notions for Brownian paths.
For \(0 \leq \tilde \tau < \tau\), let \(\Omega ^{\tau, \tilde{\tau}}\) be the space of continuous paths \(\omega\col [\tilde{\tau}, \tau] \to \mathbb{R}^d\), and for $x, \tilde x \in \mathbb R^d$, \(\Omega ^{\tau, \tilde{\tau}}_{x, \tilde{x}}\) the space of continuous paths that start at $\tilde x \in \mathbb{R}^d$ at time $\tilde \tau $ and finish at $x \in \mathbb{R}^d$ at time $\tau $. Given $x, \tilde x \in \mathbb R^d$ and $0 \leq \tilde \tau < \tau$, we denote by \(\mathbb{P}_{x,\tilde{x}}^{\tau, \tilde \tau}\) the law of the Brownian bridge equal to \(\tilde x\) at time \(\tilde \tau\) and to \(x\) at time \(\tau\). 
Moreover, for $x,\tilde x \in \bR^d$ and $0 \leq \tilde \tau <\tau$, we define the positive measure over paths in \(\Omega ^{\tau, \tilde{\tau}}\)  through
\begin{equation}\label{def:W law of Brownian Bridge starting and finishing at x, tilde x}
    \mathbb{W}_{x, \tilde{x}}^{\tau, \tilde{\tau}}(\mathrm{d} \omega)\deq\psi^{\tau-\tilde{\tau}}(x-\tilde{x}) \mathbb{P}_{x, \tilde{x}}^{\tau, \tilde{\tau}}(\mathrm{d} \omega)\,,
\end{equation}
where for $t>0$ we defined the heat kernel 
\begin{equation}\label{eq:heatkernel}
  \psi^t(x)\deq\bigl(\mathrm{e}^{t\Delta  / 2}\bigr)_{x, 0}=(2 \pi t)^{-d / 2} \mathrm{e}^{-|x|^2 / 2 t} \,.
\end{equation}
The law $\mathbb{W}_{x, \tilde{x}}^{\tau, \tilde{\tau}} $ is characterised by its finite dimensional distribution. For $n \in \mathbb{N}^*$, $\tilde{\tau}<t_1<\cdots<t_n<\tau$ and $f\col \mathbb{R}^{dn} \rightarrow \mathbb{R}$ a continuous function, we have
\[
\begin{aligned}
\int \mathbb{W}&_{\tilde{x}, x}^{\tau, \tilde{\tau}}(\mathrm{d} \omega) \, f\bigl(\omega(t_1), \ldots, \omega(t_n)\bigr) \\
& =\int \mathrm{d} x_1 \cdots \mathrm{d} x_n \,\psi^{t_1-\tilde{\tau}}(x_1-\tilde{x}) \psi^{t_2-t_1}(x_2-x_1) \cdots \psi^{t_n-t_{n-1}}(x_n-x_{n-1})\\
&\hspace{7cm}\times\psi^{\tau-t_n}(x -x_{n}) f(x_1, \ldots, x_n)\,.
\end{aligned}
\]

In the next lemma we provide some useful bounds, whose proofs can be found in \cite{FKSS_24,FKSS_2020}.

\begin{lemma}[Analogous to \text{\cite[Lemmas~2.2 and 2.3]{FKSS_2020}}] \label{Lemma:bounds_on_the_law_of_brownian_bridge}
\leavevmode
\begin{enumerate}[label=(\roman*)] \item There exists a constant \(C \equiv C_d>0\) such that, for all \(0\leq \tilde{\tau} < \tau\), we have
              \begin{equation*} \sup_{x, \tilde{x}}\int \mathbb{W}^{\tau, \tilde{\tau}}_{x,\tilde{x}}(\d{\omega}) = \sup_{x, \tilde{x}} \psi ^{\tau, \tilde{\tau}}(x-\tilde{x}) \leq \frac{C}{(\tau - \tilde{\tau})^{\frac{d}{2}}}\,.\end{equation*}
        \item \label{Lemma:bounds_on_the_law_of_brownian_bridge(ii)} There exists a constant \(C > 0\) such that for all \(\tilde{\tau} \leq s \leq t \leq \tau\), we have
              \begin{equation*} \int \mathbb{P}_{x, \tilde{x}}^{\tau, \tilde{\tau}}(\mathrm{d} \omega)\,|\omega(t)-\omega(s)|^2 \leq C \biggl( (t-s)+|x-\tilde{x}|^2 \frac{(t-s)^2}{(\tau-\tilde{\tau})^2}\biggr) \,.\end{equation*}
               
    \end{enumerate}
\end{lemma}
The following result is well known.
\begin{lemma}[Feynman-Kac formula]\label{Feynman-Kac formula}
    Let \(I \in \mathbb{R}\) be a closed bounded interval and \(V\col I\times {\mathbb{R}^d} \to [0, \infty)\) a locally bounded function. Let \(\bigl(W^{\tau, \tilde{\tau}}\bigr)_{\tilde{\tau}\leq \tau \in I }\) be the propagator satisfying 
    \[\partial_\tau W^{\tau, \tilde{\tau} } = \bigg(\frac{\Delta}{2}+ V(\tau) \biggr)W^{\tau, \tilde{\tau}} \quad W^{\tau, \tau} = 1\,.\]
    Then \(W^{\tau, \tilde{\tau}}\) has operator kernel
    \[W_{x, \tilde{x}}^{\tau, \tilde{\tau}}=\int \mathbb{W}_{x, \tilde{x}}^{\tau, \tilde{\tau}}(\mathrm{d} \omega)\, \mathrm{e}^{\int_{\tilde{\tau}}^\tau \mathrm{d} s \,V(s, \omega(s))}\,. \]
\end{lemma}

\subsection{Proof of Proposition \ref{Prop:regularisation of phi42 quantum rel part funct}}

In this subsection we show Proposition \ref{Prop:regularisation of phi42 quantum rel part funct}. In particular we explain how to adapt the strategy in \cite[Proposition 3.12]{FKSS_2020} to our setting. To keep the paper as concise as possible, we present here only the main changes.

\begin{proof}[Proof of Proposition \ref{Prop:regularisation of phi42 quantum rel part funct}]
Consider $H^0$ as in \eqref{def: free Hamiltonian acting on the fock space} and \(H_n\) the restriction of the Hamiltonian \eqref{eq:WickHamiltonian} to the \(n\)-particle space \(\mathcal{H}_n\).
 Let us define on the Fock space $\cF$ the operator \(R_\eta\deq \bigoplus_{n \geq 0} R_{n,\eta} \) for $\eta>0$, where 
 \begin{multline}
    (R_{n,\eta})_{{\mathbf{x}}, \mathbf{\tilde{x}}}  \deq\ee^{- \nu \kappa n+ E^\eps} \int \mu_{C_\eta}(\d{\sigma})\,\ee^{-\ii\langle\sigma,\varrho^U_{\nu} \rangle_\nu- \langle\tau^\eps,\varrho^U_\nu\rangle} \\ \times \int \biggl(\prod_{i=1}^n \mathbb{W}_{x_i,\tilde{x}_i}^{\nu, 0}(\d{\omega_i}) \ee^{-{\int\limits_0^\nu \d{t}\, (U(\omega_i(t))- \tau^\eps(\omega_i(t)))}} \biggr) \prod_{i=1}^{n}\ee^{\ii\int\limits_{0}^\nu \d{t}\, \sigma(t, \omega_i(t)) },
    \label{eq:R_phi42}
\end{multline} 
and     \[
\langle \sigma, \varrho\rangle_\nu \deq \frac{1}{\nu}\int_0^\nu\d\tau\int\d x \, \sigma(\tau,x)\varrho(x)\,. 
\]
From \cite[Proposition 3.12]{FKSS_2020}, the operator kernel \eqref{eq:R_phi42} converges locally uniformly to \(\bigl(\ee^{-H_n}\bigr)_{{\mathbf{x}}, \mathbf{\tilde{x}}}\).

The rest of the proof is devoted to showing that the regularised partition function 
\[\mathcal{Z}_{\eta}\deq\Tr_{\mathcal{F}}(R_{\eta})/\Tr_{\mathcal{F}}\bigl(\ee^{-H^0}\bigr)\,. \] can be written as the right-hand side of \eqref{def:quantum_relative_partition_function}. 
We start by replacing \(\varrho_\nu^U\) by \(\varrho_\nu^{U^\eps_\eta}\) when convenient, by adding the term \eqref{eq:def_Ewick_nueta}. We obtain 
\begin{multline}
    \eqref{eq:R_phi42}= \ee^{- \nu \kappa n + E^\eps} \int \mu_{C_\eta}(\d{\sigma})\,\ee^{-\ii\langle\sigma,\varrho^{U^\eps_\eta}_{\nu} \rangle_\nu-\ii\langle\sigma,T^\eps_{\nu,\eta } \rangle_\nu - \langle\tau^\eps, \varrho^U_\nu\rangle} \\ \times \int \biggl(\prod_{i=1}^n \mathbb{W}_{x_i,\tilde{x}_i}^{\nu, 0}(\d{\omega_i}) \ee^{-{\int\limits_0^\nu \d{t}\, (U(\omega_i(t))- \tau^\eps(\omega_i(t)))}} \biggr) \prod_{i=1}^{n}\ee^{\ii\int\limits_{0}^\nu \d{t}\, \sigma(t, \omega_i(t)) }\,.\label{eq:R_phi42, second exp}
\end{multline}
We introduce the function \(\alpha^\eps_\eta\) in \eqref{eq:R_phi42, second exp} by performing a Hubbard-Stratonovich transformation similarly as in \cite{FKSS_2020}, with an additional term \(\alpha^\eps\), and then compensate for the resulting contributions. In particular, for \(f(t,x) \deq\sum_{i=1}^n \delta(x-\omega_i(t)) - \varrho^U_{\nu}(x)/\nu  - \alpha^\eps(x)/\nu\), we find that  
\[
\begin{split}-\frac{1}{2}\langle f, C_\eta f\rangle& = -\frac{1}{2}\int_0^\nu \d{t}\, \d{\tilde{t}} \int \d{x}\, \d{\tilde{x}} 
        \biggl( \sum_{i=1}^n  \delta(x-\omega_i(t)) - \frac{\varrho_{\nu}(x)}{\nu}\biggr)  \\
        &\hspace{1cm}\times
        \bigl( \nu \delta_{\eta,\nu}(t-\tilde{t}) v^\eps(x- \tilde{x})\varphi(\eta x)\varphi(\eta\tilde{x})\bigr)
        \biggl( \sum_{i=1}^n \delta(\tilde{x}-\omega_i(\tilde{t})) - \frac{\varrho_{\nu}(\tilde{x})}{\nu}\biggr) \\
        &+ \sum_{i=1}^n \int_0^\nu \d t\, \alpha^\eps_{\eta}(\omega_i(t)) - \langle\alpha^\eps,\varrho^U_\nu\rangle -\frac{1}{2}\int \d x\, \d {\tilde{x}}\, \alpha^\eps_\eta(x)\alpha^\eps(\tilde{x})\,. 
\end{split}\]
 Hence, we can rewrite
\begin{multline}
    \eqref{eq:R_phi42, second exp}= \ee^{- \nu \kappa n+ E^\eps+\frac{1}{2}\int \d x\, \d {\tilde{x}} \,\alpha^\eps_\eta(x)\alpha^\eps(\tilde{x})} \int \mu_{C_\eta}(\d{\sigma})\,\ee^{-\ii\langle\sigma,\varrho^{U^\eps_\eta}_{\nu} \rangle_\nu-\ii\langle\sigma,T^\eps_{\nu,\eta } + \alpha^\eps\rangle_\nu + \langle\alpha^\eps- \tau^\eps,\varrho^U_\nu\rangle} \\ \times \int \biggl(\prod_{i=1}^n \mathbb{W}_{x_i,\tilde{x}_i}^{\nu, 0}(\d{\omega_i}) \ee^{-{\int\limits_0^\nu \d{t}\, U^\eps_\eta(\omega_i(t))}} \biggr) \prod_{i=1}^{n}\ee^{\ii\int\limits_{0}^\nu \d{t} \,\sigma(t, \omega_i(t)) }\label{eq:R_phi42, third}\,.
\end{multline}
Taking the trace of $R_\eta$ we find
\begin{multline*}
     \Tr_{\mathcal{F}}(R_\eta) = \ee^{E^\eps+\frac{1}{2}\int \d x\, \d {\tilde{x}}\, \alpha^\eps_\eta(x)\alpha^\eps(\tilde{x})} \int \mu_{C_\eta}(\d{\sigma})\ee^{\ii\langle\sigma,T^\eps_{\nu,\eta } + \alpha^\eps\rangle_\nu + \langle\alpha^\eps- \tau^\eps,\varrho^U_\nu\rangle}\\ 
     \times \exp\biggl\{\int_0^\infty \d t \,\Tr\biggl(\frac{1}{t + K(-\kappa - U^\eps_\eta - \ii\sigma)}- \frac{1}{t + K(-\kappa - U^\eps_\eta )}\ii\sigma \frac{1}{t + K(-\kappa - U^\eps_\eta )}\biggr)\biggr\}
\end{multline*}
where we used the fact that from \cite[lemma 3.13]{FKSS_2020}
\[\big\langle\sigma,\varrho_\nu^{U^\eps_\eta}\big\rangle_\nu = \int_{0}^\infty \d t\, \Tr\biggl(\frac{1}{t + K(-\kappa - U^\eps_\eta )}\sigma \frac{1}{t + K(-\kappa - U^\eps_\eta )}\biggr)\, .  \]
 Finally, recalling \eqref{eq:def_S_nueta} and  
    \[\Tr_{\mathcal{F}}\bigl(\ee^{-H^0}\bigr) =\exp\biggl\{\int_0^\infty \d t \Tr\biggl(\frac{1}{t + K(-\kappa - U)}\biggr)\biggr\} \,,\]
we can conclude the proof. 
\end{proof}

\section{Proof of Proposition \ref{prop:convergence quantum many body system and Weps}} \label{sec:quantumtoclassic}

 In this section we show Proposition \ref{prop:convergence quantum many body system and Weps}. We use the functional integral representation introduced in Section \ref{sec:funcintrep} to study the rate of convergence of the relative partition function and the correlation functions in the mean-field limit, while keeping track of the parameter $\eps$. The basic strategy is similar to that of \cite[Section 5]{FKSS_24}. However, as we saw in Section \ref{sec:funcintrep}, in the inhomogeneous setting the functional integral representation is more complicated and the approach has to be suitably adapted. Our methods work for $d=2,3$, and hence all results of this section are stated for both dimensions.

Furthermore, as a result of our analysis, in analogy with \cite{FKSS_24}, we have some restrictions on the choice of the interaction potential $v^\eps$. In particular, we have that $\eps \equiv \eps(\nu)>0$  satisfies the lower bounds in \eqref{eq:epsbound}.
Let us define \(\chi\col [0,\infty) \to \mathbb{R}\) by 
\begin{equation}\label{eq:def:chieps}
    \chi(t)\deq\begin{cases}
        \log {t^{-1}} & d=2 \\
        t^{-1} & d=3 \,.
    \end{cases}
\end{equation}
From Lemmas \ref{lemma:bounds on taueps} and \ref{assumption:existence of alpha}, we have, uniformly in \(\eta>0\), that 
\[\|\tau^\eps\|_{L^\infty}\vee \|\alpha^\eps_{\eta} \|_{L^\infty}\lesssim_{v, \kappa}\chi(\eps)\]
as well as the bounds \begin{equation}\label{bound of E eps and S eps}
    E^\eps\vee S^\eps_{\eta} \lesssim_{v} \chi(\eps)^2 \,. 
\end{equation} 
 Moreover, \eqref{eq:epsbound} and \eqref{eq:def:chieps} imply that, for all \( C, b>0\) 
\begin{equation}\label{eq:limchisquare}
\lim _{\nu \rightarrow 0} \mathrm{e}^{C \chi(\varepsilon)^2} \nu^b=0\,.\end{equation}
The following propositions and lemma, which contains quantitative estimates on the relative partition function and the correlation functions, together with \eqref{eq:limchisquare}, allow us to show Proposition \ref{prop:convergence quantum many body system and Weps}. 
\begin{prop}\label{prop:estimates difference of partition functions}
    There exists \(C\equiv C_{\kappa, v}>0\) such that 
    \[\abs{\mathcal{Z}-\zeta^{W^{\eps}}} \lesssim_{\kappa, v,s, \theta} \ee^{C\chi(\eps)^2}(\nu^{1/4}\vee\nu^{1-s/2})\]
    for any \(s\in (d/2,2)\) satisfying \eqref{eq:tracehs}. 
\end{prop}

\begin{lemma}\label{lemma:bound on decay of wick ordered correlation functions}
    There exists a constant \(C\equiv C_{\kappa, v}\) for all \(p\in \mathbb{N}\), we have 
    \[(\widehat{\gamma}_p)_{\mathbf{x}, \tilde{\mathbf{x}}}\vee\nu^p(\widehat{\Gamma}_p)_{\mathbf{x}, \tilde{\mathbf{x}}}\lesssim_{d,v,\kappa,p,\theta} \ee^{C\chi(\eps)^2} \Upsilon_{\theta}({\mathbf{x}, \tilde{\mathbf{x}}})\,.\]
\end{lemma}
The result follows from \cite[Proposition 7.2]{FKSS_2020}, recalling Propositions \ref{Prop:regularisation of euclidean field theory} and \ref{Prop:regularisation of phi42 quantum rel part funct} and \eqref{bound of E eps and S eps}.

Using a more precise bound, we improve \cite[Proposition 5.2]{FKSS_24} and obtain the following analogous result regarding the convergence of the Wick-ordered correlation function. 
\begin{prop}\label{prop:p point correlation functions convergence} Let \(s\in(d/2,2)\) satisfy \eqref{eq:tracehs}.
    Consider \(R>0\) and let  \(\mathbf{x},\mathbf{\tilde {x}} \in \mathbb{R}^{dp}\) satisfy \(\abs{\mathbf{\tilde{x}}}+\abs{\mathbf{x}}\leq R\).  We define $$
\theta(s,d)\deq \frac{(2-s)(4-d)}{8 } \wedge \frac{(4-d)}{16}\,.
$$
Then, for all \(p\in \mathbb{N}\), there exists a constant \(C_1\) depending on \(v,\kappa\) such that 
\[\big|{\nu^p(\widehat{\Gamma}_p)_{\mathbf{x}, \tilde{\mathbf{x}}}-(\widehat{\gamma}_p)_{\mathbf{x}, \tilde{\mathbf{x}}}}\big|\lesssim_{d,p,v, \kappa}\ee^{C_1\chi(\eps)^2}\nu^{\theta(s,d)}(1+ \sqrt{\nu R})\,.\]
\end{prop}
\begin{proof}
    The proof follows from \cite[Section 5.2,Appendix C]{FKSS_24}, using Proposition \ref{prop:estimates difference of partition functions} and the second estimate of Lemma \ref{Lemma:bounds_on_the_law_of_brownian_bridge}; we omit further details. 
\end{proof}
We can now show Proposition \ref{prop:convergence quantum many body system and Weps}.

\begin{proof}[Proof of Proposition \ref{prop:convergence quantum many body system and Weps}]
    We obtain the convergence of the relative partition functions as a direct consequence of Proposition \ref{prop:estimates difference of partition functions}. The bound in   \eqref{eq:convergence of relative partition function with decay} follows from
\begin{equation}\label{eq:corrfunctionGammawUpsilon}
\lim_{\eps,\nu \to 0 } \sup _{\mathbf{x}, \tilde{\mathbf{x}}} \frac{\big|\nu^p(\widehat{\Gamma}_p)_{\mathbf{x}, \tilde{\mathbf{x}}}-(\widehat{\gamma}_p^{W^\eps})_{\mathbf{x}, \tilde{\mathbf{x}}}\big|}{\Upsilon_{\theta-c}(\mathbf{x}, \tilde{\mathbf{x}})}=0\,.
  \end{equation}  
  Hence, we show \eqref{eq:corrfunctionGammawUpsilon}. 
    Let \(c, R>0\), we have from Lemma \ref{lemma:bound on decay of wick ordered correlation functions}
    \begin{equation}
        \sup_{\abs{\mathbf{x}}+\abs{\mathbf{\tilde x}}>R} \frac{\big|{\nu^p(\widehat{\Gamma}_p)_{\mathbf{x}, \tilde{\mathbf{x}}}-(\widehat{\gamma}_p)_{\mathbf{x}, \tilde{\mathbf{x}}}}\big|}{\Upsilon_{\theta-c}({\mathbf{x}, \tilde{\mathbf{x}}})} \leq \sup_{\abs{\mathbf{x}}+\abs{\mathbf{\tilde x}}>R} \frac{C\ee^{C_1\chi(\eps)^2}\Upsilon_{\theta}({\mathbf{x}, \tilde{\mathbf{x}}})} {\Upsilon_{\theta-c}({\mathbf{x}, \tilde{\mathbf{x}}})} \leq C\ee^{C_1\chi(\eps)^2}R^{- c(2-d/2)}\,.
    \end{equation}
    By Proposition \ref{prop:p point correlation functions convergence}, 
    \[\sup_{\abs{\mathbf{x}}+\abs{\mathbf{\tilde x}}\leq R} \frac{\big|{\nu^p(\widehat{\Gamma}_p)_{\mathbf{x}, \tilde{\mathbf{x}}}-(\widehat{\gamma}_p)_{\mathbf{x}, \tilde{\mathbf{x}}}}\big|}{\Upsilon_{\theta-c}({\mathbf{x}, \tilde{\mathbf{x}}})} \leq C\ee^{C_1\chi(\eps)^2}\nu^{\theta(s,d)}(1+ \sqrt{\nu R})R^{\theta(2-d/2)}\,. \]
    Choosing \(R=\nu ^{-\alpha}\) with \(\alpha>0\) small enough (depending on \(d,s\)), we can conclude the proof. 
\end{proof}

\subsection{Proof of Proposition \ref{prop:estimates difference of partition functions}}\label{sub:RiemannsecEstimates}
 
 The rest of the section is dedicated to show Proposition \ref{prop:estimates difference of partition functions}. We start by noticing the following result, whose proof is similar to the corresponding lemma in \cite[Lemma 5.9]{FKSS_24}.
\begin{lemma}\label{lemma:diff of relative partition function} Recall \eqref{def:f_2}, \eqref{eq:S_eta}, and \eqref{def:E^Wick_eta} as well as \eqref{eq:def_of_F_2}, \eqref{eq:def_S_nueta}, and \eqref{eq:def_Ewick_nueta}. 
    Uniformly in \(\eta \geq0\), there exists a constant \(C>0\) such that
    \begin{multline}\label{eq:diff of relative partition function}
        \big|\mathcal{Z}_\eta-\zeta^{W^{\varepsilon}}\big| \leq \mathrm{e}^{C\chi(\eps)^2} \biggl[ \biggl(\int \mu_{\mathcal{C}_\eta}(\mathrm{d} \sigma)\big|F_2(\sigma)-f_2(\langle\sigma\rangle)\big|^2\biggr)^{1 / 2} + \abs{S^\eps_{\nu,\eta}-S^\eps_{\eta}} \\  +\int \mu_{v^\eps_\eta}(\d {\xi}) \big| \langle\xi,T^\eps_\eta- T^\eps_{\nu, \eta}\rangle\big|\biggr] \,.
    \end{multline}
\end{lemma}
In order to show Proposition \ref{prop:estimates difference of partition functions}, we need some key bounds on the three sums on the right-hand side of \eqref{eq:diff of relative partition function}. We are going to split the analysis in two different steps. First, we need to control the terms containing the differences $S^\eps_{\nu,\eta}-S^\eps_{\eta}$ and $T^\eps_\eta- T^\eps_{\nu, \eta}$. Then, we consider the term with difference of the functions $F_2(\sigma)-f_2(\langle\sigma\rangle)$. We sum up the results in the following two lemmas. Hence, the proof of Proposition \ref{prop:estimates difference of partition functions} is a direct consequence of  Lemma \ref{lemma:diff of relative partition function} together with Lemma \ref{lemma:Riemann sum estimates} and \ref{lemma:bound on diff F2 and f2} below. 
\begin{lemma}\label{lemma:Riemann sum estimates}
    Uniformly in \(\eta\geq0\) small enough, and for any \(s\in(d/2,2)\) satisfying \eqref{eq:tracehs} 
    \[\big|S^\eps_{\nu,\eta}-S^\eps_{\eta}\big|  +\int \mu_{v^\eps_\eta}(\d {\xi}) \big| \langle\xi,T^\eps_\eta- T^\eps_{\nu, \eta}\rangle\big| \lesssim_{ \kappa,s}\chi(\eps)^2\nu^{1/2-s/4}\,.\]
\end{lemma}

\begin{lemma}\label{lemma:bound on diff F2 and f2}
 Uniformly in \(\eta\geq0\) small enough, and for any \(s\in(d/2,2)\) satisfying \eqref{eq:tracehs} 
    \[\int \mu_{\mathcal{C}_\eta}(\mathrm{d} \sigma)|F_2(\sigma)-f_2(\langle\sigma\rangle)|^2\lesssim_{\kappa,v, s}\frac{\nu^{1/2}\vee\nu^{2-s}}{\eps^{5d+1}}\,.\]
\end{lemma}
\subsubsection{Proof of Lemma \ref{lemma:Riemann sum estimates}}
\label{subsec:Riemann sum estimates}
We split the proof of Lemma \ref{lemma:Riemann sum estimates} in two steps. We start with the bound of the term involving the functions $T^\eps_\eta$ and $T^\eps_{\nu, \eta}$ defined in \eqref{def:E^Wick_eta} and \eqref{eq:def_Ewick_nueta}  respectively.
\begin{lemma}\label{lemma:convergence_of_Ewick_quantitative}
    Let the $T^\eps_\eta$ and $T^\eps_{\nu, \eta}$ be as in \eqref{def:E^Wick_eta} and \eqref{eq:def_Ewick_nueta}. Let \(\xi\) be a centred Gaussian field of covariance \(v^\eps_\eta\). Then, under the assumptions of Lemma \ref{lemma:Riemann sum estimates},
    \begin{equation}\label{eq:conv of EWicks}
        \int \mu_{v^\eps_\eta}(\d {\xi}) \big| \langle\xi,T^\eps_\eta- T^\eps_{\nu, \eta}\rangle\big|\lesssim_{s,\kappa}\nu^{1/2-s/4}\|\alpha^\eps_\eta-\tau^\eps\|_{L^\infty}\,. 
    \end{equation}
\end{lemma}
\begin{proof}
    Note that 
    \begin{equation*}
        \varrho_\nu(x) =\Bigl(\frac{\nu }{\ee^{\nu h}-1}\Bigr)(x,x) =\Bigl(\frac{\nu \ee^{-\nu h}}{1-\ee^{-\nu h}}\Bigr)(x,x)  = \nu\sum_{n=1}^\infty\bigl(\ee^{-\nu n h}\bigr)(x,x)\,. 
    \end{equation*}
    Thus from the definition \eqref{eq:def_Ewick_nueta}, and Feynman-Kac formula, we can rewrite
    \begin{align}
        T^\eps_{\nu, \eta}(x)& =\nu\sum_{n=1}^\infty\bigl(\ee^{-\nu n h^U}- \ee^{-\nu n h^{U^\eps_\eta}} \bigr)(x,x) \nonumber \\
        & = \nu\sum_{t\in \nu \mathbb{N}^*}\ee^{-\kappa t } \int \mathbb{W}^{t,0}_{x,x}(\d \omega)\bigl(\ee^{-\int_0^t \d s\, U(\omega(s))}-\ee^{-\int_0^t \d s \, U^\eps_\eta(\omega(s))  }\bigr)\,.
    \end{align}
    Using the assumptions that \(U, U^\eps_\eta >0\), we find a uniform upper bound on \(T^\eps_{\nu, \eta}\) in both \(x \in \mathbb{R}^d\) and \(\eta>0\) small enough, since
    \begin{align}
        \abs{ T^\eps_{\nu, \eta}(x)} &\leq \nu\sum_{t\in \nu \mathbb{N}^*}\ee^{-\kappa t } \int \mathbb{W}^{t,0}_{x,x}(\d \omega)\,\Big|\ee^{-\int_0^t \d s\, U(\omega(s))}-\ee^{-\int_0^t \d s \,U^\eps_\eta(\omega(s))  }\Big| \nonumber \\&\leq \|\alpha^\eps_\eta-\tau^\eps\|_{L^\infty}\nu \sum_{t\in \nu \mathbb{N}^*} t\ee^{-\kappa t } \int \mathbb{W}^{t,0}_{x,x}(\d \omega) \nonumber\\
        & \lesssim \kappa^{d/2-2}\|\alpha^\eps_\eta-\tau^\eps\|_{L^\infty}\,.\label{eq:uniform upper bound on Ewicketanu x}
    \end{align}
    for every \(\eps>0\), where we used a Riemann sum estimate in the last line together with Lemma \ref{Lemma:bounds_on_the_law_of_brownian_bridge} (i). Proceeding as in \eqref{eq:epsEwickEwick_eta_nu}, we find 
    \begin{equation}
         \biggl(\int \mu_{v^\eps_\eta}(\d {\xi}) \abs{ \langle\xi,T^\eps_\eta- T^\eps_{\nu, \eta}\rangle}\biggr)^2\lesssim  \|\alpha^\eps_\eta-\tau^\eps\|_{L^\infty}  \|T^\eps_\eta- T^\eps_{ \nu,\eta}\|_{L^1}
    \end{equation}
    which follows from \eqref{eq:uniform upper bound on Ewicketanu x} and Lemma \ref{lemma:convergence_of_EWick} \ref{lemma:convergence_of_Ewick_conv_of_EN}. Hence, \eqref{eq:conv of EWicks} follows if we can show
    \begin{equation}\label{eq:diffTeps_Tepsnu}
        \|T^\eps_\eta- T^\eps_{ \nu,\eta}\|_{L^1} \lesssim_{s}\|\alpha^\eps_\eta-\tau^\eps\|_{L^\infty} \nu^{1-\frac{s}{2}} \,.
    \end{equation}
    By duality, we have 
    \[\|T^\eps_\eta- T^\eps_{ \nu,\eta}\|_{L^1} = \sup_{\|\varphi\|_{L^\infty}\leq 1}\absbb{\int \d{x}\, (T^\eps_\eta- T^\eps_{ \nu,\eta})(x)\varphi(x)}\, .\]
     For a fixed \(\varphi \in \mathcal{H}\) with \(\|\varphi\|_{L^\infty}\leq 1\), we consider the expression inside the integral  
    \begin{equation}\label{eq:diff_pointwise_Teta_Tnueta}
        T^\eps_{ \eta}(x)-T^\eps_{\nu, \eta}(x)=
         \int \d t \, \bigl(\ee^{-t h^U}- \ee^{-t h^{U^\eps_\eta}} \bigr)(x,x)  - \nu\sum_{t\in \nu \mathbb{N}} \bigl(\ee^{-t h^U}- \ee^{-t h^{U^\eps_\eta}} \bigr)(x,x) \, . 
    \end{equation}
    Thus, to bound \(\abs{\int \d{x}\, (T^\eps_\eta- T^\eps_{ \nu,\eta})(x)\varphi(x)}\), we can  estimate the two terms on the right-hand side of \eqref{eq:diff_pointwise_Teta_Tnueta} considering small and large times. For small times \(t\leq\delta\) we have
    \begin{align}\label{eq:Tsmalltime}
       \bigg| \int \d x\, \nu\sum_{\substack{t\in \nu \mathbb{N}\\t \leq \delta}} \bigl(\ee^{-t h^U}- \ee^{-t h^{U^\eps_\eta}} \bigr)(x,x)\varphi(x)\bigg|&\lesssim_s \|\alpha^\eps_\eta-\tau^\eps\|_{L^\infty} \delta^{2-s}
    \end{align}
    and a similar estimate can be obtained if we replace the Riemann sum by the corresponding integral. For large times \(t>\delta\), we compare the differences between the Riemann sums and the corresponding integral and find 
    \begin{align}\label{eq:Tlargetime}
       \abs{ \int \d x \int_\delta^\infty \d t\, \bigl(\ee^{-t h^U}- \ee^{-[t]_\nu h^{U}} \bigr)(x,x)\varphi(x)}&\lesssim_s \nu \delta^{-s}
    \end{align}
    where we used the notation \([t]_\nu\deq \min\{s \in \nu \mathbb{N^*} \col s<t\}\). The same bound holds if we replace \(h^U \) in \eqref{eq:Tlargetime} by \(h^{U^\eps_\eta}\). Optimizing over the size \(\delta \sim \nu^{1/2}\) yields the claim \eqref{eq:diffTeps_Tepsnu}. 
    To conclude the proof of \eqref{eq:conv of EWicks}, it remains to show \eqref{eq:Tsmalltime} and \eqref{eq:Tlargetime}. We start with \eqref{eq:Tsmalltime},  using Duhamel's formula and \eqref{eq: heat kernel of U tilde upper bound}, we have 
    \begin{align}
        \Biggr|\int \d x\, &\bigl(\ee^{-t h^U}- \ee^{-t h^{U^\eps_\eta}} \bigr)(x,x)\varphi(x)\Biggl|& \\& =\abs{\int \d x \int_0^1 \d \lambda\, \bigl(\ee^{-\lambda t h^U}t(\alpha^\eps_\eta - \tau^\eps) \ee^{-(1-\lambda)t h^{U^\eps_\eta}} \bigr)(x,x)\varphi(x)} \nonumber\\
        &\leq t\|\alpha^\eps_\eta - \tau^\eps\|_{L^\infty} \|\varphi\|_{L ^\infty} \int_0^1 \d \lambda \int \d x\, \d{\tilde{x}}\, \ee^{-\lambda t h^{U/2}}(x,\tilde{x}) \ee^{-(1-\lambda)t h^{U/2}}(\tilde{x}, x)   \nonumber\\
        &\leq t \|\alpha^\eps_\eta - \tau^\eps\|_{L^\infty} \|\varphi\|_{L ^\infty} \tr\bigl({\ee^{-t h^{U/2}}\bigr)}\,.\label{eq:bound_on_trace_of_diff}
    \end{align}
 Thus, using the bound 
    \[\tr{\ee^{-th^{U/2}}}\lesssim_st^{-s}\tr{h^{-s}}\lesssim_s t^{-s},\]
    and \eqref{eq:bound_on_trace_of_diff} we finally have that the left-hand side of \eqref{eq:Tsmalltime} is bounded by 
    \[\eqref{eq:Tsmalltime} \leq C_s\|\alpha^\eps_\eta - \tau^\eps\|_{L^\infty} \nu\sum_{\substack{t\in \nu \mathbb{N}\\t \leq \delta}} t^{1-s} \lesssim_s\|\alpha^\eps_\eta - \tau^\eps\|_{L^\infty}\delta^{2-s}\,,\]
    where in the last step we performed a Riemann sum approximation. 

    As for  \eqref{eq:Tlargetime}, using Hölder's inequality, we have
    \begin{align}
        \abs{\int\d x\, \bigl(\ee^{-t h^U}- \ee^{-[t]_\nu h^{U}} \bigr)(x,x)\varphi(x)} &\leq \int_{[t]_\nu}^t\d r\,\abs{\tr\bigl({h\ee^{-rh^U}\varphi}\bigr)}\nonumber\\
        &\leq \|\varphi\|_{L^\infty}\int_{[t]_\nu}^t\d r\, \tr\bigl({h \ee^{-rh}\bigr)}\nonumber\\
        &\lesssim_s\int_{[t]_\nu}^t \d r\, \tr\Bigl({h \frac{1}{h^{1+s}}}\Bigr)\frac{1}{r^{1+s}} \lesssim_s \nu t^{-1-s}\,,\label{eq:diff Riemann sums}
    \end{align}
    \eqref{eq:Tlargetime} then follows by integrating \eqref{eq:diff Riemann sums} with respect to \(t\in (\delta, \infty)\).  
\end{proof}
    
We now prove the estimate for the term involving the constants $S_\eta^\eps$ and $S^\eps_{\nu, \eta}$ defined in \eqref{eq:S_eta} and \eqref{eq:def_S_nueta} respectively. 

\begin{lemma}\label{lemma:Convergence of K_eta}
Let $S_\eta^\eps$ and $S^\eps_{\nu, \eta}$ be as in \eqref{eq:S_eta} and \eqref{eq:def_S_nueta}. Then, under the assumptions of Lemma \ref{lemma:Riemann sum estimates}
\begin{equation}\label{eq:diffS}  \abs{S^\eps_\eta-S^\eps_{\nu,\eta}}\lesssim_{\kappa, s} \|\alpha^\eps_\eta-\tau^\eps\|^2_{L^\infty} \nu^{2-s}\,.\end{equation}
\end{lemma}
\begin{proof}
    From \eqref{eq:trace_pho_tau}, we can write
    \begin{align}
        S_\eta^\eps & = \int_0 ^\infty \d t \,\tr\Bigl({\frac{1}{t + h^ {U^{\eps}_\eta }}-\frac{1}{t+h^U } + \frac{1}{t+h^U }(\alpha^\eps_\eta- \tau^\eps) \frac{1}{t+h^U }  }\Bigr)\nonumber \\
            & = \int_0 ^\infty \d t \,\tr{\Bigl( \frac{1}{t+h^U }(\alpha^\eps_\eta- \tau^\eps) \frac{1}{t+h^U } (\alpha^\eps_\eta- \tau^\eps) \frac{1}{t+h^{U^\eps_\eta} } \Bigr)} \nonumber\\
            & = \int_0 ^\infty \d t \, \int_{[0, \infty)^3} \d {\mathbf{r}} \, \tr{\Bigl( \ee^{- r_1(t+h^U)}(\alpha^\eps_\eta- \tau^\eps) \ee^{- r_2(t+h^U)} (\alpha^\eps_\eta- \tau^\eps) \ee^{- r_3(t+h^{U^\eps_\eta})} \Bigr)} \label{eq:rewriting Seta}
    \end{align}
    where \(\abs{\mathbf{r}} = r_1 + r_2 +r_3\). From \cite[proof of lemma 5.4]{FKSS_2020} with a re-indexing of the variables \(\tau_i, i=1,2,3\), we can write \(S^\eps_{\nu,\eta}\) similarly as
    \begin{multline}\label{eq:Snueta from mf} 
        S^\eps_{\nu, \eta}= \int_0 ^\infty \d t \int_{[0, \nu]^3}\d {\mathbf \tau} \sum_{\substack{ \mathbf{q}\in (\nu\mathbb{N})^3\\\abs{\mathbf{q}}>0}}  \tr\Bigl( \ee^{- (q_1+\tau_1-\tau_2)(t+h^U)}(\alpha^\eps_\eta- \tau^\eps)\ee^{- (q_2+\tau_2-\tau_3)(t+h^U)}\\ \times (\alpha^\eps_\eta- \tau^\eps) \ee^{- (q_3+\tau_3-\tau_1)(t+h^{U^\eps_\eta})} \Bigr)\,. 
    \end{multline}
    By performing the change of variables \(r_i = q_i + \tau_i\) and denoting \(\{r_i\}_\nu= r_i-[r_i]_\nu \in[0, \nu) \), we can rewrite \eqref{eq:Snueta from mf} as
    \begin{multline}\label{eq:Snuetareindexed} 
        S^\eps_{\nu, \eta}= \int_0 ^\infty \d t \, \int_{ [0, \infty)^3}\d {\mathbf{r}}\, \mathbbold{1}_{\abs{r}\geq\nu} \tr\Bigl( \ee^{- (r_1-\{r_2\}_{\nu})(t+h^U)}(\alpha^\eps_\eta- \tau^\eps)\ee^{- (r_2-\{r_3\}_{\nu})(t+h^U)}\\ \times (\alpha^\eps_\eta- \tau^\eps) \ee^{- (r_3-\{r_1\}_{\nu})(t+h^{U^\eps_\eta})} \Bigr)\,. 
    \end{multline}
    We proceed as in the proof of Lemma \ref{lemma:convergence_of_Ewick_quantitative} splitting the estimate in small and large times: first, we restrict the integration to the case where one of the time variables \(r_i\) lies in the interval \([0, \delta]\), where  \(\delta\equiv\delta(\nu)= C_0\nu\) and \(C_0\geq2\) is a positive constant. Then we consider when the time variables \(r_i> \delta\). In the first case we have the bounds for $S_\eta^\eps$ 
    \begin{multline}\label{eq:small time estimate S_eta}
         \absbb{\int_0 ^\infty \d t\,\int_{[0, \infty)^3} \d {\mathbf{r}} \, \mathbbold{1}_{\min_i {r_i}\leq \delta} \,  \tr\Bigl({ \ee^{- r_1(t+h^U)}(\alpha^\eps_\eta- \tau^\eps) \ee^{- r_2(t+h^U)} (\alpha^\eps_\eta- \tau^\eps) \ee^{- r_3(t+h^{U^\eps_\eta})}}\Bigr)}\\ \lesssim_s \|\alpha^\eps_\eta-\tau^\eps\|^2_{L^\infty}\delta^{2-s}\,,
    \end{multline}    
    and for $S_{\nu, \eta}^\eps$
    \begin{multline}\label{eq:small time estimate S_eta_nu}
        \Biggl|\int_0 ^\infty \d t\, \int_{ [0, \infty)^3} \d {\mathbf{r}} \,\mathbbold{1}_{\nu\leq\min_i {r_i}\leq \delta}\\ \times \tr\Bigl( \ee^{- (r_1-\{r_2\}_{\nu})(t+h^U)}(\alpha^\eps_\eta- \tau^\eps) \ee^{- (r_2-\{r_3\}_{\nu})(t+h^U)} (\alpha^\eps_\eta- \tau^\eps) \ee^{- (r_3-\{r_1\}_{\nu})(t+h^{U^\eps_\eta})} \Bigr)\Biggr|\\\lesssim_s \|\alpha^\eps_\eta-\tau^\eps\|^2_{L^\infty}\delta^{2-s}\,.
    \end{multline}
    On the other hand, for large times, we have
    \begin{multline}\label{eq:large time diff S}
          \Biggl|\int_0 ^\infty \d t\, \int_{\mathbf{r}\in [\delta, \infty)^3} \d {\mathbf{r}} \,\tr\Bigl( \ee^{- (r_1-\{r_2\}_{\nu})(t+h^U)}(\alpha^\eps_\eta- \tau^\eps) \ee^{- (r_2-\{r_3\}_{\nu})(t+h^U)} (\alpha^\eps_\eta- \tau^\eps) \\\times\bigl(\ee^{- r_3(t+h^{U^\eps_\eta})}-\ee^{- (r_3-\{r_1\}_{\nu})(t+h^{U^\eps_\eta})}\bigr) \Bigr)\Biggr| \lesssim_s \|\alpha^\eps_\eta-\tau^\eps\|^2_{L^\infty}\nu(\delta-\nu)^{1-s}\,.
    \end{multline}
    Hence, the proof of \eqref{eq:diffS} follows from \eqref{eq:small time estimate S_eta}, \eqref{eq:small time estimate S_eta_nu} and by using a telescoping argument, applying \eqref{eq:large time diff S} to each pair of terms. 
    
    Now we need to establish \eqref{eq:small time estimate S_eta}, \eqref{eq:small time estimate S_eta_nu} and \eqref{eq:large time diff S}. For \eqref{eq:small time estimate S_eta}, it is sufficient to show 
    \begin{multline}\label{eq:small time estimate S_eta_reduced claim}
        \absbb{\int_0 ^\infty \d t \,\int_{[0, \infty)^3} \d {\mathbf{r}}\,  \mathbbold{1}_{{r_3}\leq \delta}  \tr{ \ee^{- r_1(t+h^U)}(\alpha^\eps_\eta- \tau^\eps) \ee^{- r_2(t+h^U)} (\alpha^\eps_\eta- \tau^\eps) \ee^{- r_3(t+h^{U^\eps_\eta})}}}\\ \lesssim_s \|\alpha^\eps_\eta-\tau^\eps\|^2_{L^\infty}\delta^{2-s}\,. 
    \end{multline}
    From the bound in \eqref{eq: heat kernel of U tilde upper bound}, we find that the left-hand side of \eqref{eq:small time estimate S_eta_reduced claim} is bounded by 
    \begin{align}
        \begin{split}
             &  \|\alpha^\eps_\eta-\tau^\eps\|^2_{L^\infty}\int_0 ^\infty \d t \,\int_{[0, \infty)^3} \d {\mathbf{r}}\,  \mathbbold{1}_{{r_3}\leq \delta} \\&\qquad \qquad \qquad\qquad \times\int \d {\mathbf {x}}\,  \ee^{- r_1(t+h^{U/2})}(x_1,x_2) \ee^{- r_2(t+h^{U/2})}(x_2,x_3)\ee^{- r_3(t+h^{U/2})}(x_3,x_1) \\
             &\quad\leq \|\alpha^\eps_\eta-\tau^\eps\|^2_{L^\infty}\int_0 ^\infty \d t \,\int_{[0, \infty)^3} \d {\mathbf{r}}\,  \mathbbold{1}_{ {r_3}\leq \delta} \tr{ \ee^{- r_1(t+h^{U/2})} \ee^{- r_2(t+h^{U/2})}\ee^{- r_3(t+h^{{U/2}})}}    \\    
            &\quad = \|\alpha^\eps_\eta-\tau^\eps\|^2_{L^\infty}\int_0 ^\infty \d t\, \tr\Bigl({\frac{1}{(t+h^{U/2})^3}(1-\ee^{-\delta(t+ h^{U/2})})}\Bigr) \, .
        \end{split}\label{eq:small time estimate S_eta_reduced claim bound before t integration}
    \end{align}
For fixed \(\lambda>0\), using the fact that \(h^{U/2}\) is a positive operator and the Hölder's inequality for Schatten norms, we find
 \begin{equation}\label{eq:small time estimate S_eta_reduced claim bound trace}
 \begin{split}
\tr\biggl(&\frac{1}{(t+h^{U/2})^3}(1-\ee^{-\delta(t+ h^{U/2})})\biggr) \\& =  \tr\biggl(\frac{1}{(t+h^{U/2})^3}(1-\ee^{-\delta(t+ h^{U/2})})(\mathbbold{1}_{t+h^{U/2}\leq \lambda}+\mathbbold{1}_{h^{U/2}+t>\lambda})\biggr)\\
     &\leq\delta\tr\biggl({\frac{1}{(t + h^{U/2})^{2}}\mathbbold{1}_{t+h^{U/2}\leq \lambda}}\biggr) + \tr\biggl({\frac{1}{(t+h^{U/2})^{3}}\mathbbold{1}_{h^{U/2}+t>\lambda}}\biggr)\\
         &\leq\delta\tr\biggl({\frac{1}{(t + h^{U/2})^{1+s}}(t+h^{U/2})^{s-1}\mathbbold{1}_{t+h^{U/2}\leq \lambda}} \biggr)\\&\qquad\qquad\qquad\qquad\qquad\qquad\qquad + \tr\biggl({\frac{1}{(t+h^{U/2})^{1+s}}(t+h^{U/2})^{s-2}\mathbbold{1}_{h^{U/2}+t>\lambda}}\biggr)\\
     & \leq \tr\biggl({\frac{1}{(t + h^{U/2})^{1+s}}}\biggr)(\delta\lambda^{s-1}  + \lambda^{s-2})\,.
     \end{split}
 \end{equation}
 Choosing \(\lambda=\delta^{-1}\) in \eqref{eq:small time estimate S_eta_reduced claim bound trace}, we get 
 \begin{align}
     \eqref{eq:small time estimate S_eta_reduced claim bound before t integration} \leq \|\alpha^\eps_\eta-\tau^\eps\|^2_{L^\infty}\delta^{2-s}\int_0 ^\infty \d t\, \tr\biggl({\frac{1}{(t + h^{U/2})^{1+s}}}\biggr)& = \|\alpha^\eps_\eta-\tau^\eps\|^2_{L^\infty}\delta^{2-s}\tr\biggl({\frac{1}{ (h^{U/2})^{s}}\biggr)}\nonumber\\
     & \lesssim_{s} \|\alpha^\eps_\eta-\tau^\eps\|^2_{L^\infty}\delta^{2-s}\,, \nonumber
 \end{align}
which proves \eqref{eq:small time estimate S_eta_reduced claim}. 

In a similar way, we can establish \eqref{eq:small time estimate S_eta_nu} by showing the bound
\begin{equation} \label{eq:bound on S_nu sum small time}\begin{split}
     &\int_0 ^\infty \d t \int_{ [0, \infty)^3} \d {\mathbf{r}}\,\mathbbold{1}_{\abs{r}\geq\nu} \mathbbold{1}_{r_3\leq \delta} \\
     &\hspace{3.5cm}\times \tr\Bigl( \ee^{- (r_1-\{r_2\}_{\nu})(t+h^{U/2})}\ee^{- (r_2-\{r_3\}_{\nu})(t+h^{U/2})}\ee^{- (r_3-\{r_1\}_{\nu})(t+h^{U/2})} \Bigr)\\
     & \quad =  \int_0 ^\infty \d t\, \nu^3\sum_{\mathbf{r}\in {\nu \mathbb{N}}^3 }  \mathbbold{1}_{\abs{r}\geq\nu}\mathbbold{1}_{r_3\leq \delta} \tr\bigl(\ee^{-\abs{\mathbf r}(t+h^{U/2})}\bigr) \\& \quad \leq \int_0 ^\infty \d t \,\nu^3\sum_{r_3\in {\nu \mathbb{N}^*} }  \tr\bigl(\ee^{-\abs{r_3}(t+h^{U/2})}
     \bigr)+ \int_0 ^\infty \d t \,\nu^3\sum_{\mathbf{r}\in {\nu \mathbb{N}}^3 } \mathbbold{1}_{r_3\leq \delta}\mathbbold{1}_{r_1+r_2\geq\nu}  \tr\bigl(\ee^{-\abs{\mathbf{r}}(t+h^{U/2})}\bigr)\\
     & \quad \lesssim C_s\nu^3 \sum_{r_3\in {\nu \mathbb{N}^*} } \frac{1}{r_3 ^{1+s}}\int_0 ^\infty \d t\,\tr\biggl(\frac{1}{(t+h^{U/2})^{1+s}}\biggr) \\
     &\quad \quad + 
     C_s \delta \nu^2 \sum_{r_1, r_2\in {\nu \mathbb{N}^*} }\frac{1}{{(r_1+r_2)} ^{1+s}} \int_0^\infty \d t \,\tr\biggl(\frac{1}{(t+h^{U/2})^{1+s}}\biggr) \lesssim_{s}\nu^{2-s} +\delta\nu^{1-s}\lesssim_s \delta ^{2-s}\,, 
\end{split}
\end{equation}
where in the last step we used that \(\tr\bigl(\ee^{-r_3h}\bigr)\leq C_s\tr\Bigl(\frac{1}{(r_3h)^{1+s}}\Bigr)\) to bound the two  terms. 
To conclude, we need to show \eqref{eq:large time diff S}. Using Hölder's inequality, the fact that, for a multiplication operator  \(f\), we have\footnote{We use $\norm{\cdot}_{\fra S^p}$ to denote the $p$-th Schatten norm.} \(\|f\|_{\mathfrak{S^\infty}}=\|f\|_{L^\infty}\),  and for \(h>0\) we have $\int_s^t \d\lambda \, h\, \ee^{-\lambda h }= \ee^{-s h}- \ee^{-t h}$, we can bound the left-hand side of \eqref{eq:large time diff S} by
\begin{multline}\label{eq:bound traces S3}
    \|\alpha^\eps_\eta- \tau^\eps\|_{L^\infty}^2\int_0 ^\infty \d t\int_{ [\delta, \infty)^3} \d {\mathbf{r}} \,\ee^{-{(|\mathbf r|-3 \nu)\kappa}} \\\times \big\|\ee^{- (r_1-\{r_2\}_{\nu})(t-\Delta/2 + U)}\big\|_{\mathfrak{S}^3} \big\|\ee^{- (r_2-\{r_3\}_{\nu})(t-\Delta/2 + U)}\big\|_{\mathfrak{S}^3}\\ \times \int_{0}^{\{r_1\}_\nu} \d\lambda\, \big\|(t-\Delta/2 + {U^\eps_\eta})\ee^{- (r_3-\lambda)(t-\Delta/2 + U^\eps_\eta)}\big\|_{\mathfrak{S}^3}\,.  
\end{multline}
Since \(\ee^{-3x}x^3 \leq \ee^{-x}\) for all \(x\in \mathbb{R}\), we find 
\begin{equation}\label{eq:firsttermshatten}
\begin{split}
    &\int_{0}^{\{r_1\}_\nu} \d\lambda\, \big\|(t-\Delta/2 + U^\eps_\eta )\ee^{- (r_3-\lambda)(t-\Delta/2 + U^\eps_\eta)}\big\|_{\mathfrak{S}^3} \\&\quad\deq\int_{0}^{\{r_1\}_\nu} \d\lambda\,\tr\Bigl({\bigl(t-\Delta/2 + U^\eps_\eta\bigr)^3\ee^{- 3(r_3-\lambda)(t-\Delta/2 + U^\eps_\eta)}}\Bigr)^{1/3}\\
    &\quad \leq \nu \frac{1}{r_3-\nu} \tr{\Bigl(\ee^{- (r_3-\nu)(t-\Delta/2 + U^\eps_\eta)}\Bigr)}^{1/3}\,. 
    \end{split}
\end{equation}
Moreover, using  the bound
\begin{multline*}
\tr{\Bigl( \ee^{- (r_i-\nu)(t-\Delta/2 + U)}\Bigr)}^{1/3}\vee\tr{\Bigl( \ee^{- (r_i-\nu)(t-\Delta/2 + U^\eps_\eta)}\Bigr)}^{1/3}\lesssim_{a} \frac{1}{(r_i-\nu)^a}\tr\biggl({\frac{1}{(t-\Delta/2 + U/2)^{3a}}}\biggr)^{1/3}
\end{multline*}
with \(a=\frac{s+1}{3}\), together with \eqref{eq:firsttermshatten} we can conclude 
 \begin{align*}
     \eqref{eq:bound traces S3} &\lesssim_{s}  \nu\|\alpha^\eps_\eta- \tau^\eps\|_{L^\infty}^2 \int_{\mathbf{r}\in [\delta-\nu, \infty)^3} \d {\mathbf{r}}\, \frac{\ee^{-\abs{\mathbf r}\kappa}}{r_1^{a}r_2^{a}r_3^{1+a}} \int_0 ^\infty \d t\,   \tr{\frac{1}{(t-\Delta/2 + U/2)^{3a}}}\\&\lesssim_{\kappa,s} \|\alpha^\eps_\eta- \tau^\eps\|_{L^\infty}^2 \nu(\delta-\nu)^{1-s}\,. \qedhere
 \end{align*}
 
 \end{proof}

\subsubsection{Proof of Lemma \ref{lemma:bound on diff F2 and f2}}
We finally prove Lemma \ref{lemma:bound on diff F2 and f2}. We proceed by introducing two lemmas which are similar to \cite[Lemma 5.10-5.11]{FKSS_24}. The two results below show some quantitative estimates for the quantities $F_2(\sigma)$ and $f_2(\xi)$ as defined in \eqref{def:f_2} and \eqref{eq:def_of_F_2}.

\begin{lemma}\label{lemma:FFff}
    Uniformly in \(\eta>0\), we have 
    \[\bigg|\int \mu_{\mathcal{C}_\eta}(\mathrm{d} \sigma)\, \overline{F_2(\sigma)} F_2(\sigma)-\int \mu_{v^{\varepsilon}}(\mathrm{d} \xi)\, \overline{f_2(\xi)} f_2(\xi)\bigg| \lesssim_{\kappa, v }\frac{\nu^{1/2}\vee \nu^{2-s}}{\varepsilon^{5 d+1}} \,.\]
\end{lemma}
\begin{lemma}\label{lemma:Ffff}
    Uniformly in \(\eta>0\), we have 
    \[\bigg|\int \mu_{\mathcal{C}_\eta}(\mathrm{d} \sigma)\, \overline{f_2(\langle\sigma\rangle)} F_2(\sigma)-\int \mu_{v^{\varepsilon}}(\mathrm{d} \xi) \,\overline{f_2(\xi)} f_2(\xi)\bigg| \lesssim_{\kappa, v }\frac{\nu^{1/2}\vee \nu^{2-s}}{\varepsilon^{5 d+1}} \,.\]
\end{lemma}
Since the proofs of Lemmas \ref{lemma:FFff} and \ref{lemma:Ffff} are similar, we focus our discussion only on showing    Lemma \ref{lemma:FFff}. We start by recalling some useful notation from \cite[Section 5.1]{FKSS_24}. 
\begin{definition}[Classical interactions]
    Let \(x, \tilde{x} \in \mathbb{R}^d\) and \(\omega \in \Omega^{\tau_1, \tilde{\tau}_1}\), \(\tilde \omega \in \Omega^{\tau_2, \tilde{\tau}_2}\) be continuous paths. We define the \emph{point-point interaction}
    \[
\bigl(\mathbb{V}^{\varepsilon}\bigr)_{x, \tilde{x}}\deq\int \mu_{v^{\varepsilon}}(\d \xi)\, \xi(x) \xi(\tilde{x})=v^{\varepsilon}(x-\tilde{x})
\]
the \emph{point-path interaction}
\[
\bigl(\mathbb{V}^{\varepsilon}\bigr)_x(\omega)\deq\int \mu_{v^{\varepsilon}}(\d \xi) \int \d s\, \xi(\omega(s))=\int \d s\, v^{\varepsilon}(x-\omega(s))
\]
and the \emph{path-path interaction}
\[
\mathbb{V}^{\varepsilon}(\omega, \tilde{\omega})\deq\int \mu_{v^{\varepsilon}}(\d \xi) \int \d s\, \xi(\omega(s)) \int \d{\tilde{s}}\, \xi(\tilde{\omega}(\tilde{s}))=\int \d s \int \d  {\tilde{s}} \,v^{\varepsilon}(\omega(s)-\tilde{\omega}(\tilde{s}))\,.
\]

\end{definition}

We use the notation:  \(x_{i,0}\equiv x_i, \tilde{x}_{i,0}\equiv \tilde{x}_i\) for \(i=1,2\), and  \[A\deq \{2,3\}\times \{0,1\}.\]
Following \cite[(5.8)-(5.9)]{FKSS_2020}, we have 
\begin{equation}\label{eq:int_barf2f2}
    \int \mu_{v^{\varepsilon}}(\mathrm{d} \xi)\, \overline{f_2(\xi)} f_2(\xi)=\int_{[0, \infty)^3} \mathrm{d} \mathbf{r} \,\frac{\mathrm{e}^{-\kappa|\mathbf{r}|}}{|\mathbf{r}|} \int_{[0, \infty)^3} \mathrm{d} \tilde{\mathbf{r}}\, \frac{\mathrm{e}^{-\kappa|\tilde{\mathbf{r}}|}}{|\tilde{\mathbf{r}}|} I^{\varepsilon}(\mathbf{r}, \tilde{\mathbf{r}})
\end{equation}
where 
\begin{equation}\label{eq:def:Ieps}
    \begin{aligned}
I^{\varepsilon}(\mathbf{r}, \tilde{\mathbf{r}})\deq & \int \mathrm{d} \mathbf{x} \int \mathrm{d} \tilde{\mathbf{x}} \int\WU_{x_1, x_3}^{r_3, 0}(\mathrm{d} \omega_3)\WU_{x_3, x_2}^{r_2, 0}(\mathrm{d} \omega_2)\WU_{x_2, x_1}^{r_1, 0}(\mathrm{d} \omega_1)\WU_{\tilde{x}_1, \tilde{x}_3}^{\tilde{r}_3, 0}(\mathrm{d} \tilde{\omega}_3) \\
& \quad\quad  \times\WU_{\tilde{x}_3, \tilde{x}_2}^{\tilde{r}_2, 0}(\mathrm{d} \tilde{\omega}_2)\WU_{\tilde{x}_2, \tilde{x}_1}^{\tilde{r}_1,0}(\mathrm{d} \tilde{\omega}_1) \mathrm{e}^{-\frac{1}{2}\bigl(\mathbb{V}^{\varepsilon}(\omega_1, \omega_1)+\mathbb{V}^{\varepsilon}(\tilde{\omega}_1, \tilde{\omega}_1)-2 \mathbb{V}^{\varepsilon}(\omega_1, \tilde{\omega}_1)\bigr)} \\
& \quad\quad  \times \sum_{\Pi \in \mathfrak{M}(A)} \prod_{\{a, b\} \in \Pi}(\mathbb{V}^{\varepsilon})_{x_a, x_b} \prod_{a \in A \backslash[\Pi]} \mathrm{i}\bigl((\mathbb{V}^{\varepsilon})_{x_a}(\omega_1)-(\mathbb{V}^{\varepsilon})_{x_a}(\tilde{\omega}_1)\bigr)\,,
\end{aligned}
\end{equation}
Here we defined $[\Pi] \deq \{1,\dots,\Pi\}$, and $\mathfrak{M}(A)$ the set of partial pairings of the set $A$. Moreover, we introduced the notation 
\begin{equation}\label{eq:def:WU}
    \WU^{t,s}_{y,x}(\d \omega)\deq{\mathbb{W}}^{t,s}_{y,x}(\d \omega)\ee^{-\int_s^t \d r\,U^{\eps}_\eta(\omega(r))}\eqd\psi^{t-s}(x-y){\PU}^{t,s}_{y,x}(\d \omega)\,.
\end{equation}

\begin{definition}[Quantum interactions]\label{def:quantum_interactins}
Let $(\tau, x),(\tilde{\tau}, \tilde{x}) \in[0, \nu] \times \mathbb{R}^d$ and let $\omega \in \Omega^{\tau_1, \tilde{\tau}_1}, \tilde{\omega} \in$ $\Omega^{\tau_2, \tilde{\tau}_2}$ be continuous paths. With $\delta_\eta$ given by  \eqref{eq:defdelta}, we define the \emph{point-point interaction}
    \[
\bigl(\mathbb{V}_\eta\bigr)_{x, \tilde{x}}^{\tau, \tilde{\tau}}\deq\int \mu_{\mathcal{C}_\eta}(\mathrm{d} \sigma) \,\sigma(\tau, x) \sigma(\tilde{\tau}, \tilde{x})=\nu \delta_\eta(\tau-\tilde{\tau}) v^{\varepsilon}(x-\tilde{x})\,,
\]
the \emph{point-point interaction}
\begin{align*}
\bigl(\mathbb{V}_\eta\bigr)_x^\tau(\omega)& \deq\int \mu_{\mathcal{C}_\eta}(\mathrm{d} \sigma)\, \sigma(\tau, x) \int_0^\nu \mathrm{d} t \int \mathrm{d} s \,\delta(t-[s]_\nu) \sigma(t, \omega(s))\\ & =\nu \int \mathrm{d} s\, \delta_\eta(\tau-[s]_\nu) v^{\varepsilon}(x-\omega(s))\,,
\end{align*}
and the \emph{path-path interaction}
\begin{align*}
    \mathbb{V}_\eta(\omega, \tilde{\omega})&\deq\int \mu_{\mathcal{C}_\eta}(\mathrm{d} \sigma) \int_0^\nu \d t \int \mathrm{d} s \, \delta(t-[s]_\nu) \sigma(t, \omega(s)) \int_0^\nu \d {\tilde{t}} \int \d {\tilde{s}} \, \delta(\tilde{t}-[\tilde{s}]_\nu) \sigma(\tilde{t}, \tilde{\omega}(\tilde{s})) \\
& =\nu \int \d {s} \int \d{\tilde{s}} \,\delta_\eta([s]_\nu-[\tilde{s}]_\nu) v^{\varepsilon}(\omega(s)-\tilde{\omega}(\tilde{s}))\,.
\end{align*}

\end{definition}

Recalling \eqref{eq:def_of_F_2}, and proceeding as in \cite[(5.17)–(5.18), (5.21)]{FKSS_2020}, we have 
\begin{equation}\label{eq:int_barF2F2}
    \int \mu_{\mathcal{C}_\eta}(\mathrm{d} \sigma)\, \overline{F_2(\sigma)} F_2(\sigma)=\sum_{\mathbf{r} \in(\nu \mathbb{N})^3} \frac{\mathbbold{1}_{|\mathbf{r}|>0} \mathrm{e}^{-\kappa|\mathbf{r}|}}{|\mathbf{r}|} \sum_{\tilde{\mathbf{r}} \in(\nu \mathbb{N})^3} \frac{\mathbbold{1}_{|\tilde{\mathbf{r}}|>0} \mathrm{e}^{-\kappa|\tilde{\mathbf{r}}|}}{|\tilde{\mathbf{r}}|} J(\mathbf{r}, \tilde{\mathbf{r}})
\end{equation}
where 
\begin{equation}\label{eq:def:J}
    \begin{aligned}
& J(\mathbf{r}, \tilde{\mathbf{r}})\\&\deq\int_{[0, \nu]^3} \d {\boldsymbol{\tau}} \int_{[0, \nu]^3} \d {\tilde{\boldsymbol{\tau}}} \int \d {\mathbf{x}} \int \d{ \tilde{\mathbf{x}} }\int\WU_{x_1, x_3}^{\tau_1+r_3, \tau_3}(\d{ \omega_3})\WU_{x_3, x_2}^{\tau_3+r_2, \tau_2}(\d{ \omega_2})\WU_{x_2, x_1}^{\tau_2+r_1, \tau_1}(\d {\omega_1}) \\
& \quad \times\WU_{\tilde{x}_1, \tilde{x}_3}^{\tilde{\tau}_1+\tilde{r}_3, \tilde{\tau}_3}(\d {\tilde{\omega}_3})\WU_{\tilde{x}_3, \tilde{x}_2}^{\tilde{\tau}_3+\tilde{r}_2, \tilde{\tau}_2}(\d {\tilde{\omega}_2})\WU_{\tilde{x}_2, \tilde{x}_1}^{\tilde{\tau}_2+\tilde{r}_1, \tilde{\tau}_1}(\d {\tilde{\omega}_1})  \ee^{-\frac{1}{2}(\mathbb{V}_\eta(\omega_1, \omega_1)+\mathbb{V}_\eta(\tilde{\omega}_1, \tilde{\omega}_1)-2 \mathbb{V}_\eta(\omega_1, \tilde{\omega}_1))} \\
& \quad \times \sum_{\Pi \in \mathfrak{M}(A)} \prod_{\{a, b\} \in \Pi}(\mathbb{V}_\eta)_{x_a, x_b}^{\tau_a, \tau_b} \prod_{a \in A \backslash[\Pi]} \mathrm{i}\bigl((\mathbb{V}_\eta)_{x_a}^{\tau_a}(\omega_1)-(\mathbb{V}_\eta)_{x_a}^{\tau_a}(\tilde{\omega}_1)\bigr)\,.
\end{aligned}
\end{equation}

In order to prove Lemma \ref{lemma:FFff}, we first need to compare \eqref{eq:def:Ieps} with \eqref{eq:def:J}. As a result we have the following statement.
\begin{lemma}\label{eq:approximation_of_J}
    For all \(\mathbf{r}, \mathbf{\tilde r} \in (\nu \mathbb{N})^3\) with \(\abs{\mathbf{r}}, \abs{\mathbf{\tilde r}}>0\), uniformly in \(\eta >0\) we have 
    \begin{multline*}
        \big| J(\mathbf{r}, \tilde{\mathbf{r}})-\nu^6 I^{\varepsilon}(\mathbf{r}, \tilde{\mathbf{r}})\big|\\\lesssim \frac{\nu^{13/2}}{\eps^{5d + 1}}\Bigl(\tr\bigl(\ee^{-\abs{{\mathbf{r}}}(-\Delta/2 + U/2)}\bigr) +\frac{1}{|\mathbf r|^{d/2}}\Bigr)  \Bigl(\tr\bigl(\ee^{-|\tilde{\mathbf{r}|}(-\Delta/2 + U/2)}\bigr) +\frac{1}{|{\tilde{\mathbf r}|^{d/2}}}\Bigr) (1 + |{\mathbf{r}}|+ |\tilde{\mathbf{r}}|)^6\,.
    \end{multline*}
\end{lemma}

Finally, Lemma \ref{lemma:FFff} follows from controlling the difference between Riemann sums and the corresponding integral. 
\begin{lemma}\label{lemma:riemannsumestimates}
    Recall \eqref{eq:def:Ieps}. Under the assumptions of Remark \ref{rmk:thetha_s2}, we have that uniformly in $\eta >0$ 
    \begin{multline}\label{eq:diff riemann sums Ieps}
        \bigg|\int_{[0, \infty)^3} \mathrm{~d} \mathbf{r} \frac{\mathrm{e}^{-\kappa|\mathbf{r}|}}{|\mathbf{r}|} \int_{[0, \infty)^3} \mathrm{~d} \tilde{\mathbf{r}} \frac{\mathrm{e}^{-\kappa|\tilde{\mathbf{r}}|}}{|\tilde{\mathbf{r}}|} I^{\varepsilon}(\mathbf{r}, \tilde{\mathbf{r}})\\
        -\nu^6 \sum_{\mathbf{r} \in(\nu \mathbb{N})^3} \frac{\mathbbold{1}_{|\mathbf{r}|>0} \mathrm{e}^{-\kappa|\mathbf{r}|}}{|\mathbf{r}|} \sum_{\tilde{\mathbf{r}} \in(\nu \mathbb{N})^3} \frac{\mathbbold{1}_{|\tilde{\mathbf{r}}|>0} \mathrm{e}^{-\kappa|\tilde{\mathbf{r}}|}}{|\tilde{\mathbf{r}}|} I^{\varepsilon}(\mathbf{r}, \tilde{\mathbf{r}})\bigg| \lesssim_{\kappa, v} \frac{\nu^{2-s}}{\eps^{5d+1}}\, .
    \end{multline}
\end{lemma}
We dedicate the rest of the subsection to prove \eqref{eq:approximation_of_J} and Lemma \ref{lemma:riemannsumestimates}. 
To start we recall the following result from \cite{FKSS_24} regarding path-path interactions.
\begin{prop}[\text{\cite[lemma C.1]{FKSS_24}}]\label{prop:phi42:path-path_interaction_bounds}
    Consider two continuous paths \(\omega \in \Omega^{\tau_1, \tilde{\tau}_1}\), \(\tilde{\omega} \in\Omega^{\tau_2, \tilde{\tau}_2} \). Then, uniformly in \(\eta>0\), 
    \[\big|\mathbb{V}_{\eta}(\omega, \tilde \omega)- \bigl(\mathbb{V}^\eps(\omega, \tilde \omega) +\hat{\mathbb{V}}_\eta (\omega, \tilde \omega)\bigr)\big|\lesssim \frac{\nu}{\eps^d }(\tau_1-\tilde{\tau}_1)\]
    where 
    \begin{equation}\label{eq:bound_on_hatVeta}
        \big|\hat{\mathbb{V}}_\eta(\omega, \tilde{\omega})\big| \lesssim \frac{1}{\varepsilon^{d+1}} \int \mathrm{d} s \sum_{i=1}^N \int_{I_i} \d {\tilde{s}} \int_{I_i} \d {\hat{s}}\,|\tilde{\omega}(\tilde{s})-\tilde{\omega}(\hat{s})| \delta_\eta(\{s\}_\nu-\{\hat{s}\}_\nu)\,, 
    \end{equation}
     for \(N\deq\big \lceil\frac{\tau_2-\tilde{\tau}_2}{\nu} \big\rceil\) and 
     \[I_i=\begin{cases} 
         [\tilde{\tau}_2+(i-1)\nu,\tilde{\tau}_2+i\nu ), i=1, \dots ,N-1 \\
         [\tilde{\tau}_2+(N-1)\nu, \tau_2], i=N\,.
     \end{cases}\]
\end{prop}

The rest of this subsection closely follows the arguments of \cite[Section 5]{FKSS_24}, with any modifications arising from the need to bound quantities of the form 
 \begin{equation}\label{eq:diff between close time BB}
        \int\WU_{x, x}^{|\mathbf{r}|, 0}(\mathrm{d} \omega )\, |{\omega}(\tilde{s})-{\omega}(\hat{s})|
    \end{equation}
for times \(\hat s, \tilde{s}\) satisfying \(\abs{\hat s-\tilde{s}}\leq\nu\), by an integrable function of \(x \in \mathbb{R}^d\). For this purpose we show the following result. 

\begin{lemma}\label{lemma: bounds on V hat}
    Let \(\omega \in \Omega_{x_1, x_1}^{\tau_1+|\mathbf{r}|, \tau_1}\) be a Brownian loop of length \(\abs{r}\deq r_1+r_2+r_3\), with \(r_i\in \nu \mathbb{N}\) for \(i=1,2,3\),  and \({\omega}_1\deq\restr{{\omega}}{[{\tau}_1, {r}_1+ {\tau}_2]}\) its restriction between the times \(\tau_1\) and \(r_1 + \tau_2\). Then
    \begin{multline}\label{eq:first_error_term5.46}
    \int_{[0, \nu]^3} \mathrm{d} \boldsymbol{\tau} \int_{[0, \nu]^3} \mathrm{d} \tilde{\boldsymbol{\tau}} \int \d  x_1 \int \d {\tilde{x}_1}\int \WU_{x_1, x_1}^{\tau_1+|\mathbf{r}|, \tau_1}(\mathrm{d} \omega)\WU_{\tilde{x}_1, \tilde{x}_1}^{\tilde{\tau}_1+ |\tilde{\mathbf{r}}|, \tilde{\tau}_1}(\mathrm{d} \tilde{\omega})\big|\hat{\mathbb{V}}_\eta(\omega_1, \omega_1)\big| \\ \lesssim  \frac{\nu^{13/2}}{\eps^{d+1}}\tr{\Bigl(\ee^{-\abs{\tilde{\mathbf{r}}}(-\Delta/2 + U/2)}\Bigr)} \frac{(1+{|\mathbf{r}|)^2}}{|\mathbf r|^{d/2}}\,.\end{multline} 
\end{lemma}
\begin{proof}

     Using Lemma \ref{Feynman-Kac formula}, \eqref{eq:bound_on_hatVeta}, the definition of $\delta_\eta$ \eqref{eq:defdelta}, in particular \(\int \delta_{\eta} =1\), and Fubini's theorem, we find that the left-hand side of \eqref{eq:first_error_term5.46} is bounded by
    \begin{multline} 
    \label{eq:second_error_term5.46}
        \frac{C}{\eps^{d+1}}\tr\bigl(\ee^{-|\tilde{\mathbf{r}}|(-\Delta/2 + U/2)}\bigr)\\\times \int_{[0, \nu]^3} \mathrm{d} \boldsymbol{\tau}\int_{[0, \nu]^3} \mathrm{d} \tilde{\boldsymbol{\tau}}  \sum_{i=1}^N \int_{I_i} \mathrm{d} \tilde{s} \int_{I_i} \mathrm{d} \hat{s}\int \d {x_1 }\int \WU_{x_1, x_1}^{\tau_1+|\mathbf{r}|, \tau_1}(\mathrm{d} \omega ) |{\omega_1}(\tilde{s})-{\omega_1}(\hat{s})|\,.
    \end{multline}
    We now bound the distance between two positions of the Brownian bridge \(\omega_1\) in a time interval of size \(\nu\) by Lemma \ref{Lemma:bounds_on_the_law_of_brownian_bridge} \ref{Lemma:bounds_on_the_law_of_brownian_bridge(ii)}, by integrating over the final position \(x_2\) of \(\omega_1\col [\tau_1, \tau_2+r_1] \to \mathbb{R}^d\), and obtain
    \begin{equation}\label{eq:lemma5.10_first bound}\begin{split}
        &\int\WU_{x_1, x_1}^{\tau_1+|\mathbf{r}|, \tau_1}(\mathrm{d} \omega ) |{\omega_1}(\tilde{s})-{\omega_1}(\hat{s})| \\&\quad = \int \d {x_2}\int\WU_{x_1, x_2}^{\tau_1+|\mathbf{r}|, \tau_2+r_1}(\mathrm{d} \omega_2 )\int \WU_{x_2, x_1}^{\tau_2+r_1, \tau_1}(\mathrm{d} \omega_1 ) |{\omega_1}(\tilde{s})-{\omega_1}(\hat{s})| \\
         &\quad \lesssim \int   \d {x_2}\, \int \PU_{x_1, x_2}^{t_2, 0}(\d {\omega_2}) \psi^{t_2}(x_1-x_2)\psi^{t_1}(x_1-x_2)\Bigl(|\hat s-\tilde{s}|^{1/2} + \abs{x_1-x_2}\frac{|\hat s-\tilde{s}|}{{t_1}} \Bigr)
         \end{split}
    \end{equation}
        where \(t_1\deq r_1+ \tau_2-\tau_1 \) and  \(t_2\deq\abs{\mathbf{r}}-t_1\) are the lengths of the time domains of \(\omega_1\) and \(\omega_2\) respectively and we recall the notation \eqref{eq:def:WU}. Using heat kernel estimates Lemma \ref{lem:boundonexp2} with \(g(x)=1+\abs{x}^\theta\), 
        \[\int \PU_{x_1, x_2}^{t_2, 0}(\d {\omega_2}) \leq \mathbbold{1}_{t_2 \leqslant|x_1|^{-2(\theta+1)}}+\mathrm{e}^{-c|x_1|^\theta t_2}+\mathrm{e}^{-c|x_1|^{1+\theta / 2}}+\mathrm{e}^{-c(\sqrt{t_2}|x_1|^{\theta+1})^{2 / 3}}\,, \]
        and  integrating in \(x_1\), we obtain the bound
        \begin{equation}\label{eq:t_2 bound of kernel}
            \int \d {x_1}\int \PU_{x_1, x_2}^{t_2, 0}(\d {\omega_2})  \lesssim_d 1+ t_2^{d/(2(\theta+1))} + t_2^{-d/\theta} + t_2^{-d/(2(\theta+1))}\,.
        \end{equation}
        We can now compute the two integrals involving the heat kernel appearing on the right-hand side of \eqref{eq:lemma5.10_first bound}
        \begin{align}
            \int \d {x_2} \,\psi^{t_2}(x_1-x_2)\psi^{t_1}(x_1-x_2)  &\lesssim  \frac{1}{(t_1 + t_2)^{d/2}} =  \frac{1}{\abs{\mathbf r}^{d/2}}\,, \label{eq:intheatkernels1} \\
        \int \d {x_2}\, \psi^{t_2}(x_1-x_2)\psi^{t_1}(x_1-x_2)\abs{x_1-x_2} &\lesssim\frac{\sqrt{t_1t_2}}{{(t_1+t_2)^{\frac{d+1}{2}}}}= \frac{\sqrt{t_1t_2}}{{\abs{\mathbf{r}}^{\frac{d+1}{2}}}} \,.\label{eq:intheatkernels2}
        \end{align}
Notice that by definition, for any \(\hat s,\tilde{s} \in I_i, i=1\dots N\) we have \(\abs{\hat s-\tilde s}\leq \nu \) and \(\nu N\lesssim t_1\). Combining this bound together with \eqref{eq:second_error_term5.46}-\eqref{eq:intheatkernels2}, we have 
        \begin{multline}\label{eq:boundwitht2t1}
            \int \d  x_1 \int \d {\tilde{x}_1}\int \WU_{x_1, x_1}^{\tau_1+|\mathbf{r}|, \tau_1}(\mathrm{d} \omega)\WU_{\tilde{x}_1, \tilde{x}_1}^{\tilde{\tau}_1+ |\tilde{\mathbf{r}}|, \tilde{\tau}_1}(\mathrm{d} \tilde{\omega})\big|\hat{\mathbb{V}}_\eta(\omega_1, \omega_1)\big| \\\lesssim  \frac{1}{\eps^{d+1}}\tr\Bigl({\ee^{-\abs{\tilde{\mathbf{r}}}(-\Delta/2 + U/2)}}\Bigr) \bigl(1+ t_2^{d/(2(\theta+1))} + t_2^{-d/\theta} + t_2^{-d/(2(\theta+1))}\bigr)\\\times\nu t_1 \biggl(\frac{\sqrt{\nu}}{\abs{\mathbf{r}}^{d/2}}+\frac{\nu\sqrt{t_1t_2}}{t_1\abs{\mathbf{r}}^{(d+1)/2}}\biggr)\,.
        \end{multline}
        Moreover, using that \(\nu\leq\abs{\mathbf{r}}\), \(t_1, t_2\leq \abs{\mathbf{r}}\), we can further obtain 
            \begin{equation}
                \eqref{eq:boundwitht2t1}\lesssim \frac{1}{\eps^{d+1}}\tr\bigl(\ee^{-|\tilde{\mathbf{r}}|(-\Delta/2 + U/2)}\bigr)\Bigl(1+ t_2^{d/(2(\theta+1))} + t_2^{-d/\theta} + t_2^{-d/(2(\theta+1))}\Bigr)\nu^{3/2} \biggl(\frac{1}{|\mathbf{r}|^{d/2-1}}\biggr)\,. 
            \end{equation}
         We now integrate over the variables \(\tau_1, \tau_2\), and, since from Remark \ref{rmk:thetha_s2} we have \(1-d/\theta>0\), we find
        \begin{align*}
            \int_{0\leq \tau_2,\tau_1\leq \nu }\d{\tau_1}\,\d{\tau_2}\Bigl(1+ t_2^{d/(2(\theta+1))} + t_2^{-d/\theta} + t_2^{-d/(2(\theta+1))}\Bigr) & \\\lesssim_{\theta} \nu^2+\nu^2|\mathbf r| +\int_0^\nu \d{\tau_2}\, (|\mathbf r|+\nu-\tau_2)^{1-d/\theta}
            & \lesssim \nu (1+|\mathbf r|)  \,.
        \end{align*}
        We thus obtain \eqref{eq:first_error_term5.46} by integrating over the variables \(\tau_3\) and \(\tilde{\boldsymbol{\tau}}\). 
\end{proof}
We conclude this subsection by  we showing  \eqref{eq:approximation_of_J} and Lemma \ref{lemma:riemannsumestimates}. 
\begin{proof}[Proof of \eqref{eq:approximation_of_J}]
The proof proceeds along the same lines as in \cite[lemma 5.15]{FKSS_24}, where we replace the estimates involving \(\hat{\mathbb{V}}_\eta\) with \eqref{eq:first_error_term5.46}. For times \(s,t\geq0\) satisfying \(\abs{s-t}\leq\nu\), we can prove as in Lemma \ref{lemma: bounds on V hat} that 
    \begin{multline}\label{eq:diff ws wt}
    \nu\int_{[0, \nu]^3} \mathrm{d} \boldsymbol{\tau} \int_{[0, \nu]^3} \mathrm{d} \tilde{\boldsymbol{\tau}} \int \d  x_1 \int \d {\tilde{x}_1}\int \WU_{x_1, x_1}^{\tau_1+|\mathbf{r}|, \tau_1}(\mathrm{d} \omega)\WU_{\tilde{x}_1, \tilde{x}_1}^{\tilde{\tau}_1+ |\tilde{\mathbf{r}}|, \tilde{\tau}_1}(\mathrm{d} \tilde{\omega})|\omega(s)-\omega(t)| \\ \lesssim  \nu^{13/2}\tr\bigl(\ee^{-|\tilde{\mathbf{r}}|(-\Delta/2 + U/2)}\bigr) \frac{(1+|\mathbf{r}|)^2}{|\mathbf r|^{d/2}}\,,\end{multline} 
    which allows us to conclude the proof. 
\end{proof}

\begin{proof}[Proof of Lemma \ref{lemma:riemannsumestimates}]
    We can show this result following the steps of the proof Lemma \ref{lemma:Convergence of K_eta} by noting that 
    \(\int_0^\infty \d t\, \ee^{-\abs{\mathbf r}t}=\frac{1}{\abs{\mathbf {r}}}\), and 
    \begin{equation}
        |I^\eps(\mathbf{r}, \mathbf{\tilde{r}})| \lesssim \frac{1}{\eps^{4d}}\tr\bigl(\ee^{-\abs{\mathbf{r}}(-\Delta/2 + U/2)}\bigr)\tr\bigl(\ee^{-\abs{\mathbf{\tilde r}}(-\Delta/2 + U/2)}\bigr)(1 +|\mathbf{r}|+ |\mathbf{ \tilde r}|)^4\,,
    \end{equation} 
    combined with the estimates \cite[proof of lemma 5.15, claims (i)-(iii)]{FKSS_24} and \eqref{eq:diff ws wt}. 
\end{proof}

\section{Regularity of the Green function} \label{app:Greenfunction}

In this section we state and prove estimates on the Green function $G$ from \eqref{eq:defG} of the Hamiltonian \eqref{def:one body hamiltonian} and its regularised version $G_N$ from \eqref{eq:defGN} with exponential cutoff function $\vartheta(x)\deq \ee^{-x}$, as well as estimates on the gradient of $G$. As an application, we prove bounds on the function $\tau^\eps$ from \eqref{def:taueps_Eeps}. Since the cutoff function is exponential, we have
\begin{equation}
	\label{eq:GN}
	G_N = \frac{1}{h} \ee^{- h/N} =  \int_{\frac1N}^\infty \d t \, \ee^{-th}\,, \qquad G_\infty = G\,.
\end{equation}

Throughout this section, \(g\col\mathbb{R}^d \to [1, \infty)\) is a positive, radial, non-decreasing function satisfying $g \leq U$, and for fixed $0 < \gamma < 1$ we define \(\tilde{g}\) as in \eqref{eq:gtilde} .

The first main result of this section gives bounds on $G$ and the difference $G_N - G$.

\begin{prop}[Estimates on $G_N$] 
	\label{prop:GN} Let $d=2$. For $0 < N \leq \infty$, let $G_N$ the truncated Green function \eqref{eq:GN}. Then for any fixed \(0<\gamma<1\) there exists a constant \( c\equiv c_\gamma>0\) such that
	\begin{equation}\label{eq:boundG}
		G_N(x,y) \lesssim \Bigl[\Bigl( \bigl[\log \Bigl(|x-y|\sqrt{\tilde g(x)\vee \tilde g(y)}\Bigr)^{-1}\bigr]_+ \wedge \log N \Bigr)\vee 1\Bigr] \ee^{-c|x-y|\sqrt{\tilde g(x)\vee \tilde g(y)}}
	\end{equation}
	and
	\begin{equation}
		\label{eq:GN-G}
		|G_N(x,y) - G(x,y) | \lesssim   \Bigl[\log(N|x-y|^2)^{-1}\Bigr]_+  \wedge \frac{\ee^{-N|x-y|^2}}{N|x-y|^2}\,.
	\end{equation} 
\end{prop}

An upper bound similar to \eqref{eq:boundG} was previously obtained in \cite{CW-heatkernel}. The estimate  \eqref{eq:boundG} is a generalisation of this result to potentials with faster radial growth and to finite \(N\). 

\begin{remark}
By the same proof, the result of Proposition \ref{prop:GN} holds in any dimension $d\geq 3$ with bound
\begin{equation}\label{eq:Gdim3}
    { G(x,y) \lesssim |x-y|^{2-d}\ee^{-c|x-y|\sqrt{\tilde g(x) \vee \tilde g(y)}}\,.}
\end{equation}
\end{remark}

The second main result of this section shows that under Assumption \ref{assumption:P and D for U}, $G$ is continuously differentiable with a bound on its gradient.

\begin{prop}[Estimate on $\nabla G$] \label{lm:gradG}
Fix $d \geq 2$ and $0 < \gamma < 1$. Suppose that \(U\) satisfies Assumption \ref{assumption:P and D for U}. Then there exists a constant \( c\equiv c_\gamma>0\) such that $G$ is differentiable for all $x \neq y$ and
	\begin{equation}\label{eq:gradG}
    |\nabla_xG(x,y)| \lesssim \frac{1}{\abs{x-y}^{d-1}}\ee^{-c\abs{x-y}\sqrt{ \tilde g( y)}}\,.
	\end{equation}
\end{prop}
\begin{remark} \label{rmk:gradG_continuous}
    For any positive locally bounded potential \(U\), we have from Proposition \ref{prop:GN} that \(U\cdot G \) is \(L^p_{\text{loc}}(\Omega)\), for any \(p\in [1,\infty]\) and \(\Omega \subset \subset \mathbb{R}^d\setminus \{(x,\cdots,x)\col x\in \mathbb{R}\}\). Using Proposition \ref{lm:gradG}, Calderón-Zygmund estimates and Sobolev embedding with \(p>d\), we find that \[x \mapsto \nabla_x G(x,y) \] is continuous for all \(x\neq y\).  
\end{remark}

The rest of this section is devoted to the proofs of Proposition \ref{prop:GN} and Proposition \ref{lm:gradG}.

\subsection{Proof of Proposition \ref{prop:GN}}
To prove Proposition \ref{prop:GN} we first need the following results from \cite[Section 7]{FKSS_2020}. Notice that these are modified versions of Lemma 7.3 and Lemma 7.4 in that paper. In particular, Lemma \ref{lem:boundonexp2} is a more general version of \cite[Lemma 7.3]{FKSS_2020}, and the proof follows using a similar stopping time argument which we omit here.

\begin{lemma}\label{lem:boundonexp2}
Let $d \geq 1$. For any fixed $\gamma \in (0,1) $ there exist constants $C, c >0$ such that
	\begin{equation} \label{eq:boundU}
		\begin{split}
			&\int \mathbb{W}_{x,\tilde x}^{t,0}(\d{\omega})\ee^{-\int_0^t \d{s}\, U(\omega(s))} \,\\
            &\quad \leq  C \psi^t(x-\tilde x)\ee^{-c \frac t2} \Bigl(\mathbbm{1}_{t\leq (g(\gamma x)\vee g(\gamma \tilde x))^{-2}(|x| +|\tilde x|)^{-2}} +\ee^{-c(g(\gamma x)\vee g(\gamma \tilde x))t} \\& \hspace{3cm}+ \ee^{-c(|x|+|\tilde x|)\sqrt{g(\gamma x)\vee g(\gamma \tilde x)}} +\ee^{-c t^{1/3}(|x|+|\tilde x|)^{2/3}(g(\gamma x)\vee g(\gamma\tilde x))^{2/3}}\Bigr)
	\end{split} \end{equation}
for all $x, \tilde x \in \R^d$ and $t > 0$.
\end{lemma}

\begin{proof}[Proof of Proposition \ref{prop:GN}] We start with the proof of \eqref{eq:boundG}. Without loss of generality we can assume $|x|\geq |y|$. We also assume that $|x|\geq 1$, for otherwise the statement is trivial. Hence, it is sufficient to prove that
	\[
	G_N(x,y) \lesssim \Bigl[\Bigl( \Bigl[\log \Bigl(|x-y|\sqrt{\tilde g(x)}\Bigr)^{-1}\Bigr]_+ \wedge \log N\Bigr)\vee 1\Bigr] \ee^{-c|x-y|\sqrt{\tilde g(x)}}\,.
	\]
The Feynman-Kac formula from Lemma \ref{Feynman-Kac formula} yields
	\[\begin{split}
		G_N(x,y) & =  \int_{\frac1N}^\infty \d{t} \, \ee^{-th}(x,y) = \int_{\frac 1N}^\infty \d{t}\, \ee^{-\kappa t} \int \mathbb W_{x,y}^{t,0} (\d{\omega})\, \ee^{-\int_0^{t}\d{s}\, U(\omega(s))},
	\end{split}\]
	from which we deduce using \eqref{eq:boundU} that
	\begin{align} 
				G_N(x,y)
&\lesssim 
			\int_{\frac1N}^\infty \d{t} \, \frac{1}{t}  \ee^{-\frac{|x-y|^2}{2t}}\ee^{-c \frac t2}\Bigl(\mathbbm{1}_{t \leq \tilde g(x)^{-2}|x|^{-2}} + \ee^{-c\tilde g(x) t}+ \ee^{-c|x|\sqrt{\tilde g(x)}} +\ee^{-ct^{1/3}|x|^{2/3}\tilde g(x)^{2/3}}\Bigr) \notag \\ \label{eq:G1}
			 &\eqd \text{(a)} + \text{(b)} + \text{(c)} +\text{(d)}\,.
	\end{align}
	We analyse the four terms separately. First we consider the term (a): 
	\[\begin{split}
		\text{(a)} \lesssim \int_{\frac1N}^{\tilde g(x)^{-2}|x|^{-2}} \d{t}\, \frac{1}{t}  \ee^{-\frac{|x-y|^2}{2t}} & =\,  \int_{\frac{|x-y|^2}2\tilde g(x)^2|x|^2}^{\frac{N|x-y|^2}{2}}\, \d{s}\, \frac{1}{s} \ee^{-s}\,.
	\end{split}\]
	We consider different cases: if $|x-y|^2\tilde g(x) \leq 2$ and $|x-y|^2\leq \frac2N$ then
	\begin{equation}\label{eq:a1}
		\text{(a)} \lesssim \log(\frac{ \tilde g(x)^2|x|^2}N)^{-1} \ee^{-c |x-y|^2\tilde g(x)}\lesssim \log N \ee^{-c |x-y|^2\tilde g(x)}\,.
	\end{equation}
	On the other hand, if $|x-y|^2\tilde g(x) \leq 2$ and $|x-y|^2> \frac2N$ then
	\begin{equation}\label{eq:a2}\begin{split}
			\text{(a)} &\lesssim \int_{\frac{|x-y|^2}2\tilde g(x)^2|x|^2}^1\, \d{s}\, \frac{1}{s} \ee^{-s} +\int_{1}^{\frac{N|x-y|^2}{2}}\, \d{s} \,\frac{1}{s} \ee^{-s}\\
			&\lesssim 
			\Bigl(\Bigl[-\log(|x-y| \tilde g(x)|x|)\Bigr]_+ +1\Bigr) \ee^{-c |x-y|^2\tilde g(x)^2|x|^2} \\
			&\lesssim \Bigl(\Bigl[-\log(|x-y|\sqrt{\tilde g(x)})\Bigr]_+ +1\Bigr)\ee^{-c |x-y|\sqrt{\tilde g(x)}}\,.
	\end{split}\end{equation}
	On the other hand, if $|x-y|^2\tilde g(x) >2 $, we have $\tilde g(x)|x|^2 \geq \tilde g(x)|x-y|^2 >2$, which also implies $|x-y|^2 \tilde g(x)^2 |x|^2 >2$. Thus, 
	\begin{equation}\label{eq:a3}
		\text{(a)} \lesssim \int_{\frac{|x-y|^2}2\tilde g(x)^2|x|^2}^{\frac{N|x-y|^2}{2}}\d s\, \ee^{-s} \lesssim \ee^{-c|x-y|^2\tilde g(x)^2|x|^2} \lesssim \ee^{-c|x-y|^2\tilde g(x)} \lesssim \ee^{-c|x-y|\sqrt{\tilde g(x)}}\,.
	\end{equation}
	Together \eqref{eq:a1}, \eqref{eq:a2}, and \eqref{eq:a3} are all bounded by the right-hand side of \eqref{eq:boundG}.
	
Next, we consider the term (b) in \eqref{eq:G1}. We perform a change of variable $s = \tilde g(x)t$ and use that $\tilde g(x) +\frac 12 \geq \tilde g(x)$, which yields 
	\begin{equation}\label{eq:b0}
		\begin{split}
			\text{(b)} &= \int_{\frac1N}^\infty \d{t}\, \frac{1}{t}  \ee^{-\frac{|x-y|^2}{2t}} \ee^{-c\tilde g(x) t - c\frac t2} 
			\\& \lesssim \int_{\frac{ \tilde g(x)}N}^\infty \d{s}\, \frac1s \ee^{-\frac{|x-y|^2\tilde g(x)}{2s}}\ee^{-s}\\
			& =  \int_{\frac{ \tilde g(x)}N}^1 \d{s}\, \frac1s \ee^{-\frac{c|x-y|^2\tilde g(x)}{2s}}\ee^{-s} +\int_1^\infty  \d{s}\, \frac1s \ee^{-\frac{c|x-y|^2\tilde g(x)}{2s}}\ee^{-s}\,,
		\end{split}
	\end{equation}
	where, without loss of generality,  we assumed $\frac{ \tilde g(x)}N \leq 1$. Notice that, whenever  $\frac{ \tilde g(x)}N > 1$, we can bound $(b)$ with the second integral on the right-hand side on the last line of \eqref{eq:b0}.
	We perform another change of variables in the first integral on the right-hand side of \eqref{eq:b0} to get
	\begin{equation}\label{eq:intb}\begin{split}
			\int_{\frac{\tilde g(x)}N}^1 \d{s}\, \frac1s \ee^{-\frac{c|x-y|^2\tilde g(x)}{2s}}\ee^{-s} & \lesssim \int_{\frac{\tilde g(x)}N}^1 \d{s}\, \frac1s \ee^{-\frac{c|x-y|^2\tilde g(x)}{2s}} \\
			& = \int_{c|x-y|^2\tilde g(x)}^{cN|x-y|^2} \d{t}\, \frac1t \ee^{-t}  = \int_{c|x-y|^2\tilde g(x)}^1 \d{t}\, \frac1t \ee^{-t}  + \int_1^{cN|x-y|^2} \d{t}\, \frac1t \ee^{-t}
	\end{split}\end{equation}
	for $cN|x-y|^2> 1 $. Hence,
	\[\begin{split}
		\int_{c|x-y|^2\tilde g(x)}^1 \d{t}\, \frac1t \ee^{-t}  + \int_1^{cN|x-y|^2} \d{t}\, \frac1t \ee^{-t}  &\lesssim \bigl([-\log(|x-y|^2\tilde g(x))]_+ +1 \bigr)\ee^{-c |x-y|^2\tilde g(x)}\,.
	\end{split}\]
	If $ cN|x-y|^2\leq 1$ then the right-hand side of \eqref{eq:intb} can be bounded as 
	\[
	\int_{c|x-y|^2\tilde g(x)}^{cN|x-y|^2} \d{t}\, \frac1t \ee^{-t} \lesssim \log \pbb{\frac{\tilde g(x)}N}^{-1} \ee^{-c |x-y|^2\tilde g(x)}\lesssim \log N \ee^{-c |x-y|^2\tilde g(x)}\,.
	\]
	The second integral in \eqref{eq:b0} can be estimated using an elementary argument by
	\[
	\int_1^\infty  \d{s}\, \frac1s \ee^{-\frac{c|x-y|^2\tilde g(x)}{2s}}\ee^{-s} \lesssim \ee^{-c |x-y|\sqrt{\tilde g(x)}}\,.
	\]

Next, we consider the term (c) in \eqref{eq:G1}. We argue analogously to \eqref{eq:b0}, \eqref{eq:intb}. Using $|x-y| \leq 2 |x|$ we get
	\[\begin{split}
	&\int_{\frac 1N}^\infty \d{t}\, \frac{1}{t}  \ee^{-\frac{|x-y|^2}{2t}}\ee^{-c|x|\sqrt{\tilde g(x)}} \ee^{-c\frac t2}\\ &\quad\lesssim \Bigl(\Bigl( \bigl[\log|x-y|^{-1}\bigr]_+ \wedge \log N \Bigr) \vee 1 \Bigr) \ee^{-c |x|\sqrt{\tilde g(x)}}\\ 
    &\quad = \Bigl(\Bigl( \Bigl[\log(|x-y|\sqrt{\tilde g(x)})^{-1} + \log(\sqrt{\tilde g(x)})\Bigr]_+ \wedge \log N \Bigr) \vee 1 \Bigr) \ee^{-c |x|\sqrt{\tilde g(x)}}\\
    &\quad \lesssim \Bigl(\Bigl( \Bigl[\log(|x-y|\sqrt{\tilde g(x)})^{-1} \Bigr]_+ \wedge \log N \Bigr) \vee 1 \Bigr) \ee^{-c |x|\sqrt{\tilde g(x)}}\,,
    \end{split}
	\]
	where in the last step we used that $\log(\sqrt{\tilde g(x)})\ee^{-c/2\sqrt{\tilde g(x)}}$ is of order one.

Finally, we consider the term (d) in \eqref{eq:G1}. This time we perform the change of variables $s = t (\tilde g(x)|x|)^2$ to get 
	\[\begin{split}
		\int_{\frac 1N}^\infty \d{t}\, \frac{1}{t}  \ee^{-\frac{|x-y|^2}{2t}}\ee^{-ct^{1/3}(|x|\tilde g(x))^{2/3}} & = \int_{\frac{ (\tilde g(x)|x|)^2}N}^\infty \d s \, \frac{1}{s} \ee^{-\frac{c|x-y|^2(\tilde g(x)|x|)^2}{s}} \ee^{-s^{1/3}}\\
		& \lesssim \int_{\frac{ (\tilde g(x)|x|)^2}N}^1 \d s\, \frac1s \ee^{-\frac{c\lambda}{s}} + \int_1^\infty \d s\, \frac1s \ee^{-\frac{c\lambda}{s}}\ee^{-s^{1/3}}\,,\\
		& = (\r d_1) + (\r d_2)\,,
	\end{split}\]
	where we set $\lambda  = |x-y|^2(\tilde g(x)|x|)^2 \geq |x-y|^4\tilde g(x)^2$ and assuming (without loss of generality, similarly as in \eqref{eq:b0}) that $ \frac{\tilde g(x)}N \leq 1$. For the term $(\r d_1)$, proceeding along the same lines as before, if $|x-y|^2 > \frac{2}N $ then
	\begin{equation}\label{eq:d1}\begin{split}
			(\r d_1) & = \int_{\lambda}^{\frac{N|x-y|^2}{2}} \d u\, \frac1u \ee^{-u}=  \int_{\lambda}^{1} \d u\, \frac1u \ee^{-u} + \int_{1}^{\frac{N|x-y|^2}{2}} \d u\, \frac1u \ee^{-u} \lesssim \bigl([-\log(\lambda)]_+ +1 \bigr)\ee^{-\lambda}\,.
	\end{split}\end{equation}
	On the other hand, if $|x-y|^2 \leq \frac{2}N$ then
	\begin{equation}\label{eq:d12}(\r d_1) = \int_{\lambda}^{\frac{N|x-y|^2}{2}} \d u\, \frac1u \ee^{-u} \lesssim (\log N)\ee^{-\lambda}\,,
	\end{equation}
	which satisfies the bound \eqref{eq:boundG}. For the term $(\r d_2)$ instead we have 
	\begin{equation}\label{eq:d2}
		\begin{split}
			(\r d_2) \lesssim \int_1^{\lambda^{3/4}} \d s\,  \ee^{-\frac{c\lambda}{s}}\mathbbm{1}_{\lambda\geq 1} + \int_{\lambda^{3/4}\vee 1}^\infty \d s\,  \ee^{-s^{1/3}}\lesssim \ee^{c\lambda^{-1/4}}\,
		\end{split}
	\end{equation}
	since $\lambda \geq |x-y|^4\tilde g(x)^2$. Putting \eqref{eq:d1}, \eqref{eq:d12}, and \eqref{eq:d2} together, we conclude the proof of \eqref{eq:boundG}.

We now prove  \eqref{eq:GN-G}. As above,
\[\begin{split}
	G(x,y) -G_N(x,y) & = \int_0^\infty \d t\,  \ee^{-t h}(x,y) - \int_{\frac1N}^\infty \d t\,\ee^{-t h}(x,y) \\
	& = \int_0^{\frac1N} \d t\,  \ee^{-t \k}\int \mathbb{W}_{x,y}^{t,0}(\d \omega) \, \ee^{-\int_0^t\d s\, U(\omega(s))}\,.
\end{split}\]
Using Lemma \ref{lem:boundonexp2} and assuming $|x| \geq |\tilde x|$, and $|x| \geq 1$,  we obtain
\[
\begin{split}
	\big|G(x,y) -G_N(x,y)\big| &\leq  \int_0^{\frac 1N} \d t\, \frac{1}{t}  \ee^{-\frac{|x-y|^2}{2t}}\ee^{-\frac{c}{2}t}\Bigl(\mathbbm{1}_{t \leq \tilde g(x)^{-2}|x|^{-2}} + \ee^{-c\tilde g(x) t}\\
	&\hspace{5cm}+ \ee^{-c|x|\sqrt{\tilde g(x)}} +\ee^{-ct^{1/3}|x|^{2/3}\tilde g(x)^{2/3}}\Bigr)\,\\
	& \eqd (\r a) +(\r b) +(\r c) +(\r d)\,.
\end{split}
\]
We begin with (a), and consider the cases $\tilde g(x)^2|x|^2 > N$ and $\tilde g(x)^2|x|^2 \leq N$ separately. For  $\tilde g(x)^2|x|^2> N$ we have  a similar bound as in  \eqref{eq:a2},
\[
(\r a) \lesssim \int_{\frac{|x-y|^2}2\tilde g(x)^2|x|^2}^\infty\, \d s \,\frac{1}{s} \ee^{-s}\,.
\]
If $|x-y|^2\tilde g(x) \leq 2$ then we split the integral as 
\[\begin{split}
(a) \lesssim &\int_{\frac{|x-y|^2}2\tilde g(x)^2|x|^2}^1\, \d s\, \frac{1}{s} \ee^{-s}  + \int_1^\infty\, \d s\, \frac{1}{s} \ee^{-s} \\\lesssim & \Bigl[\log(\tilde g(x)^2|x|^2|x-y|^2)^{-1}\Bigr]_+ + 1 \lesssim\Bigl[\log(N|x-y|^2)^{-1}\Bigr]_+ + 1 \,,
\end{split}\]
 using that $\tilde g(x)^2|x|^2 >N$. On the other hand, if $|x-y|^2\tilde g(x) > 2$, then $|x-y|^2\tilde g(x) > N |x-y|^2$. Thus we have
 \[
 \int_{\frac{|x-y|^2}2\tilde g(x)^2|x|^2}^\infty\, \d s \,\frac{1}{s} \ee^{-s} \lesssim \frac{\ee^{-c|x-y|^2\tilde g(x)^2|x|^2}}{|x-y|^2\tilde g(x)^2|x|^2} \lesssim \frac{\ee^{-cN|x-y|^2}}{N|x-y|^2}\,.
 \]
Let us now consider when $ \tilde g(x)^2|x|^2\leq N$, we perform a change of variable to get 
\begin{equation}\begin{split} \label{eq:terma}
		(a) & \lesssim  \int_0^{\frac1N} \d t\, \frac{1}{t}  \ee^{-\frac{|x-y|^2}{2t}} = \int_{\frac{N|x-y|^2}{2}}^\infty \d s\, \frac{1}{s} \ee^{-s} \\& \lesssim \mathbbm{1}_{N|x-y|^2<2}\Bigl(\int_{\frac{N|x-y|^2}{2}}^1 \d s\, \frac{1}{s} \ee^{-s} + \int_1^\infty \d s\, \frac{1}{s} \ee^{-s}\Bigr)  + \mathbbm{1}_{N|x-y|^2>2}\int_{\frac{N|x-y|^2}{2}}^\infty \d s\, \frac{1}{s} \ee^{-s} \\ & \lesssim \Bigl( 1+ \Bigl[\log(N|x-y|^2)^{-1}\Bigr]_+\Bigr) \wedge \, \frac{\ee^{-cN|x-y|^2}}{N|x-y|^2}\,.
\end{split}\end{equation} 

Next, we consider the term (b). We proceed similarly as in \eqref{eq:intb} to obtain
\[
\int_0^{\frac1N} \d{t}\, \frac{1}{t} \ee^{-\frac{|x-y|^2}{2t}}\ee^{-\frac{c}{2}t -c\tilde g(x)t} \lesssim \int_0^{\frac{\tilde g(x)}N}\d s\, \frac{1}{s}\ee^{-\frac{|x-y|^2\tilde g(x)}{2s}}\ee^{-s} \, .
\]
If $\tilde g(x)\geq N$ we split the term as   
\[
\begin{split}
      \int_0^{1}\d s\, \frac{1}{s}\ee^{-\frac{|x-y|^2\tilde g(x)}{2s}}  +  \int_1^{\frac{\tilde g(x)}N}\d s\, \frac{1}{s}\ee^{-\frac{|x-y|^2\tilde g(x)}{2s}}\ee^{-s}   &\lesssim \int_{\frac{|x-y|^2\tilde g(x)}{2}}^\infty \d u\, \frac{\ee^{-u}}{u} + \ee^{-|x-y|^2N}\\
      & \lesssim  1+ \Bigl[\log(N|x-y|^2)^{-1}\Bigr]_+ + \ee^{-|x-y|^2N}\,.
\end{split}
\]
Otherwise, if $\tilde g(x)<N$ ,
\[
\begin{split}
    \int_0^{\frac{\tilde g(x)}N}\d s\, \frac{1}{s}\ee^{-\frac{|x-y|^2\tilde g(x)}{2s}}\ee^{-s} &\lesssim \int_{\frac{|x-y|^2N}{2}}^\infty \d u\, \frac{1}{u}\ee^{-u} \\
    & \lesssim 1+ \Bigl[\log(N|x-y|^2)^{-1}\Bigr]_+ \wedge \frac{\ee^{-|x-y|^2N}}{N|x-y|^2}\,,
\end{split}
\]
where the last inequality is obtained by considering the cases $|x-y|^2 <N^{-1}$ and $|x-y|^2> N^{-1}$ separately.

The terms (c) and (d) are treated similarly. This concludes the proof.
\end{proof}

To conclude this subsection, we estimate the function $\tau^\eps$ from \eqref{def:taueps_Eeps}. 
\begin{lemma}\label{lemma:bounds on taueps}
    Let \(d  = 2,3\) and recall \eqref{eq:def:chieps}. Consider \(U\geq0\) and \(\tau^\eps \equiv \tau^{\eps,U}\) defined as in \eqref{def:taueps_Eeps}. Then for any $\epsilon > 0$ we have \[\|\tau^\eps \|_{L^\infty} \lesssim_{\kappa,v} \chi(\varepsilon)\,.\] 
\end{lemma}
\begin{proof}
Using Proposition \ref{prop:GN}, we get 
    \begin{align*}
        \abs{\tau^\eps(x)}& \leq  \int \d y\, \abs{v^\eps(x-y)}G(x,y)\\&\lesssim_{v}\,\frac{1}{\eps^d} \int_{\abs{z}<\eps}\d z \,\Bigl(\log\abs{z}^{-1} \wedge1 \cdot \mathbbold{1}_{d=2} +\abs{z}^{-1}\mathbbold{1}_{d=3} \Bigr) \\
        &\lesssim\, \chi(\eps)\,. \qedhere
    \end{align*}
\end{proof}

\subsection{Proof of Proposition \ref{lm:gradG}}

The key estimate behind the proof of Proposition \ref{lm:gradG} is the following bound on the gradient of the expectation of the Brownian bridge.

\begin{lemma} \label{lm:gradE} Let \(U\) satisfy Assumption \ref{assumption:P and D for U} and \(d\geq 2\). Then,  for almost every \(x,y \in \mathbb{R}^d\),
\[\Big|\nabla_{x} \mathbb{E}_{y,x}^{t,0}\Bigl[\ee^{- \int_0^t \d{s}\, U(X_s)}\Bigr]\Big| \lesssim \frac{1}{\sqrt{t}} \,,\]
where \(\mathbb{E}_{y,x}^{t,0}\) denotes the expectation with respect to the law of the Brownian bridge measure $\P_{y,x}^{t,0}$ from Section \ref{sec:brownian}, between \(x\) and \(y\) over time \(t\) .
\end{lemma}
\begin{proof}
Suppose first that $U$ is replaced with a smooth potential \(V\col \mathbb{R}^d \to [0, \infty)\). The Brownian bridge $(X_s)_{s\in [0,t]}$ going from $x$ to $y$ in time $t$ can be written in terms of standard Brownian motion $(B_s)_{s\geq 0}$ as
\begin{equation}
\label{eq:bbcoupling}
X_s=B_s - \frac{s}{t}B_t + x + \frac{s}{t}(y-x)\,.
\end{equation}
\noindent
We denote the expectation with respect to Brownian motion $(B_s)_{s\geq 0}$ by $\bE[\,\cdot \,]$. Thus,
\begin{equation*}
\begin{split}
&\nabla_x \mathbb{E}_{y,x}^{t,0} \Bigl[\ee^{-\int_0^t\d s \, V(X_s)}\Bigr] \\
&\; = \nabla_x\bE\Bigl[\ee^{-\int_0^t \d s\, V(B_s-\frac{s}{t}B_t + x +\frac{s}{t}(y-x))}\Bigr]\\
&\; = - \bE\biggl[\ee^{-\int_0^t \d r\, V(B_r-\frac{r}{t}B_t + x +\frac{r}{t}(y-x))}\int_0^t \d s\, \biggl(1-\frac{s}{t}\biggr)\nabla V\pbb{B_s - \frac{s}{t}B_t + x +\frac{s}{t}(y-x)}\biggr]\\
&\; =-\int_0^t \d s\, \biggl(1-\frac{s}{t}\biggr) \bE\biggl[\ee^{-\int_0^t \d r\, V(B_r-\frac{r}{t}B_t + x +\frac{r}{t}(y-x))}\nabla V\biggl(B_s -\frac{s}{t} B_t + x +\frac{s}{t}(y-x)\biggr)\biggr]\,.
\end{split}
\end{equation*}
\noindent
Therefore, 
\begin{equation}\label{eq:gradEandgradV}
\nabla_x \mathbb{E}_{y,x}^{t,0} \Bigl[\ee^{-\int_0^t \d s\, V(X_s)}\Bigr]=-\int_0^t \d s\, \biggl(1-\frac{s}{t}\biggr)\mathbb{E}_{y,x}^{t,0}\Bigl[\ee^{-\int_0^t\d r\, V(X_r)}\nabla V(X_s)\Bigr]\,.
\end{equation}
We focus on the expectation inside the integral, conditioning on the value of $X_s$,
\begin{equation}\label{gt}
\begin{split}
&\mathbb{E}_{y,x}^{t,0}\Bigl[\ee^{-\int_0^t \d r\, V(X_r)}\nabla V(X_s)\Bigr]\\
&\quad =\int \d w\, \mathbb{E}_{y,x}^{t,0}\Bigl[\ee^{-\int_0^t \d r \, V(X_r)}\delta(X_s-w)\Bigr]\nabla V(w)\\
&\quad =\int \d w\, \mathbb{E}_{y,x}^{t,0}\Bigl[\ee^{-\int_0^s \d r \, V(X_r)}\delta(X_s-w)\ee^{-\int_s^t \d r \, V(X_r)}\Bigr]\nabla V(w)\\
&\quad =\int \d w\,\frac{\psi^s(x-w)\psi^{t-s}(w-y)}{\psi^t(x-y)} \bE_{w,x}^{s,0}\Bigl[\ee^{-\int_0^s \d r\, V(X_r)}\Bigr]\bE_{y,w}^{t,s}\Bigl[\ee^{-\int_s^t\d r\, V(X_r)}\Bigr]\nabla V(w)\,,
\end{split}
\end{equation}
where in the last equality we used the Markov property of Brownian motion to rewrite the probability measure of the Brownian bridge.

Let us first consider \(U\) satisfying \assD{} in Assumption \ref{assumption:P and D for U}. Let \(g\) be the lower bound on \(U\) and \(0 < \gamma<1\) given by Assumption \ref{assumption:P and D for U}. Then \eqref{gt} holds for \(V=U\). Using the heat kernel bounds Lemma \ref{lem:boundonexp2}, we find 
\begin{align}
   &\absbb{ \Bigl(1-\frac{s}{t}\Bigr)\bE_{y,w}^{t,s}\Bigl[\ee^{-\int_s^t\d r\, U(X_r)}\Bigr]\nabla U(w)}\nonumber \\ &\quad \lesssim\frac{t-s}{t}\ee^{-ct} \Bigl(\mathbbm{1}_{t-s\leq (g(\gamma y)\vee g(\gamma w))^{-2}(|y| +|w|)^{-2}}  +\ee^{-c(g(\gamma y)\vee g(\gamma w))(t-s)} + \ee^{-c(|y|+|w|)\sqrt{g(\gamma y)\vee g(\gamma w)}}\nonumber \\& \qquad \qquad  \qquad +\ee^{-c (t-s)^{1/3}(|y|+|w|)^{2/3}(g(\gamma y)\vee g(\gamma w))^{2/3}}\Bigr) \abs{\nabla U(w)} \label{eq:estimate_grdE}\\ &
   \quad \lesssim\frac{\ee^{-ct}}{t} \abs{\nabla U(w)} \Biggl( \frac{1}{g(\gamma w)^2 \abs{w}^2} + \frac{(t-s)}{(t-s)^{3/2}g(\gamma w)^{3/2}} + \frac{(t-s)}{g(\gamma w)^2 \abs{w}} + \frac{(t-s)}{(t-s)g(\gamma w)^2 \abs{w}^2} \Biggr)\nonumber\\
   & \quad  \lesssim \frac{1}{t\sqrt{t-s}}\nonumber
\end{align}
for some positive constant \(c\equiv c_U>0\) and for all \(w \in \mathbb{R}^d\). 
Thus,
\begin{equation*}
\absB{\nabla_x \mathbb{E}_{y,x}^{t,0} \Bigr[\ee^{-\int_0^t \d s\,U(X_s)}\Bigr]}\lesssim\int_0^t \d{s}\,\frac{1}{t\sqrt{t-s}} \lesssim \frac{1}{\sqrt{t}}\,,
\end{equation*}
as desired.

Next, suppose now that \(U\) satisfies \assP{} in Assumption \ref{assumption:P and D for U} for some \(\theta\geq0\). Consider \(V\equiv V_\delta=U *\varphi_\delta \) where \(\{\varphi_\delta\}_{\delta>0}\) is a family of mollifiers such that  \(\varphi_\delta\) has support inside the ball of radius \(\delta>0\). Notice that \(V_\delta\) is smooth, and for any \(\delta<1\), \(V_{\delta}\) satisfies \assP{} for the same power \(\theta\). Moreover, 
\[\abs{\nabla V_\delta(x) }\leq\int \d{y}\, U(y) \abs{\nabla \varphi_{\delta}(x-y)} \lesssim \frac{1}{\delta^{d+1}}(1+\abs{x}^\theta). \]
Thus, using Lemma \ref{lem:boundonexp2} with \(g\col x \mapsto1+\abs{x}^\theta\), \eqref{eq:gradEandgradV}, \eqref{eq:estimate_grdE} and the fact that \(\abs{\nabla V_\delta(x) }/g(\gamma x) \lesssim_\gamma \delta^{-d-1}\) for any \(\gamma>0\), 
\begin{equation}\label{eq:firstGrownwall}
    \absB{\nabla_x \mathbb{E}_{y,x}^{t,0} \Bigl[\ee^{-\int_0^t\d{s}\, V_{\delta}(X_s)}\Bigr]}\lesssim_U \frac{1}{\delta^{d+1}} \int_0^t \d{s} \,\Bigl(1-\frac{s}{t}\Bigr)\frac{1}{(t-s)} \lesssim \frac{1}{\delta^{d+1}} 
\end{equation}
is bounded uniformly in $x,y \in \R^d$ and $t > 0$. We integrate by parts in \eqref{gt}, which yields
\begin{multline*}
    \mathbb{E}_{y,x}^{t,0}\Bigl[\ee^{-\int_0^t\d{r}\, V_\delta(X_r)}\nabla V_\delta(X_s)\Bigr]\\=
-\int \d w\, V_\delta(w)\biggl(\Psi(w)\, \bE_{w,x}^{s,0}\Bigl[\ee^{-\int_0^s \d r \, V_\delta(X_r)}\Bigr]\bE_{y,w}^{t,s}\Bigl[\ee^{-\int_s^t \d r\, V_\delta(X_r)}\Bigr]\qquad\\
+\frac{\psi^s(x-w)\psi^{t-s}(w-y)}{\psi^t(x-y)} \,\bE_{y,w}^{t,s}\Bigl[\ee^{-\int_s^t\d r\, V_\delta(X_r)}\Bigr]\nabla_x\bE_{w,x}^{s,0}\Bigl[\ee^{-\int_0^s\d r\, V_\delta(X_r)}\Bigr]\biggr)\,,
\end{multline*}
where we defined
\begin{equation*}
    \Psi(w) \deq \nabla_x \pbb{\frac{\psi^s(x-w)\psi^{t-s}(w-y)}{\psi^t(x-y)}}
= \biggl(\frac{x-y}{t}-\frac{x-w}{s}\biggr)\frac{\psi^s(x-w)\psi^{t-s}(w-y)}{\psi^t(x-y)}\,.
\end{equation*}
Introducing the notations
\begin{align*}
    f_\delta(t,x,y)&\deq \nabla_x \mathbb{E}_{y,x}^{t,0} \bigl[\ee^{-\int_0^t \d s\, V_\delta(X_s)}\bigr]\,,\\
f_\delta(t)&\deq\sup\limits_{x,y \in \mathbb{R}^d}  | f_\delta(t,x,y)|\,,
\end{align*}
we have
\begin{multline}\label{eq:g2}
f_\delta(t,x,y) =
\int_0^t \d s\,\Bigl(1-\frac{s}{t}\Bigr)\int \d w \,V_\delta(w)\Psi(w) \, \bE_{w,x}^{s,0}\Bigl[\ee^{-\int_0^s \d r\, V_\delta(X_r)}\Bigr]\bE_{y,w}^{t,s}\Bigl[\ee^{-\int_s^t \d r\, V_\delta(X_r)}\Bigr]\\   
     +\int_0^t \d s\,\Bigl(1-\frac{s}{t}\Bigr)\int \d w \,V_\delta(w)\frac{\psi^s(x-w)\psi^{t-s}(w-y)}{\psi^t(x-y)} \bE_{y,w}^{t,s}\Bigl[\ee^{-\int_s^t\d r\, V_\delta(X_r)}\Bigr]f_{\delta}(s,x,w)\,.
\end{multline}
Define 
\begin{align*}
\alpha(s)&\deq \Bigl(1-\frac{s}{t}\Bigr)\sup_{x,y \in \mathbb{R}^d} \bigg|\int \d w \,V_\delta(w)\Psi(w) \, \bE_{w,x}^{s,0}\Bigl[\ee^{-\int_0^s \d r \,V_\delta(X_r)}\Bigr]\bE_{y,w}^{t,s}\Bigl[\ee^{-\int_s^t \d r\, V_\delta(X_r)}\Bigr]\bigg|\,;\\
\beta (s)& \deq \Bigl(1-\frac{s}{t}\Bigr)\sup\limits_{x,y\in \mathbb{R}^d}\bigg| \int \d w\, V_\delta(w)\frac{\psi^s(x-w)\psi^{t-s}(w-y)}{\psi^t(x-y)} \bE_{y,w}^{t,s}\Bigl[\ee^{-\int_s^t \d r\, V_\delta(X_r)}\Bigr]\bigg|\,.
\end{align*}
To estimate $\alpha(s)$ and $\beta(s)$, we use Hölder's inequality, the fact that  \[\Big\|V_{\delta}(\,\cdot\,)\bE_{y,\,\cdot\,}^{t, s}\Bigl[\ee^{-\int_s^t \d r\, V_\delta(X_r)}\Bigr]\Big\|_\infty \lesssim \frac{1}{t-s}\]
and 

\begin{equation}\label{eq:Egaussianxyw}
    \int \d w\, \bigg|\frac{x-y}{t} - \frac{x-w}{s}\bigg|\frac{\psi^s(x-w)\psi^{t-s}(y-w)}{\psi^t(x-y)}=\frac{1}{s}\mathbb{E}_{\mathcal{W}}\bigl[|\mathcal{W}-\mu|\bigr]\,\lesssim \frac{\sigma}{s}\,,
\end{equation}
where $\mathcal{W}$ is a Gaussian random variable with law $\mathcal{N}(\mu, \sigma^2 I_d)$, with $\mu=x+\frac{s}{t}(y-x)$ and $\sigma^2=\frac{s(t-s)}{t}$. 
Thus, we can bound
\[\a(s) \lesssim\frac{\sqrt{t-s}}{t^{3/2}}\frac{1}{\sqrt{s}} , \quad\beta (s) \lesssim \frac{1}{t}.\]
Using \eqref{eq:firstGrownwall} and \eqref{eq:g2}, we find 
\begin{equation*}
f_\delta(t)\leq\int_0^t \d s\, \alpha(s)+\int_0^t \d s \,\beta (s) f_\delta(s)\, < \infty\,. 
\end{equation*}
We can thus use Grönwall's inequality to bound
\begin{equation} \label{f_d_estimate}
f_\delta(t)  \leq \int_0^t \d{s}\,\alpha(s)+\int_0^t \d{s}\, \Bigl(\int_0^s \d r\,\alpha(r) \Bigr)\beta(s) \exp \biggl(\int_s^t  \d r\, \beta(r)\biggr) \lesssim \frac{1}{\sqrt{t}}\,.
\end{equation}
This estimate holds for any fixed \(t>0\) uniformly in \(\delta>0\). The rest of the proof consists in controlling the limit $\delta \to 0$.

 We claim that for any $\epsilon > 0$
\begin{equation}\label{eq:Pu-vdelta}
\lim_{\delta \to 0} \mathbb{P}_{y,x}^{t,0}\biggl(\absbb{ \int_0^t \d s\,\bigl(U(X_s)- V_\delta(X_s)\bigr)}> \eps\biggr)= 0
\end{equation}
To prove \eqref{eq:Pu-vdelta}, we choose \(\eta>0\) and then \(R_\eta >0\) such that \[\P_{y,x}^{t,0}(E_{R_\eta}) \geq 1 - \eta\,, \qquad E_{R} \deq \hB{\max_{0 \leq s \leq t} \abs{X_s}\leq R}\,.\] We can thus bound
\[\begin{aligned}
& \P_{y,x}^{t,0}\biggl(\bigg|\int_0^t \mathrm{~d} s\, \bigl(U\bigl(X_s\bigr)-V_\delta\bigl(X_s\bigr)\bigr)\bigg|>\varepsilon\biggr) \\
& \qquad \leqslant \P_{y,x}^{t,0}\biggl(\int_0^t \mathrm{~d} s\,\big|U\bigl(X_s\bigr)-V_\delta\bigl(X_s\bigr)\big|>\varepsilon\biggr) \\
& \qquad \leqslant \P_{y,x}^{t,0}\biggl(\biggl\{\int_0^t \mathrm{~d} s\,\big|U\bigl(X_s\bigr)-V_\delta\bigl(X_s\bigr)\big|>\varepsilon\biggr\} \cap E_{R_\eta}\biggr)+\P_{y,x}^{t,0}\bigl(E_{R_\eta}^c\bigr) \\
& \qquad \leqslant \frac{1}{\varepsilon} \mathbb{E}_{y,x}^{t,0}\biggl[\int_0^t \mathrm{~d} s\,\big|U\bigl(X_s\bigr)-V_\delta\bigl(X_s\bigr)\big| \mathbbold{1}_{|X_s| \leqslant R_\eta}\biggr]+\eta\,,
\end{aligned}\]
where we used the Markov inequality and \(E_{R_\eta}\subset \{\abs{X_s}\leq R_\eta\}\) for any \(0\leq s\leq t \). 
We can now estimate
\begin{equation}\label{eq:Etxy_integral}
     \mathbb{E}_{y,x}^{t,0}\biggl[ \int_0^t \d s\, \big|U(X_s)- V_\delta(X_s)\big|\mathbbold{1}_{\abs{X_s}\leq R_\eta}\biggr] \leq \int_\nu^{t-\nu} \d s\, \mathbb{E}_{y,x}^{t,0}\Bigl[ \big|U(X_s)- V_\delta(X_s)\big|\mathbbold{1}_{\abs{X_s}\leq R_\eta}\Bigr] + C_{R_\eta}\nu
\end{equation}
for any \(0<\nu<t/2\) using the fact that \(U, V_\delta \) are locally bounded. To bound the first term on the right-hand side of \eqref{eq:Etxy_integral}, we use the fact that \(\nu\leq s\leq t-\nu\) and thus 

\begin{equation}\label{eq:first term of Exyt}
    \begin{aligned}
        &\int_\nu^{t-\nu} \d{s}\, \mathbb{E}^{t,0}_{y,x}\Bigl[ \big|U(X_s)- V_\delta(X_s)\big|\mathbbold{1}_{\abs{X_s}\leq R_\eta}\Bigr]\\
        &\qquad\qquad=\int_\nu ^{t-\nu} \d s\int_{B_{R_\eta}} \d{w}\,\frac{\psi^s(x-w)\psi^{t-s}(y-w)}{\psi^t(x-y)}\abs{V_\delta(w) - U(w)} \\&\qquad\qquad\leq  C_{\nu,x,y,t} \int_{B_{R_\eta}} \d{w}\,\abs{V_\delta(w)-U(w)}\,.
    \end{aligned}
\end{equation}
Since \(V_\delta\) converges to \(U\) in \(L^1_{\operatorname{loc}}\), \eqref{eq:first term of Exyt} tends to zero as $\delta \to 0$. Taking first $\delta \to 0$, then $\nu \to 0$, and then $\eta \to 0$, we therefore conclude \eqref{eq:Pu-vdelta}. 

From \eqref{eq:Pu-vdelta} and the Vitali convergence theorem, it follows that
\[\mathbb{E}_{y,x}^{t,0}\Bigl[\ee^{- \int_0^t \d s\, U(X_s)}\Bigr] = \lim_{\delta \to 0}\mathbb{E}_{y,x}^{t,0}\Bigl[\ee^{- \int_0^t \d s\, V_{\delta}(X_s)}\Bigr]\,.\]
Hence, for any \(x,y,z \in \mathbb{R}^d\),
\begin{align*}
&\mspace{-10mu} \Bigr|\mathbb{E}_{y,x}^{t,0}\Bigl[\ee^{- \int_0^t \d s\, U(X_s)}\Bigr]-\mathbb{E}_{y,z}^{t,0} \Bigl[\ee^{- \int_0^t \d s\, U(X_s)}\Bigr]  \Bigl| \\ & = \lim_{\delta \to 0}\absB{\mathbb{E}_{y,x}^{t,0}\Bigl[\ee^{- \int_0^t \d s\, V_{\delta}(X_s)}\Bigr]-\mathbb{E}_{y,z}^{t,0}\Bigl[\ee^{- \int_0^t \d s\, V_{\delta}(X_s)}\Bigr] } \\
    & = \lim_{\delta \to 0} \absbb{\int_0^1 \d{\lambda}\, {\nabla_{2} \mathbb{E}_{y,\lambda x+(1-\lambda)z}^{t,0}\Bigl[\ee^{- \int_0^t \d s\, V_{\delta}(X_s)}\Bigr]\cdot (x-z)}} \\
    &\leq  \limsup_{\delta \to 0} f_{\delta}(t) \cdot \abs{x-z} \lesssim \frac{1}{\sqrt{t}}\cdot \abs{x-z}\,,
\end{align*}
where we used \eqref{f_d_estimate}. 
Since \(x \mapsto \mathbb{E}_{y,x}^{t,0}\bigl[\ee^{- \int_0^t \d s\, U(X_s)}\bigr]\) is \(\frac{C}{\sqrt{t}}-\)Lipschitz for some \(C>0\), by Rademacher's theorem it is almost everywhere differentiable and
\[\bigg|{\nabla_x \mathbb{E}_{y,x }^{t,0}\Bigl[\ee^{- \int_0^t \d s\, U(X_s)}\Bigr]}\bigg|\lesssim \frac{1}{\sqrt{t
}}\]
for almost every \(x \in \mathbb{R}^d.\)
\end{proof}
We are now ready to prove Proposition \ref{lm:gradG}.
\begin{proof}[Proof of Proposition \ref{lm:gradG}] 
Using Lemma \ref{Feynman-Kac formula}, we write the kernel of the Green function \(G\) as
\[
G(x,y) = \int_0^\infty \d t\, \ee^{-\kappa t} \psi^t(x-y) \mathbb E_{y,x}^{t,0}\Bigl[\ee^{-\int_0^t \d s\, U(X_s)}\Bigr]\,.
\]
Using the definition of $\psi^t(x)$ as in \eqref{eq:heatkernel}, then we have
\[
\begin{split}
    \nabla_x G(x,y) & =\int_0^\infty \d t\, \ee^{-\kappa t} \frac{x-y}{t}\psi^t(x-y) \mathbb  E_{y,x}^{t,0}\Bigr[\ee^{-\int_0^t \d s\, U(X_s)}\Bigl]\\
    & + \int_0^\infty \d t\, \ee^{-\kappa t} \psi^t(x-y) \nabla_x\mathbb  E_{y,x}^{t,0}\Bigl[\ee^{-\int_0^t \d s\, U(X_s)}\Bigr] \eqd  G_1 +G_2\,.
\end{split}
\]
  
   Following the proof of Proposition \ref{prop:GN}, we now have 
\[\abs{G_1} \lesssim |x-y|^{-d+1}\ee^{-c|x-y|\sqrt{\tilde g(x) \vee \tilde g(y)}}\]
for some \(c>0\). 

We can estimate \(G_2\) similarly. Indeed, from \eqref{eq:Egaussianxyw}, Lemmas \ref{lm:gradE} and \ref{lem:boundonexp2}, we get
\begin{align*}
    \bigg|\nabla_x \mathbb  E_{y,x}^{t,0}\Bigr[& \ee^{-\int_0^t \d s\, U(X_s)}\Bigl]\bigg|\\ 
    & =\bigg|\int \d w \, \nabla_x \bigg(\frac{\psi^{t/2}(x-w)\psi^{t/2}(w-y)}{\psi^{t}(x-y)}\bigg)\bE_{y,w}^{t,t/2}\Bigl[\ee^{-\int_{t/2}^t \d s\,U(X_s)}\Bigr]\, \bE_{w,x}^{t/2,0}\Bigl[\ee^{-\int_0^{t/2} \d s \, U(X_s)}\Bigl] \bigg|\\&\quad+\bigg|\int \d w \, \frac{\psi^{t/2}(x-w)\psi^{t/2}(w-y)}{\psi^{t}(x-y)}\bE_{y,w}^{t,t/2}\Bigl[\ee^{-\int_{t/2}^t \d s\,U(X_s)}\Bigr]\, \nabla_x \bE_{w,x}^{t/2,0}\Bigl[\ee^{-\int_0^{t/2} \d s \, U(X_s)}\Bigl] \bigg| \\&\lesssim \frac{1}{\sqrt{t}} \sup_{w \in \mathbb{R}^d}\bE_{y,w}^{t,t/2}\Bigl[\ee^{-\int_{t/2}^t \d s\, U(X_s)}\Bigr]\bigg(\sup_{w \in \mathbb{R}^d}\bE_{w,x}^{t/2,0}\Bigl[\ee^{-\int_0^{t/2} \d s\, U(X_s)}\Bigr]+1\bigg) \\
    & \lesssim \frac{1}{\sqrt{t}}\ee^{-ct} \Bigl(\mathbbm{1}_{t-s\leq g(\gamma y)^{-2}|y|^{-2}}  +\ee^{-cg(\gamma y)(t-s)} + \ee^{-c|y|\sqrt{g(\gamma y)}} +\ee^{-c (t-s)^{1/3}|y|^{2/3}g(\gamma y)^{2/3}}\Bigr) \,.
\end{align*}
We can thus estimate
\[\abs{G_2} \lesssim \abs{x-y}^{-d+1}\ee^{-c|x-y|\sqrt{ \tilde g( y)}}\,. \]

By combining the estimates on \(G_1\) and \(G_2\), we find \[ \abs{\nabla_x G(x,y)}\lesssim \abs{x-y}^{-d+1}\ee^{-c|x-y|\sqrt{ \tilde g( y)}}\,. \qedhere\]

\end{proof}

\section{The counterterm problem}\label{Section:Counterterm problem}

This section is devoted to the proof of Theorem \ref{thm:Solution to the counterterm problem}. In particular, we extend the result of \cite[Section 5]{frohlich2017gibbs} by solving the counterterm problem, defined in \eqref{counterterm problem}, for a wider class of potentials.
An important requirement for the solutions of the counterterm problem constructed in this section is that they give rise to a well-defined \(\phi^4_2\)
field theory. We prove that this holds for all bare potentials 
\(\mathcal{U}\) satisfying Assumption \ref{assumption:P and D for U}, assuming the growth condition \eqref{eq:boundOnMaxg} on the lower bound $g$.

The main difference between the argument presented below and previous works lies in the use of heat kernel estimates (see Proposition \ref{prop:GN}), which allow us to treat faster-growing potentials without relying solely on an increase of the chemical potential \(\kappa\). In addition, our approach requires control of the gradient of solutions to \eqref{counterterm problem}, a property that is essential for constructing the associated field theory via the methods of Section \ref{sec:field} whenever \(\mathcal{U}\) grows radially faster than a polynomial.

The proof of Theorem \ref{thm:Solution to the counterterm problem} relies on showing that the function
\begin{equation}\label{fixed point function for counterterm problem}
    \Phi(U)\deq {\mathcal{U}}+(\tau^\eps-\tau^{\eps,{0}})  +v^\eps *(\varrho_{\nu}^{U} -\varrho_\nu^0) 
\end{equation}
admits a fixed point in an appropriate Banach space which we shall define below.
To this end, for any \(\vartheta\col\mathbb{R}^d\to(0,\infty)\), we define the norm
\[
f\in L^\infty_{\operatorname{loc}}(\mathbb{R}^d) \longmapsto \sup\limits_{x\in \mathbb{R}^d}\bigg|\frac{f(x)}{\vartheta(x)}\bigg|\eqd\|f\|_\vartheta \in (0,\infty]\,.
\]
In the following, we shall consider two different norms and their associated Banach spaces. In particular, for \(\alpha\in \{0,1\}\), we define the Banach space
\[
\mathcal{B}_\alpha\deq\{\,f \in W^{\alpha,\infty}_{\operatorname{loc}}(\mathbb{R}^d) \col \|f\|_{\mathcal{U},\alpha}<\infty\,\}\,,
\] 
with the norm \[\|f\|_{\mathcal{U},\alpha} \deq \|f\|_{\mathcal{U}}+\alpha \|\nabla f\|_{ {\tilde g}^{3/2}}\,,\]
and \(\tilde{g}\) is defined in \eqref{eq:gtilde}. This choice allows us, when \(\mathcal{U}\) is differentiable, to control appropriately the gradient of the solution of \eqref{counterterm problem}. 
Consider \[B_{\vartheta,r}({\mathcal{\vartheta}})\deq\Bigl\{f \in L^\infty_{\operatorname{loc}}(\mathbb{R}^d)\col \|f-{\mathcal{\vartheta}}\|_{\mathcal{\vartheta}} \leq r\Bigr\}\]
the ball of radius \(r>0\) around \(\vartheta\) with respect to the norm \(\| \cdot\|_\vartheta\).  We set
\[
B_{r}^0(\mathcal{U)}  \deq  B_{\mathcal{U}, r}(\mathcal{U}),
\]
\[
B_{r}^1(\mathcal{U}) \deq 
\Bigl\{f \in W^{1,\infty}_{\text{loc}}\col f \in B_{\mathcal U,r}(\mathcal{U}),  \nabla f \in B_{ {\tilde g}^{3/2},\tilde{r}}\bigl({{\tilde g}^{3/2}}\bigr)  \Bigr\}\,,
\]
where \[\tilde{r}= \|\nabla \mathcal{U}\|_{ \tilde{g}^{3/2}} + 2\,. \]

\begin{remark}\label{Remark on the radius of the ball}
    Consider \(r\in (0,1)\). If \(u \in B_{\mathcal{U},r}({\mathcal{U}})\), then
    \[
    (1-r) {\mathcal{U}}(x) \leq u(x)\leq  (1+r) {\mathcal{U}}(x) \quad \text{a.s.}
    \]
     Hence, we get that \(u\) is a positive function and 
    \[
    \tr{(h^u) ^{-2}} \leq (1-r)^{-1}\tr{(h^{\mathcal{U}}) ^{-2}} < \infty \,.
    \] 
Finally, we observe that if \(g \lesssim \mathcal{U}\), then \(g \lesssim_r u\).
\end{remark}

For \(\alpha \in \{0,1\}\), \(B^{\alpha}_r(\mathcal{U})\) are complete metric spaces, where the metric is the one induced by the norm on \(\mathcal{B}_\alpha\). The existence and uniqueness of solutions to the counterterm problem (for fixed \(\eps,\nu >0\)) follows by showing that the map \(\Phi\) satisfies, for any \(r\in (0,1)\), \(\Phi(B_{r}^\alpha({\mathcal{U}})) \subset B_{r}^\alpha({\mathcal{U}})\), and it is a contraction on \(B^{\alpha}_r(\mathcal{U})\), see \cite[Theorem 5.2]{frohlich2017gibbs}.
 
In the following arguments, \eqref{eq:boundOnMaxg} will allow us to control the behaviour of \(\mathcal{U}\) in a neighbourhood of a given point \(x \in \mathbb{R}^2\) using \eqref{cp:boundU}, as we can then estimate    
    \begin{equation}
        \label{rmk:cp_bound_U}\esssup_{\abs{x-y}\leq 1} \mathcal{U}(x) \lesssim_c g^{3/2}\bigl((\gamma+c) x\bigr)
    \end{equation} 
    for any \(c>0\).  

\begin{lemma}\label{lemma:bounds_countertermproblem_tau}
    Let \(r \in (0,1)\),  and \(u, u_1, u_2 \in B_{r}^\alpha({\mathcal{U}})\) be positive functions. Then, uniformly in \(\eps>0\) small enough and for  \(\kappa>1\),  there exists \(C>0\) such that, 
    \begin{enumerate}[label=(\roman*)]
        \item\label{tauU_minus_tau0} \[\abs{(\tau^{\eps,u}-\tau^{\eps,0})(x)  }\leq     \frac{C}{\sqrt{\kappa}}\|u\|_{\mathcal{U}} \, g(x)\, ,\]
        \item\label{tauU_minus_tauU}  \[\abs{(\tau^{\eps,u_1}-\tau^{\eps,u_2})(x)  }\leq   \frac{C}{\sqrt{\kappa} }\|u_1-u_2\|_{\mathcal{U}} \,g(x) \, ,\]
        \item \label{gradtau} if  \(\mathcal{U} \) satisfies \assD{}, \[\abs{\nabla\tau^{\eps,u}(x)  }\leq \frac{C}{{\kappa }^{1/4}}\|u\|_\mathcal{U}\, {\tilde g}^{3/2}(x) \]
        and 
        \[\abs{\nabla (\tau^{\eps,u_1} - \tau^{\eps,u_2})(x) } \leq \frac{C}{{\kappa }^{1/4}}\|u_1-u_2\|_\mathcal{U} \, {\tilde g}^{3/2}(x)\,. \]
    \end{enumerate}
    
\end{lemma}
\begin{proof}
    The strategy is analogous to \cite[Theorem 5.2]{frohlich2017gibbs}. Let \(0<r<1\) and \(u \in B_{\mathcal{U},r}({\mathcal{U}})\).
    By applying a resolvent expansion and using the bound \(u(y) \le \|u\|_{\mathcal{U}} {\mathcal{U}}(y)\), we find
    \begin{align}
        \abs{(\tau^{\eps,u}- \tau^{\eps,0})(x)}
        & = \int \d {\tilde x} \,v^{\eps}(x-\tilde{x})\int \d y\, \frac{1}{h^u}(x,y)u(y)\frac{1}{h^0}(y,\tilde{x}) \nonumber \\
        &\leq \|u\|_{\mathcal{U}}\int \d {\tilde x} \,v^{\eps}(x-\tilde{x})\int \d y\, \frac{1}{h^u}(x,y){\mathcal{U}}(y)\frac{1}{h^0}(y,\tilde{x})\label{eq:taueps diff}\\
        & = \|u\|_{\mathcal{U}}\int \d {\tilde x}\, v^{\eps}(x-\tilde{x})\int_{\abs{x-y}\leq 1} \d y \,\frac{1}{h^u}(x,y){\mathcal{U}}(y)\frac{1}{h^0}(y,\tilde{x}) \label{bounded domain tauesp}\\
        & \quad + \|u\|_{\mathcal{U}}\int \d {\tilde x}\, v^{\eps}(x-\tilde{x})\int_{\abs{x-y} > 1} \d y \,\frac{1}{h^u}(x,y){\mathcal{U}}(y)\frac{1}{h^0}(y,\tilde{x})\label{unbounded domain tauesp}\,.
    \end{align}
    Let us first bound \eqref{bounded domain tauesp}. We can estimate
    \begin{align}
        \eqref{bounded domain tauesp}&\leq C\|u\|_{\mathcal{U}} \esssup_{\abs{z-x} \leq 1}{\mathcal{U}}(z)\int \d {\tilde x}\, v^{\eps}(x-\tilde{x})\int_{\abs{x-y}\leq 1} \d y \,\frac{1}{h^u}(x,y)\frac{1}{h^0}(y,\tilde{x}) \nonumber\\
        \begin{split}
            &\leq C\|u\|_{\mathcal{U}}  \esssup_{\abs{z-x} \leq 1}{\mathcal{U}}(z)\\&\qquad\qquad\times\int \d {\tilde x}\, v^{\eps}(x-\tilde{x})\int_{\abs{x-y}\leq 1}  \d y\, \Bigl(\frac{1}{h^u}(x,y)\frac{1}{h^0}(y,x)+\frac{1}{h^u}(\tilde{x},y)\frac{1}{h^0}(y,\tilde x)  \Bigr)\nonumber
        \end{split} \\
        &\leq  C_{\tilde \gamma} \|u\|_{\mathcal{U}} \sup_{\abs{\tilde x-x}<\eps}\frac{\log(g(x))}{(\kappa+g(\tilde{\gamma} \tilde x))}  \esssup_{\abs{z-x} \leq 1}{\mathcal{U}}(z)\,, 
         \label{Bound of first part taueps0}
    \end{align}
    for any fixed \(0<\tilde \gamma<1\), where we used \(\widehat{v^\eps}(0) = 1\) and Proposition \ref{prop:GN} in the last inequality. From Assumption \ref{assumption:mathcalU_cp} and \eqref{rmk:cp_bound_U}, by choosing \(\tilde \gamma = \sqrt{\gamma + c} \), and using the fact that, since \(g\) is radially non-decreasing, \[g\bigl((1-\varepsilon)x\bigr)\lesssim\inf\limits_{\abs{x-\tilde{x}}<\eps} g(\tilde{x})\,,\]  we find 
    \begin{align}
        \eqref{Bound of first part taueps0} \leq   C \|u\|_{\mathcal{U}} \frac{g\bigl((\gamma+c)x \bigr)^{3/2}}{\bigl(\kappa + g(\tilde{\gamma}^2 x )\bigr)}  \leq \frac{C}{\sqrt{\kappa}}\|u\|_{\mathcal{U} } g\bigl((\gamma+c)x\bigr)\,, \label{Bound of first part taueps}
    \end{align}
    for all \(0<c<1-\gamma\) and \(\eps>0\) small enough. 
    We now consider \eqref{unbounded domain tauesp}. Using Proposition \ref{prop:GN} with the assumption that \(\kappa > 1\), we find  
    \begin{align}
        \eqref{unbounded domain tauesp} &\leq C\|u\|_{\mathcal{U}}\int \d {\tilde x}\, v^{\eps}(x-\tilde{x})\int_{|x-y|> 1} \d y\, \frac{1}{h^u}(x,y){\mathcal{U}}(y)\frac{1}{h^0}(y,\tilde{x}) \nonumber\\
        \begin{split}
        {}&\leq C\|u\|_{\mathcal{U}} \int \d {\tilde x}\, v^{\eps}(x-\tilde{x})\int_{|x-y|> 1} \d y\, \,\ee^{-c\sqrt{\kappa + g(\tilde \gamma x)\vee g( \tilde \gamma y)}|x-y|}{\mathcal{U}}(y)\ee^{-\sqrt{\kappa}\abs{\tilde x-y}} \nonumber 
         \end{split}   
         \\ &\leq  \frac{C}{\sqrt{\kappa}}\|u\|_\mathcal{U}  \label{Bound of second part taueps}\,,
    \end{align}
    where the last step follows from \eqref{cp:boundU}. 
    From Assumption \ref{assumption:mathcalU_cp}, \eqref{Bound of first part taueps} and \eqref{Bound of second part taueps}, we have
    \begin{equation*}
        \eqref{eq:taueps diff} \leq \frac{C}{\sqrt{\kappa}} \|u\|_{\mathcal{U}} g((\gamma+c)x)\,. 
    \end{equation*}
    The bound in \ref{tauU_minus_tau0} is then a consequence of the fact that \(0<\gamma<1 \) and thus \(g\bigl((\gamma + c)x\bigr) \leq g(x)\) for any \(c>0\) small enough. Starting with a resolvent expansion and following the same steps as above, \ref{tauU_minus_tauU} follows.

    Let us now consider \ref{gradtau}.  Writing \(\nabla_x \tau^{\eps,u} = \nabla_x ({\tau^{\eps,u}-\tau^{\eps,0}})\), using a resolvent expansion and integration by parts, we find 
    \begin{multline}
        \abs{\nabla_x \tau^{\eps,u}(x)} \leq \abs{\int \d {\tilde x}\, v^\eps(x- \tilde x)\int \d y \,  \nabla_x\frac{1}{h^u}(x,y) U(y)\frac{1}{h^0}(y, \tilde{x})}  \\ + \abs{\int \d {\tilde x}\, v^\eps(x- \tilde x)\int \d y \, \frac{1}{h^u}(x,y) U(y) \nabla_{\tilde x}\frac{1}{h^0}(y, \tilde{x})}\label{eq:tauepsgrad} \, .
     \end{multline}
     We show the bound of the first term \eqref{eq:tauepsgrad}, as the second one can be estimated in a similar way.

     Using Proposition \ref{lm:gradG}, and splitting cases as in step (i), by using
     \begin{align*}
         \abs{\nabla_x \frac{1}{h^u}(x,y)} \leq  \frac{C\ee^{-c\abs{x-y}\sqrt{\kappa + g(\tilde \gamma y) }}}{\abs{x-y}} \, , 
     \end{align*}
     for any \(\gamma<\tilde \gamma<1\), we find
     \[\abs{\nabla_x \tau^{\eps,u}(x)} \leq \frac{C\log(g(\tilde \gamma x))}{\sqrt{\kappa + g(\tilde \gamma ^2 x)}} g\bigl((\gamma+c)x \bigr)^{3/2}\|u\|_\mathcal{U} \leq \frac{C}{\kappa^{1/4}}\|u\|_\mathcal{U} \tilde{g}^{3/2}(x)\,, \]
     for \(c>0\) small enough. 
    Notice that, to show the last inequality, we used \eqref{eq:boundOnMaxg}.
     We can show similarly that
     \[\abs{\nabla (\tau^{\eps,u_1} -  \tau^{\eps,u_2})(x) } \leq \frac{C}{\kappa^{1/4}}\|u_1-u_2\|_\mathcal{U} \,{\tilde g}^{3/2}(x)\,.\qedhere\]
\end{proof}

For $\nu>0$, we define the quantum Green function
\begin{equation}\label{eq:quantumGreen}\mathcal{G}_\nu \deq \frac{\nu}{\ee^{\nu h}-1}=\nu \sum_{n=1}^\infty \ee^{-\nu n h}.\end{equation}
Note that the mean particle density \eqref{def quantum green function diagonal} satisfies 
\begin{equation}\label{eq:meanparticlequantumgreen}
\varrho_\nu (x)=\mathcal{G}_\nu(x,x). 
\end{equation}
In order to show the corresponding result of Lemma \ref{lemma:bounds_countertermproblem_tau} for the difference of mean particle densities, we establish the following bounds for the quantum Green function, equivalent to Proposition \ref{prop:GN}. 
\begin{prop}\label{prop:bounds_Gnu}
Let \(d\geq 2\) and \(\mathcal{G}_\nu\) be the quantum Green function as in \eqref{eq:quantumGreen}.  
\begin{enumerate}[label = (\roman*)]
    \item \label{prop:bound_Gnu} Under the same assumptions as Proposition \ref{prop:GN}, we have the following bound on its kernel
 \[\mathcal{G}_\nu (x,y) \lesssim \Bigl[\Bigl( \bigl[\log \bigl(|x-y|\sqrt{\tilde g(x)\vee \tilde g(y)}\bigr)^{-1}\bigr]_+ \wedge \log \nu^{-1} \Bigr)\vee 1\Bigr] \ee^{-c|x-y|\sqrt{\tilde g(x)\vee \tilde g(y)}}\]
 in dimension \(d=2\), and 
\[\mathcal{G}_\nu (x,y) \lesssim \frac{1}{\nu ^{d-2} \vee \abs{x-y}^{d-2}} \ee^{- c\sqrt{\tilde g(x) \vee \tilde g(y)}\abs{x-y}}\]
when \(d\geq 3\). 
\item \label{prop:boundgradGnu} Furthermore, under the same assumptions as Proposition \ref{lm:gradG}, we bound the gradient of the kernel of the quantum Green function by 
\[\abs{\nabla_x \mathcal{G}_\nu(x,y)} \lesssim \frac{1}{\nu ^{d+1} \vee \abs{x-y}^{d+1}} \ee^{- c\sqrt{\tilde g(y)}\abs{x-y}}\,.\]
\end{enumerate}
\end{prop}
\begin{proof}
    Point \ref{prop:bound_Gnu} can be shown using the heat kernel estimates in Lemma \ref{lem:boundonexp2},
    \begin{multline*}\label{eq:G1}
			\mathcal{G}_\nu(x,y) \lesssim 
			\nu\sum_{n=1}^\infty \, \frac{1}{(n\nu)^{d/2}} \ee^{-\frac{|x-y|^2}{2n\nu}} \ee^{-c \frac{n\nu }2} \Bigl(\mathbbm{1}_{n \nu\leq (\tilde{g}(x)\vee \tilde{g}(y))^{-2}(|x| +|\tilde x|)^{-2}} + \ee^{-c(\tilde{g}(x)\vee \tilde{g}(y))n \nu} \\\hspace{3cm}+ \ee^{-c(|x|+|\tilde x|)\sqrt{\tilde{g}(x)\vee \tilde{g}(y)}} +\ee^{-c {(n \nu)}^{1/3}(|x|+|\tilde x|)^{2/3}(\tilde{g}(x)\vee \tilde{g}(y))^{2/3}}\Bigr)\,,
	\end{multline*}
    and noticing that every term of the sum is of the form 
    \(\nu\sum_{n=1}^\infty  f(\nu n)\),
    where the function \(f\) has a global maximum at \(x_0 \in (0,\infty)\), is increasing in \([0, x_0]\), and decreasing in \([x_0,\infty)\). We can thus bound 
    \[\nu\sum_{n=1}^{\lfloor x_0/\nu\rfloor}  f(\nu n) \lesssim \int_\nu^{x_0}\d{t}\,f(t)\]
    and 
    \[\nu\sum_{{\lceil x_0/\nu\rceil}}^\infty f(\nu n) \lesssim \nu f(x_0)+ \int_{\lceil x_0/\nu\rceil\nu}^{\infty}\d{t}\,f(t)\,.\]
    By proceeding as in the proof of Proposition \ref{prop:GN} with the change of variables \(\nu=1/N\) and considering the cases \(\nu\leq x_0\) and \(\nu \geq x_0\) separately, we conclude the proof.

    The proof of \ref{prop:boundgradGnu} is done in a similar way. We start by noting that, following the proof of Proposition \ref{lm:gradG}, 
    \begin{align*}
           \nabla_x (\ee^{-t h})_ {x,y} & = \ \ee^{-\kappa t} \frac{x-y}{t}\psi^t(x-y) \mathbb E_{x,y}^t\Bigr[\ee^{-\int_0^t \d s\, U(X_s)}\Bigl]\\
     &\hspace{3cm}+ \ee^{-\kappa t} \psi^t(x-y) \nabla_x\mathbb E_{x,y}^t\Bigl[\ee^{-\int_0^t \d s\, U(X_s)}\Bigr]\, \nonumber\\
     & \lesssim \ee^{-\kappa t} \frac{x-y}{t}\psi^t(x-y) \mathbb E_{x,y}^t\Bigr[\ee^{-\int_0^t \d s\, U(X_s)}\Bigl] \\
     &\hspace{3cm}+   \ee^{-\kappa t} \psi^t(x-y) \frac{1}{\sqrt{t}} \sup_{w \in \mathbb{R}^d}\bE_{w,y}^{t/2}\Bigl[\ee^{-\int_0^t \d s\, U(X_s)}\Bigr]\,. \label{eq: grad e} 
    \end{align*}
    We can now conclude using heat kernel estimates, and arguing that the Riemann sums are bounded as above by the corresponding integral. 
\end{proof}

\begin{remark}\label{rmk:qGf bound}
    From a similar argument to Proposition \ref{prop:bounds_Gnu}, it follows that, uniformly in \(0\leq t<1 \),
    \[\Bigl(\frac{\nu\ee^{t\nu h}}{\ee^{\nu h}-1} \Bigr)(x,y) \]
    is bounded by the left-hand side of \eqref{eq:boundG}, and a similar result holds for its gradient. 
\end{remark}

\begin{lemma}\label{lemma:bounds_countertermproblem_varrho}
    Let \(r \in (0,1)\),  and \(u, u_1, u_2 \in B_{r}^\alpha({\mathcal{U}})\) be positive functions. Then, uniformly for \(\eps>0\) small enough and for  \(\kappa>1\),  there exists \(C>0\) such that, 
    \begin{enumerate}[label=(\roman*)]
        \item \[ \big|v^\eps *(\varrho_{\nu}^{u} -\varrho_\nu^0)(x)\big| \leq \frac{C}{\sqrt{\kappa}}\|u\|_{\mathcal{U}}\, g(x)  \]
        \item  \begin{equation*}
            \abs{v^\eps *\bigl(\varrho_{\nu}^{u_1} -\varrho_{\nu}^{u_2}\bigr)(x)} \leq \frac{C}{\sqrt{\kappa}}\|u_1-u_2\|_{\mathcal{U}} \,g(x)  
        \end{equation*} 
        \item if \(\mathcal{U} \) satisfies \assD{} then \[\big|\nabla \bigl( v^\eps *\varrho_{\nu}^{u} \bigr)(x)  \big|\leq \frac{C}{\kappa^{1/4}}\|u\|_\mathcal{U}\, {\tilde g}^{3/2}(x) \]
        and 
        \[\big|\nabla \bigl(v^\eps *\bigl(\varrho_{\nu}^{u_1} -\varrho_{\nu}^{u_2}\bigr)\bigr)(x) \big| \leq \frac{C}{\kappa^{1/4}}\|u_1-u_2\|_\mathcal{U} \, {\tilde g}^{3/2}(x)\,. \]
    \end{enumerate}
\end{lemma}
\begin{proof}
The proof strategy is the same as Lemma \ref{lemma:bounds_countertermproblem_tau}.
Recalling \eqref{eq:meanparticlequantumgreen} and applying a resolvent expansion, we get
 \begin{align}
 \nonumber
 -(\varrho_\nu^u-\varrho_\nu^0)(y)& =\mathcal{G}_\nu^0(y,y)-\mathcal{G}_\nu^u(y,y)\\
\label{eq:quantumgreenfunction_resolventexpansion}
 & =\int_0^1\d t\,\int \d z\, \Bigl(\frac{\nu \ee^{t\nu h^u}}{\ee^{\nu h^u}-1}\Bigr)(y,z)u(z)\Bigl(\frac{\nu \ee^{(1-t)\nu h^0}}{\ee^{\nu h^0}-1}\Bigr)(z,y)\,.
 \end{align}

To bound \eqref{eq:quantumgreenfunction_resolventexpansion}, we follow the same steps as in the proof of \cite[Theorem~5.2]{frohlich2017gibbs}, {\em cf.}\ (5.16), obtaining

\begin{align}
\label{eq:eval1}
\eqref{eq:quantumgreenfunction_resolventexpansion}\leq \, &C \|u\|_\mathcal{U}\esssup_{|z-y|<1}\mathcal{U}(z) \int_0^1\d t\,\int_{|z-y|<1} \d z\, \Bigl(\frac{\nu \ee^{t\nu h^u}}{\ee^{\nu h^u}-1}\Bigr)(y,z)\Bigl(\frac{\nu \ee^{(1-t)\nu h^0}}{\ee^{\nu h^0}-1}\Bigr)(z,y)\\
\label{eq:eval2}
&+C \|u\|_\mathcal{U}\int_0^1\d t\,\int_{|z-y|\geq 1} \d z\, \Bigl(\frac{\nu \ee^{t\nu h^u}}{\ee^{\nu h^u}-1}\Bigr)(y,z)\mathcal{U}(z)\Bigl(\frac{\nu \ee^{(1-t)\nu h^0}}{\ee^{\nu h^0}-1}\Bigr)(z,y)\Bigr)\,.
\end{align}
We can now refer to Lemma \ref{lemma:bounds_countertermproblem_tau} for the bound of \( 
\big|v^\eps *(\varrho_{\nu}^{u} -\varrho_\nu^0)(x)\big|
\). In particular, by Remark \ref{rmk:qGf bound}, the contribution from  \eqref{eq:eval1}, resp. \eqref{eq:eval2} can be estimated as in \eqref{bounded domain tauesp}, resp. \eqref{unbounded domain tauesp}. This proves (i). A similar argument, together with an application of Proposition \ref{prop:bounds_Gnu}, yields (ii) and (iii).
\end{proof}

\begin{prop}\label{solution to the counterterm problem}
    Let \(r \in (0,1)\), \({\mathcal{U}}\) satisfying Assumption \ref{assumption:mathcalU_cp} and let \(\alpha=0\). Then, there exists a positive chemical potential \(\kappa_0(r)>0 \), such that the function \(\Phi\) defined in \eqref{fixed point function for counterterm problem}
    satisfies
    \[\Phi(B_{r}^\alpha({\mathcal{U}})) \subset B_{r}^\alpha({\mathcal{U}})\]
    and is a contraction, for all \(\kappa > \kappa_0(r) \).

    The same result holds for \(\alpha=1\) if \(\mathcal{U}\) satisfies \assD{}. 
\end{prop}
\begin{proof}
    Let \(u \in B_{r}^\alpha({\mathcal{U}})\) and \(\alpha=0\).  From parts (i) of Lemmas \ref{lemma:bounds_countertermproblem_tau} and \ref{lemma:bounds_countertermproblem_varrho}, we have, for all \(\kappa>1\), 
    \begin{align*}
        \abs{(\tau^{\eps,u}-\tau^{\eps,{0}})(x) -v^\eps*(\varrho_{\nu}^{u} -\varrho_\nu^0) (x) } &\leq C \kappa^{-1/2}\|u\|_{\mathcal{U}} {g}(x)\,\leq C \kappa^{-1/2}\|u\|_{\mathcal{U}}\, {\mathcal{U}}(x)\,. 
    \end{align*}
    We thus find, using the triangle inequality and \(\|{\mathcal{U}}\|_{{\mathcal{U}}}=1\)
    \begin{equation}\label{eq1}
        \|\Phi(u)- {\mathcal{U}}\|_{\mathcal{U}} \leq C\kappa^{-1/2}\|u-{\mathcal{U}}\|_{\mathcal{U}} + C \kappa^{-1/2}  < r
    \end{equation}
    where the second inequality is verified for \(\kappa\) big enough. 
    From parts (ii) of Lemmas \ref{lemma:bounds_countertermproblem_tau} and \ref{lemma:bounds_countertermproblem_varrho}, we find similarly  
    \begin{equation}
        \|\Phi(u_1)- \Phi(u_2)\|_{\mathcal{U}} \leq  C \kappa^{-1/2}\|u_1-u_2\|_{\mathcal{U}} < \frac{1}{4}\|u_1-u_2\|_{\mathcal{U}}  \label{eq2}
    \end{equation}
    for \(\kappa \) big enough. 
    We can now choose \(\kappa_0(r)>0\) big enough such that for all \(\kappa >\kappa_0(r)\), \eqref{eq1} and \eqref{eq2} hold. 

    If \(\mathcal{U}\) satisfies \assD{}, we can further show that, for \(\kappa >0\) big enough, 
    \[\big\|\nabla \Phi(u)- {{\tilde g}^{3/2}}\big\|_{{\tilde g}^{3/2}} \leq \|\nabla \mathcal{U}\|_{ \tilde{g}^{3/2}} + 1 + \frac{C}{\kappa^{1/4}}\|u\|_{\mathcal{U}}\leq  \|\nabla \mathcal{U}\|_{ \tilde{g}^{3/2}} + 2\,.\]
    In particular, \(\Phi(u)\in B_{r}^1({\mathcal{U}})\). Moreover, parts (iii) of Lemmas \ref{lemma:bounds_countertermproblem_tau} and \ref{lemma:bounds_countertermproblem_varrho} yield
    \[  \big\|\nabla\Phi(u_1)- \nabla \Phi(u_2)\big\|_{{\tilde g}^{3/2}} <  \frac{1}{4}\| u_1-u_2\|_{\mathcal{U}}\,.  \] We then bound 
    \[\|\Phi(u_1)-\Phi(u_2)\|_{\mathcal{U},1}< \frac{1}{2} \| u_1-u_2\|_{\mathcal{U},0} <  \frac{1}{2}\| u_1-u_2\|_{\mathcal{U},1} \,.
    \qedhere
    \] 
\end{proof}

An application of Banach fixed point theorem shows point (i) of Theorem \ref{thm:Solution to the counterterm problem}. 

\begin{cor}\label{cor:solution_counterterm_fixed_eps_nu}
     In the hypotheses of Proposition \ref{solution to the counterterm problem}, for any fixed \(\eps, \nu>0\), there exists a unique positive solution \(u^\eps_\nu \in B^\alpha_r(\mathcal{U})\) to \eqref{counterterm problem}. 
 \end{cor}

In the next two lemmas, we show that the family of solutions \(\{u_{\nu} ^\eps\}_{\eps, \nu > 0}\) has a limit for \(\eps , \nu \to 0\) and \(\eps \) satisfies \eqref{eq:rate of convergence of eps counterterm problem}. Notice that, by completeness, the limit is in \(B^\alpha_r(\mathcal{U})\).

\begin{lemma}\label{eq:lemma:Cauchysequence counterterm problem solution for fixed nu}
    Let \(r \in (0,1)\) and consider, for every \(\eps, \nu>0\), the solution \(u_{\nu}^{\eps} \in B_{r}^\alpha({\mathcal{U}})\) to  \eqref{counterterm problem}, as per Corollary \ref{cor:solution_counterterm_fixed_eps_nu}. Then, uniformly in \(\nu>0\), and for all \(\kappa >0\) big enough, the family \(\{u_{\nu}^{\eps}\}_{\eps>0} \) is Cauchy in \((B_{r}^\alpha({\mathcal{U}}), \|\, \cdot \,\|_{\mathcal{U},\alpha})\) for \(\eps \to 0\).  
\end{lemma}
\begin{proof} Let us first consider the case \(\alpha =0\). 
    Let \(\nu, \eps', \eps >0\), \(r \in (0,1)\) and \(u^\eps,u^{\eps'} \in B_{r}^\alpha({\mathcal{U}})\) solutions to \eqref{counterterm problem} for fixed inverse temperature \(\nu\) and interaction range of respectively of \(\eps, \eps'\).  We have
    \begin{equation}\label{diff_btw_us_eps}
        u^{\eps}(x) - u^{\eps'}(x) = (\tau^{\eps,u^\eps}(x)- \tau^{\eps,{0}})-(\tau^{\eps',u^{\eps'}}(x)-\tau^{\eps'}_0)  +v^\eps *(\varrho_{\nu}^{u^{\eps}}(x)-\varrho_\nu^0 ) -v^{\eps'} *(\varrho_{\nu}^{u^{\eps'}}(x) - \varrho_{\nu}^0).
    \end{equation}
     We can show the result by treating the terms in pairs and bounding them by the difference between \(u^\eps\) and \(u^{\eps'}\) and a vanishing \(\eps\)-dependent part. We start by bounding 
     \begin{multline}
     \big|(\tau^{\eps,u^\eps}(x)- \tau^{\eps,{0}})-(\tau^{\eps',u^{\eps'}}(x)-\tau^{\eps',0}) \big| \\ \leq \big|\tau^{\eps,u^\eps}(x)- \tau^{\eps,u^{\eps'}}(x)\big| + \big|\tau^{\eps,u^{\eps'}}(x)- \tau^{\eps,0}-(\tau^{\eps',u^{\eps'}}(x)- \tau^{\eps',0} )\big|\eqd I_1 + I_2 . 
     \end{multline}
     The term \(I_1\) is bounded using Lemma \ref{lemma:bounds_countertermproblem_tau} \ref{tauU_minus_tauU} 
     \[I_1 \leq \frac{C}{\sqrt{\kappa}}\|u^{\eps} - u^{\eps'}\|_{\mathcal{U}}\]
     and \(I_2\) vanishes with \(\eps, \eps' \to 0\) by dominated convergence, since
     \begin{align*}
        I_2 & = \bigg|\int \d y\, \bigl(v^{\eps}(x-y)-v^{\eps'}(x-y) \bigr) (G^{u^{\eps'}}(x,y)-G^0(x,y)) \bigg|\\
        &\leq  \frac{C}{\sqrt{\kappa}}\|u^{\eps'}\|_{\mathcal{U}} {\mathcal{U}}(x)\leq \frac{C}{\sqrt{\kappa}} (1+r) {\mathcal{U}}(x) 
     \end{align*}
    uniformly in \(\eps, \eps'>0\) as stated in Lemma \ref{lemma:bounds_countertermproblem_tau} \ref{tauU_minus_tau0}, and \(v^\eps \rightarrow \delta\) as $\eps \to 0$. 

    The two other pairs of terms can be treated in the same way, since they are bounded uniformly in \(\nu>0\) for any fixed \(\eps>0\).

    We can argue similarly if \(\alpha=1\), whenever \(\mathcal{U}\) satisfies \assD{}, by bounding in a similar way the quantity 
    \begin{equation*}
        \nabla u^{\eps}(x) - \nabla u^{\eps'}(x) = \nabla \tau^{\eps,u^\eps}(x)-{\nabla \tau^{\eps',u^{\eps'}}} (x) +\nabla \bigl( v^\eps *\varrho_{\nu}^{u^{\eps}}\bigr)(x) -\nabla \bigl(v^{\eps'} *\varrho_{\nu}^{u^{\eps'}} \bigr)(x)\,.
        \qedhere
    \end{equation*}
\end{proof}

\begin{lemma}\label{eq:lemma:Cauchysequence counterterm problem solution}
    Let \(r \in (0,1)\), \(\alpha\in\{0,1\}\), \(u_{\nu}^{\eps} \in B_{\alpha, r}({\mathcal{U}})\) the solution to  \eqref{counterterm problem}, for fixed \(\eps, \nu>0\) and \(\theta\) as in Remark \ref{rmk:thetha_s2}. For \(\eps>0\) satisfying \eqref{eq:rate of convergence of eps counterterm problem}, the family \(\{u_{\nu}^{\eps}\}_{\nu,\eps>0} \) is Cauchy in \((B_{r}^\alpha({\mathcal{U}}), \|\, \cdot \,\|_{\mathcal{U},\alpha})\) for \(\nu,\eps\to 0\). As a consequence, there exists \(U \in B_{r}^\alpha({\mathcal{U}}) \) such that
    \[\lim_{\nu, \eps \to 0} u^{\eps}_{ \nu} = U\,.\]
\end{lemma}
\begin{proof} Let us first consider the case where \(\alpha=0\). 
    Using a triangle equality and \eqref{eq:lemma:Cauchysequence counterterm problem solution for fixed nu}, it is enough to bound appropriately the difference 
    \begin{equation}\label{diff_btw_us}
        u_{\nu}(x) - u_{\nu'}(x) = (\tau^{\eps,u_\nu}(x)-\tau^{\eps,u_{\nu'}}(x))  +v^\eps *\bigl[(\varrho_{\nu}^{u_{\nu}}(x)-\varrho_\nu^0 ) -(\varrho_{{\nu'}}^{u_{\nu'}}(x) - \varrho_{\nu'}^0)\bigr]
    \end{equation}
    where  \(\nu, \nu', \eps >0\), \(0<r<1\) and \(u_\nu,u_{\nu'} \in B_{\mathcal{U},r}({\mathcal{U}})\) solutions to \eqref{counterterm problem} for fixed \(\eps\) and inverse temperature of respectively \(\nu, \nu'\).
    Using Lemma \ref{lemma:bounds_countertermproblem_tau} (ii),  Lemma \ref{lemma:bounds_countertermproblem_varrho} (ii) and \eqref{diff_btw_us}, we find 
    \begin{equation*}
        (1-C\kappa^{-1/2})\|u_{\nu}-u_{\nu'}\|_{\mathcal{U}} \leq \|v^\eps * A\|_{\mathcal{U}}
    \end{equation*}
    where
    \begin{multline}
        A(x)\deq\biggl[ \frac{\nu}{\ee^{\nu (-\Delta/2 + \kappa + u_{\nu})}-1}- \frac{\nu}{\ee^{\nu(-\Delta/2 + \kappa )}-1}\biggr](x,x)\\-\biggl[\frac{\nu'}{\ee^{\nu'(-\Delta/2 + \kappa + u_{\nu})}-1}- \frac{\nu'}{\ee^{\nu'(-\Delta/2+\kappa) }-1}\biggr](x,x)\,.
    \end{multline}
    Using the arguments of \cite[Proof of Theorem 5.2, pp.\ 957-959]{frohlich2017gibbs}, we can show
     \[ \|v^\eps * A\|_{\mathcal{U}} \lesssim_{v}\frac{(\nu\vee \nu')^{1-s/2}}{\eps^2}\] by noting the following quantitative estimate for \(u_\nu \in B_r(\mathcal{U})\)
     \begin{equation}\nonumber
         \Big\| \frac{ \nu\ee^{\nu t h_\nu }}{\ee^{\nu h_\nu}-1}-\frac{1}{h_\nu}\Big\|^2_{\mathfrak{S}^2}\lesssim\nu^{2-s}
     \end{equation}
     and 
     \begin{equation}\nonumber
         \biggl[ \frac{\nu\ee^{\nu t(-\Delta+\kappa) }}{\ee^{\nu(-\Delta+\kappa) }-1}-\frac{1}{(-\Delta+\kappa)}\biggr]^2(0 , 0)\lesssim\nu
     \end{equation}
     uniform in \(t\in [0,1)\), which we can obtain using the methods described in Section \ref{subsec:Riemann sum estimates}. 

     The case \(\alpha=1\) follows in a similar way, with the estimate 
     \[\Biggl(1-\frac{C}{\sqrt\kappa}\Biggr)
     \|u_\nu-u_{\nu'}\|_{\mathcal{U},1} \leq  \|v^\eps * A\|_{\mathcal{U}} +  \|(\nabla v^\eps) * A\|_{\tilde{g}^{3/2}}\lesssim_{v}\frac{(\nu'\vee\nu)^{1-s/2}}{\eps^{3}}\,.
     \qedhere
     \]
\end{proof}

\appendix

\section{Non-solvability of \eqref{eq:eqAmirali} (with Amirali Hannani)}
\label{appendix:counter example}

In this section, we discuss \eqref{eq:eqAmirali}. As stated previously, finding a bounded, integrable solution \(\alpha^\eps\) to  \eqref{eq:eqAmirali} is enough to show the convergence of the renormalised many-body system to the \(\phi^4_2\) field theory, using the methods described in \cite{FKSS_24}. We prove in this section that, for any potential \(v\) chosen as in Assumption \ref{assumptions: non local interaction}, such a solution cannot exist, and thus the methods presented in Section \ref{sec:quantumtoclassic} are a necessary step to our proof. 

The argument relies on the fact that the left hand-side of \eqref{eq:eqAmirali} is necessarily as regular as \(v\), since it is a convolution, whilst the regularity of the right-hand side, \(\tau^\eps\), is a function of the smoothness of the potential \(U\) considered. 

Let us now show that, for any positive potential \(v\) satisfying Assumption \ref{assumptions: non local interaction}, there is no solution \(\alpha^\eps\) to \eqref{eq:eqAmirali}. Notice that the same proof can be extended to general potentials \(v\), by choosing an external potential \(U\) which is discontinuous at some point \(x\) where \(\tau^\eps(x)\neq 0\).

\begin{prop} Let \(v\) be a positive interaction potential satisfying  Assumption \ref{assumptions: non local interaction}, and let \(U\) be a positive locally bounded external potential that is not continuous. Then  
    \[u(x) \deq \int \d y \,G^U(x,y)v^\eps(x-y)\]  
is not in $C^2$.
\end{prop}
\begin{proof}
     Notice that
    \[x \mapsto\int \d{y}\,  G(x,y) g(x-y)\]
    and 
    \[x \mapsto\int \d{y}\,  \frac{\partial}{\partial x_i}G(x,y) g(x-y)\]
    are continuous for any smooth, bounded function \(g\) by Remark \ref{rmk:gradG_continuous} and dominated convergence. Indeed, for any \(y\in \mathbb{R}^2\), the maps \(x\mapsto G(x,y)\) and \(x\mapsto  \frac{\partial}{\partial x_i}G(x,y)\) are continuous on the set \(\mathbb{R}^2\setminus\{y\}\). 

    Let us now assume by contradiction that \[u(x) =\int \d x \,G(x,y)v^\eps(x-y) \]
    is \(C^2\) for a given positive \(v\) satisfying Assumption \ref{assumptions: non local interaction}.
    
    In the following we will write \(\ph\) to mean a continuous function, that can change from one line to the next. We can now compute
    \begin{align*}
        (hu)(x) & = \int \d y \,(h G) (x,y)v^\eps(x-y)  - \int \d y\, \nabla_xG(x,y) \cdot \nabla_x v^\eps(x-y) \\
        &\quad- \frac{1}{2} \int \d y\, G(x-y) \Delta_x v^\eps (x-y)\,\\
        & = v^\eps(0) + \ph(x) = \ph(x)
    \end{align*}
    where we used the fact that \((hG)(x,y) = \delta(x-y) \). 
    We have thus showed that \(x \mapsto (hu)(x)\) is a continuous function. We have
    \begin{multline*}
        -(\Delta u)(x_1,x_2) + (\Delta u)(\tilde x_1,\tilde x_2) = 2U(\tilde x_1,\tilde x_2)u(\tilde x_1,\tilde x_2)-2U(x_1,x_2)u(x_1,x_2)\\+ \ph(x_1,x_2)-\ph(\tilde x_1,\tilde x_2)\,.
    \end{multline*}
    Let \((a,b)\in \mathbb{R}^2\) be a point of discontinuity of \(U\). Then
    \[ \limsup_{x_1 \to a}\abs{ -(\Delta u)(x_1, b) + (\Delta u)(a, b)  }=  2 u(a,b) {\Bigl(\limsup_{x_1\to a}\abs{U(x_1,b)-U(a,b )}\Bigr)}\neq  0\,\]
    where we use the fact that \(u\) is a strictly positive function.
    This contradicts the assumption that \(u \in C^2\). 
\end{proof}

\bibliography{bibliography}

\vspace{1em}
\noindent
Cristina Caraci (\href{mailto:Cristina.Caraci@unige.ch}{Cristina.Caraci@unige.ch})
\\
Antti Knowles (\href{mailto:Antti.Knowles@unige.ch}{Antti.Knowles@unige.ch})
\\
Alessio Ranallo (\href{mailto:Alessio.Ranallo@unige.ch}{Alessio.Ranallo@unige.ch})
\\
Pedro Torres Giesteira (\href{mailto:Pedro.TorresGiesteira@unige.ch}{Pedro.TorresGiesteira@unige.ch})
\\
University of Geneva, Section of Mathematics.

\end{document}